\documentclass{tlp}

\usepackage{amsfonts}
\usepackage{amssymb}
\usepackage{amsmath}
\usepackage{stmaryrd}
\usepackage{proof}
\usepackage{graphics}
\usepackage[utf8]{inputenc}
\usepackage{url}
\usepackage{tikz} 
\usetikzlibrary{arrows}
\usetikzlibrary{snakes}
\usetikzlibrary{backgrounds}
\usetikzlibrary{er}
\usetikzlibrary{shapes}
\usepackage{ifthen}
\usepackage[left,pagewise]{lineno}

\usepackage[normalem]{ulem} 

\newcommand{\nc}{\newcommand}

\def\hh#1{\hspace*{0.#1cm}}

\nc{\qed}{$\square$}

\newcommand{\den}[1]{[\![#1]\!]}  
\newcommand{\denn}[1]{[\![\![#1]\!]\!]} 
\nc{\rw}{\to} 
\nc{\tor}{\to} 
\nc{\crwlto}{\rightarrowtriangle}  
\nc{\clto}{\crwlto}                
\nc{\clp}{{\cal P} \vdash_{\mi{CRWL}^{+}}} 
\nc{\denp}[1]{\den{#1}^+} 
\nc{\cldt}{{\cal P} \vdash_{\mi{CRWL}^{d}}} 
\nc{\dend}[1]{\den{#1}^d} 
\nc{\dg}[1]{\den{#1}_{\mi{CRWL}}}
\nc{\cl}{{\cal P} \vdash_{\mi{CRWL}}}
\nc{\gl}{{\cal P} \vdash_{\mathit{CRWL}_{\mi{let}}}}
\nc{\dgl}[1]{\den{#1}_{\mi{CRWL}_{\mi{let}}}}
\nc{\ddgl}[1]{\denn{#1}_{\mi{CRWL}_{\mi{let}}}}
\nc{\ddcl}[1]{\ddgl{#1}}
\nc{\dc}[1]{\dg{#1}} 
\nc{\dcl}[1]{\dgl{#1}} 
\nc{\vdcrwl}{\vdash_{\mathit{CRWL}}}
\nc{\vdcrwll}{\vdash_{\mathit{CRWL_{let}}}}
\nc{\fd}{\to_{lw}}    
\nc{\nr}{\leadsto} 
\nc{\fnr}{{\,\leadsto^l}}  
\nc{\fnrl}{\fnr}  
\nc{\fnrL}{{\,\leadsto^L}}
\nc{\fnre}{\leadsto^{l^*}}  
\nc{\fnrc}[1]{\leadsto^{l^{#1}}}  

\nc{\var}{{\cal V}}
\nc{\ra}{\tor}
\nc{\leqhyp}{\Subset}
\nc{\tot}[1]{{#1}^\tau}
\nc{\tlr}[1]{\widehat{#1}} 
\nc{\jn}{\Join} 
\nc{\tr}{\underline{\mbox{\textbf{t}}}}
\nc{\prog}{{\cal P}}
\nc{\com}[1]{} 
\nc{\crwl}{CRWL}
\nc{\crwllet}{\mbox{\crwl$_{let}$}}
\nc{\crwll}{$\crwllet$}
\nc{\dsord}{\unlhd} 
\nc{\clrule}[1]{\crule{#1}} 
\nc{\lrrule}[1]{\crule{#1}} 
\nc{\clinfer}[3]{\infer[\clrule{{#1}}]{{#2}}{{#3}}}
\nc{\ohs}{\leqhyp} 
\nc{\vran}{vran} 
\nc{\wrt}{wrt.}
\nc{\crule}[1]{({#1})} 
\nc{\refp}[1]{\ref{#1} (page \pageref{#1})}
\nc{\nf}[1]{\downarrow^{#1}} 
\nc{\trs}{{TRS}} 
\nc{\ctrs}{CS} 
\nc{\trss}{\trs's} 
\nc{\ctrss}{\ctrs's} 
\newcommand{\cntx}{{\cal C}}
\nc{\f}{{\,\to^{l}\,}}    
\nc{\fnf}{\to^{\mathit{lnf}}}    
\nc{\fnfe}{\to^{{\mathit{lnf}}^*}}    
\nc{\fL}{{\,\to^{L}\,}}
\nc{\fLe}{\to^{L^*}}  
\nc{\fLnf}{{\,\to^{Lnf}\,}}
\nc{\fLnfe}{{\,\to^{{Lnf}^*}\,}}
\nc{\fe}{\to^{l^*}}  
\nc{\fc}[1]{\to^{l^{#1}}}  
\nc{\fLc}[1]{\to^{L^{#1}}}  
\nc{\fnrLe}{\leadsto^{L^*}}  
\nc{\mi}[1]{\mathit{#1}}
\nc{\eqehs}{\asymp} 

\newcommand{\kk}{}
\newcommand{\ekk}{}

\newcommand{\conscrwl}{\vdash_{\textit{CRWL}}}
\newcommand{\conscrwllet}{\vdash_{\textit{CRWL$_{{\it let}}$}}}

\nc{\con}{{\cal C}}  
\nc{\cnn}[1]{\con[#1]}  
\nc{\cnnp}[1]{\con'[#1]}  

\nc{\ordsus}{\preceq} 
\nc{\ordSus}{\lesssim} 
\nc{\ordap}{\sqsubseteq} 
\nc{\hde}{\varphi} 
\nc{\hdes}{H} 
\nc{\del}{\delta} 
\nc{\ds}{Den} 
\nc{\hds}{HD} 
\nc{\uhs}{\Cup\ } 
\nc{\Uhs}{\mbox{\Large $\uhs$}} 
\nc{\Uhss}[1]{\begin{array}[t]{c}\Uhs\\[-1.2ex]\mbox{\!\!\scriptsize{$#1$}}\end{array}} 
\nc{\sd}[1]{\Delta {#1}} 
\nc{\hl}{\hat{\lambda}} 

\nc{\partes}[1]{{\cal P}(#1)} 

\tikzstyle{npath}=[circle,draw=red!60,fill=red!30,thick,inner sep=0pt,minimum size=7mm]
\tikzstyle{nnarr}=[npath, minimum size=5mm] 
\tikzstyle{apath}=[thick]
\tikzstyle{fondoTerm}=[line width=4mm,join=round,orange!20]

\newtheorem{theorem}{Theorem}
\newtheorem{lemma}{Lemma}
\newtheorem{definition}{Definition}
\newtheorem{example}{Example}
\newtheorem{counterexample}{Counterexample}

\newtheorem{corollary}{Corollary}

\newtheorem{proposition}{Proposition}

\definecolor{orange}{cmyk}{0, 0.6, 0.8, 0}
\definecolor{morado}{cmyk}{0.76, 0.98, 0, 0}
\definecolor{olive}{cmyk}{0.64,0,0.95,0.4}
\definecolor{grisoscuro}{rgb}{0.5,0.5,0.5}

\newcommand{\coment}[1]{}
\newcommand{\sobra}[1]{}

\newcommand{\nuevo}[1]{{#1}}

\newcommand{\tachar}[1]{}
\newcommand{\tacharb}[1]{}
\newcommand{\longversion}[1]{}
\newcommand{\todo}[1]{}
\newcommand{\nota}[1]{}

\title[Theory and Practice of Logic Programming]
{
Rewriting and narrowing for constructor systems with call-time choice
semantics\thanks{This work has been partially supported by the Spanish projects
  FAST-STAMP (TIN2008-06622-C03-01/TIN), PROMETIDOS-CM (S2009TIC-1465) and GPD-UCM (UCM-BSCH-GR58/ 08-910502).}
}

\author[L\'opez-Fraguas et al.]
{FRANCISCO J. L\'OPEZ-FRAGUAS, ENRIQUE MARTIN-MARTIN, \authorbreak JUAN RODR\'IGUEZ-HORTAL\'A and JAIME S\'ANCHEZ-HERN\'ANDEZ\\
 Departamento de Sistemas Inform\'aticos y Computaci\'on\\
 Universidad Complutense de Madrid, Spain\\
 \textnormal{(\textit{e-mail:}} \textnormal{\texttt{fraguas@sip.ucm.es}}, \textnormal{\texttt{emartinm@fdi.ucm.es}}, \\ \textnormal{\texttt{juan.rodriguez.hortala@gmail.com}}, \textnormal{\texttt{jaime@sip.ucm.es}}\textnormal{)}
 }

 \submitted{30 November 2010}
\revised{14 November 2011}
\accepted{4 July 2012}

\begin{document}

\nc{\ppdp}{\fbox{PPDP}\\}
\nc{\wlp}{\fbox{WLP}\\}
\nc{\flops}{\fbox{FLOPS}\\}

\maketitle


\begin{abstract}
Non-confluent and non-terminating constructor-based \nuevo{term} rewriting systems are useful for the purpose of specification and programming. In particular, existing functional logic languages use such kind of rewrite systems to define possibly non-strict  non-deterministic functions. The semantics  adopted for non-determinism is \emph{call-time choice}, whose combination with non-strictness is a non trivial issue, addressed  years ago from a semantic point of view with the Constructor-based Rewriting Logic ({CRWL}), a well-known semantic framework commonly accepted as suitable semantic basis of modern functional logic languages. A drawback of \crwl\ is that it does not come with a proper notion of one-step reduction, which would be very useful to understand and reason about how computations proceed. In this paper we develop thoroughly the theory for the first order version of let-rewriting, a simple reduction notion close to that of classical term rewriting, but extended with a let-binding construction to adequately express the combination of call-time choice with non-strict semantics.
Let-rewriting can be seen as a particular textual presentation of term graph rewriting.
We investigate the properties of let-rewriting, most remarkably their equivalence with respect to a conservative extension of the \crwl-semantics coping with let-bindings, and we show by some case studies
that having two interchangeable formal views (reduction/semantics) of the same language is a powerful reasoning tool.
After that, we provide a notion of let-narrowing which is adequate for call-time choice as proved by soundness and completeness results of let-narrowing with respect to let-rewriting.
Moreover, we relate those let-rewriting and let-narrowing relations (and hence
\crwl) with ordinary term rewriting and narrowing, providing in particular
soundness and completeness of let-rewriting with respect to term rewriting for a class of programs which are deterministic in a semantic sense.\\
To appear in \emph{Theory and Practice of Logic Programming} (TPLP).
\end{abstract}

\section{Introduction}\label{intro}

Term rewriting systems (TRS, \cite{BaaderNipkow-98}) are a well-known and useful formalism from the point of view of specification and programming. The theory of TRS underlies many of the proposals made in the last
decades  for so-called \emph{functional logic programming}, attempting to integrate into a single language the main features of both functional and logic programming ---see \cite{DeGroot86,Hanus94JLP,Hanus07ICLP} for surveys corresponding  to different historical stages of the development of functional logic languages---. Typically, functional logic programs are modeled by some kind of TRS to define functions, and logic programming capabilities are achieved by using some kind of \emph{narrowing} as operational mechanism.
Narrowing, a notion coming from the field of automated theorem proving, generalizes  rewriting by using unification instead of matching in reduction steps. Up to 14 different variants of narrowing were identified in \cite{Hanus94JLP} as being used in different proposals for the integration of functional and logic programming.


Modern functional logic languages  like {\em Curry} \cite{HanusKuchenMoreno-Navarro95ILPS,Han06curry} or
{\em Toy} \cite{LS99,toyreport} consider that programs are   constructor-based term rewrite systems, possibly non-terminating and non-confluent,
thus defining possibly non-strict non-deterministic  functions. For instance,  in  the program
of Figure \ref{coin}, non-confluence comes from the two rules of \emph{coin} and non-termination is due to the rule for \emph{repeat}.

\begin{figure*}
\fbox{
\begin{minipage}{0.60\textwidth}
\begin{center}
\begin{tabular}{ll}

$coin \ra 0$\hh9 & $repeat(X)\ \ra\ X\!:\!repeat(X)$\\
$coin \ra 1$     & $heads(X\!:\!Y\!:\!Y\!s)\ \ra\ (X,Y)$\\

\end{tabular}
\end{center}
\end{minipage}
} 
\caption{A non-terminating and non-confluent program}\label{coin}
\end{figure*}

For non-determinism, those systems adopt
\emph{call-time choice} semantics \cite{hussmann93,GHLR99}, also called sometimes {\em singular} semantics \cite{Sondergaard95}.
Loosely speaking, call-time choice  means to pick a value for each argument of a function application before
applying it. Call-time choice is easier to understand and implement in combination with strict semantics and
eager evaluation in terminating systems as in \cite{hussmann93}, but can be made also compatible ---via partial values and sharing---
with non-strictness and laziness in the presence of non-termination.

In the example of Figure \ref{coin} the expression $heads(repeat(coin))$ can take,
under call-time choice,
the values $(0,0)$ and $(1,1)$, but not  $(0,1)$ or $(1,0)$. The example  illustrates
also a key point here: ordinary term rewriting (called \emph{run-time choice} in \cite{hussmann93}) is an unsound procedure for call-time choice semantics, since a possible term rewriting derivation is:
\[\begin{array}{l}
heads(repeat(coin)) \ra heads(coin:repeat(coin)) \ra\\
heads(0:repeat(coin)) \ra  heads(0\!:\!coin\!:\! repeat(coin)) \ra\\
 heads(0:1:repeat(coin)) \ra (0,1)
\end{array}\]
In operational terms, call-time choice requires to \emph{share} the value of all copies of a given subexpression created during reduction (all the occurrences of $coin$, in the reduction above). In contrast, with ordinary term rewriting all copies evolve independently. 

\kk
It is commonly accepted (see e.g. \cite{Hanus07ICLP}) that call-time choice semantics
combined with non-strict semantics is adequately formally expressed by the
CRWL framework\footnote{\crwl\ stands for Constructor Based ReWriting Logic.}
\cite{GHLR96,GHLR99}, whose main component is a proof calculus  that determines the semantics of programs and expressions.
The flexibility and usefulness of \crwl\ is evidenced by the large set of  extensions  that have been devised for it, to cope with relevant aspects of declarative
programming: higher order functions, types, constraints, constructive failure, \ldots; see \cite{Rod01}
for a survey  on the \crwl\ approach.
However, a drawback of the \crwl-framework
is its lack of a proper one-step reduction mechanism that could play a role similar to term rewriting with respect to equational logic.
Certainly  \crwl\ includes operational procedures in the form of goal-solving calculi
\cite{GHLR99,vado03} to solve so-called \emph{joinability} conditions, but they are too complex to be seen as a basic  way to
explain or understand how a reduction can proceed in the presence of  non-strict non-deterministic functions
with call-time choice semantics.

On the other hand, other works have been more influential on the operational side of the field,
specially those based on the notion of  \emph{needed narrowing} \cite{AEH94,AEH00},
a variant of narrowing that organizes
the  evaluation of arguments in function calls in an adequate way (optimal, for some classes of programs).
Needed narrowing
became the `official' operational procedure of functional logic languages,
and has also been subject of several variations and improvements (see \cite{Hanus07ICLP,EscobarMT05}).

These two coexisting  branches of research (one based on \crwl, and the
other based on classical term rewriting, mostly via needed narrowing) have remained disconnected
for many years from the technical point of view, despite the fact that they both refer
to what intuitively is the same programming language paradigm.

A major problem to establish the connection was that the  theory underlying needed narrowing is classical term rewriting, which, as we saw above, is not valid for non-determinism with call-time choice semantics.
This was not a flaw in the conception of needed narrowing, as it emerged in a time when non-deterministic functions had not yet started to play a distinctive role in the functional logic programming paradigm.
The problem is overcome in practice by adding a sharing mechanism to the implementation of narrowing, using for instance standard Prolog programming techniques \cite{CheongFribourg93,LLR93,AntoyHanus00FROCOS}. But this is merely an implementation patch that cannot be used as a precise and
sound technical basis for the application of results and techniques from
the semantic side to the operational side and vice versa. Other works, specially \cite{EchahedJanodet98JICSLP,AHHOV05} have addressed in a more formal way the issue of sharing in functional logic programming, but they are not good starting points to establish a relationship with the  \crwl\ world (see `Related work' below).

In \cite{ppdp2007} we aimed at establishing a bridge, by looking for a new  variant of \nuevo{term} rewriting tailored to call-time choice as realized by \crwl, trying to fulfil the following requirements:

\begin{itemize}
\item  it should be based on a notion of rewrite step useful to follow
 how a computation proceeds step by step.
\item  it should be simple enough to be easily understandable for non-expert potential users.
 (e.g., students or novice programmers) of functional logic languages adopting call-time choice.
\item it should be provably equivalent to \crwl, as a well-established technical formulation of call-time choice.
\item it should serve as a basis of subsequent notion of narrowing
and evaluation strategies. 
\end{itemize}

That was realized in \cite{ppdp2007} by means of let-rewriting, a simple modification of \nuevo{term} rewriting using local bindings in the form of let-expressions to express sharing. Let-rewriting will be fully presented in Section \ref{let-rewriting}, but its main intuitions can be summarized  as follows:
\begin{itemize}
\item[(i)]  do not rewrite a function call if any of its arguments is evaluable (i.e., still contains other function calls), even if there is a matching rule;
\item[(ii)] instead, extract those evaluable arguments  to outer let-bindings of the form $let~X=e~in~e'$;
\item[(iii)]
if after some reduction steps the \emph{definiens} $e$ of the let-binding  becomes a constructor term $t$ ---a value--- then the binding $X/t$ can be made effective in the body $e'$. In this way, the values obtained for $e$ in the reduction are shared, and therefore call-time choice is respected.
\end{itemize}

Consider, for instance, the program example of Figure \ref{coin} and the expression $$\emph{heads(repeat(coin))}$$ for which we previously performed an ordinary term rewriting reduction ending in $(0,1)$. Now we are going to apply liberally the previous intuitive hints as a first illustration of let-rewriting.
Note first that no rewrite step using a program rule can be done with the whole expression \emph{heads(repeat(coin))}, since in this case there is no matching rule.
But we can extract  the argument \emph{repeat(coin)} to a let-binding, obtaining:
\[let~X=repeat(coin)~in~heads(X)\]
Now we cannot rewrite  \emph{repeat(coin)}, even  though the program rule for \emph{repeat} matches it, because \emph{coin} is evaluable.
Again, we can create a let-binding for \emph{coin}, that will be used to share the value selected for \emph{coin},
if at any later step in the reduction \emph{coin} is indeed reduced:
\[let~Y=coin~in~let~X=repeat(Y)~in~heads(X)\]
At this point there is no problem with rewriting \emph{repeat(Y)}, which gives:
\[let~Y=coin~in~let~X=Y:repeat(Y)~in~heads(X)\]
Rewriting \emph{repeat(Y)} again, we have:
\[let~Y=coin~in~let~X=Y:Y:repeat(Y)~in~heads(X)\]
Reducing \emph{repeat(Y)} indefinitely leads to non-termination,
but, at the same time,  its presence inhibits the application of  the binding for $X$. What we can do is creating a new let-binding for the remaining \emph{repeat(Y)}, which results in:
\[let~Y=coin~in~let~Z=repeat(Y)~in~let~X=Y:Y:Z~in~heads(X)\]
Now, the binding for $X$ can be  performed, obtaining:
\[let~Y=coin~in~let~Z=repeat(Y)~in~heads(Y:Y:Z)\]
At this point, we can use the rule for \emph{heads} to evaluate $heads(Y:Y:Z)$, because nothing evaluable remains in its argument $Y:Y:Z$, arriving at:
\[let~Y=coin~in~let~Z=repeat(Y)~in~(Y,Y)\]
 We proceed now by reducing \emph{coin}, for instance, to $0$ (reducing it to $1$ is also possible):
\[let~Y=0~in~let~Z=repeat(Y)~in~(Y,Y)\]
Performing the binding for $Y$ leads to:
\[let~Z=repeat(0)~in~(0,0)\]
Since $Z$ does not occur in $(0,0)$, its binding is junk that could be deleted (there will be a rule for that in the definition of let-rewriting),
and the reduction is finished yielding the value
\[(0,0)\]
It is apparent that $(1,1)$ is another possible result, but not $(0,1)$ nor $(1,0)$, a behavior coherent with call-time choice.

In this example we have tried to proceed in a more or less natural `lazy' way. However, the previous intuitive precepts ---and its complete and precise realization in Section \ref{let-rewriting}--- do not assume any particular strategy for organizing reductions, but only determine which are the `legal movements' in call-time choice respectful reductions. Strategies have been left aside in the paper, not only for simplicity, but also to keep them independent of the basic rules for \nuevo{term} rewriting with sharing (see however Section \ref{subsection:strategies}).

Let-rewriting was  later on extended to cope with narrowing  \cite{wlp2007lnai} and higher order features \cite{LRSflops08}. 

This paper is a substantially revised and completed presentation of the theory of first order let-rewriting and let-narrowing proposed in \cite{ppdp2007,wlp2007lnai}; some contents have been also taken from \cite{LRSflops08}.
Here, we   unify technically  those papers and develop a deeper investigation of the properties of let-rewriting and related semantics issues.
\\[1ex]

\noindent \textbf{Related work~~}
Our let-rewriting and let-narrowing relations are not the only nor the first formal operational  procedures tuned up to accomplish with the call-time choice semantics of functional logic languages. We have already mentioned the goal-solving calculi associated to the \crwl-framework and its variants \cite{GHLR99,GHR97,vado03}.

A natural option to express different levels of sharing in rewriting is given by the  theory of term graph rewriting \cite{BarendregtA87,Plump01}.  In \cite{EchahedJanodet97IMAG,EchahedJanodet98JICSLP}, the theory of needed rewriting and narrowing was
extended to the framework of so-called admissible graph rewriting systems, aiming at formally modeling the operational behavior of functional logic programs.
Originally, those works considered orthogonal systems,  and extra variables were not allowed. These restrictions were dropped in \cite{AntoyBrownChiangENTCS06} (however, a formal treatment of the
extension is missing).

As a matter of fact, our let-rewriting relation can be
understood as a particular textual adaptation and presentation of term graph rewriting in which a shared node is made explicit in the syntax by giving it a name in a let-binding.
The achievements of Echahed's  works are somehow incomparable to ours, even if both are attempts to formalize sharing in constructor based systems. They focus and succeed on
adapting known optimal strategies to the graph rewriting and narrowing setting; they also take profit of the fine-grained descriptions permitted by graphs to manage aspects of data structures like cycles or pointers. However, they do not try to establish a technical relationship with other formulations of call-time choice. In contrast, proving
equivalence of our operational formalisms wrt. the CRWL semantic framework has been a main motivation of our work, but we do not  deal with the issue of strategies, except for a short informal
discussion at the end of the paper.

It is our thought that proving equivalence with respect to \crwl\ of  term graph rewriting as given in  \cite{EchahedJanodet97IMAG} would have been a  task much harder than the route we follow here.
We see a reason for it. The basic pieces that term rewriting and \crwl\ work with are purely syntactic: terms, substitutions, etc.
Graph rewriting recast these notions in terms of graphs, homomorphisms, etc. In contrast, let-rewriting and let-narrowing keep the same set of basic
pieces of term rewriting and \crwl. In this way, the formalisms are relatively close and moving  from one to another becomes technically more natural and comfortable. This applies also to
some further developments of our setting that we have made so far,  
like  the extension to higher order features given in \cite{LRSflops08}, the combination of semantics proposed in \cite{LRSpepm09}, or the application of let-rewriting as underlying formal notion of reduction for type systems in functional logic languages \cite{LFMMRH10,LFMMRH_APLAS2010}. 

Another proposal that can be seen as reformulation of graph rewriting was given in \cite{AHHOV05}, inspired in Launchbury's natural semantics \cite{Lau93} for lazy evaluation in functional programming. It presents two operational (natural and small-step) semantics for functional logic programs supporting sharing  and residuation (a specific feature of Curry).
These semantics use a flat representation of programs coming from an implicit program transformation encoding the demand analysis used by needed narrowing, and some kind of heaps to express bindings for variables. As in our case, let-expressions are used to express sharing. The approach is useful as a technical basis for implementation and program manipulation purposes;
but we think that the approach is too low-level and close to a particular operational strategy to be a completely satisfactory choice as basic abstract reduction mechanism for call-time choice.
In \cite{LRSentcs06} we established a technical relation of \crwl\  with the operational procedures of \cite{AHHOV05}. But this  turned out to be a really hard task, even if it was done only for a restricted class of programs and expressions.

Our work focuses on term rewriting systems as basic formalism, as happens with the majority of papers about the foundations of functional logic programming, in particular the \crwl-series.
The idea of reformulating graph rewriting in a syntactic style by expressing  sharing through let-bindings has been applied also to other contexts,  most remarkably to
$\lambda$-calculus considered as a basis of functional programming \cite{AriolaA95,AriolaFMOW95,AriolaF97,MaraistOW98}.
In a different direction, but still in relation with $\lambda$-calculus, other papers
\cite{KutznerSchmidt98,SchmidtSchaussM08} have extended it with some kind of non-deterministic choice, an idea that comes back to McCarthy's \emph{amb} \cite{mccarthy63basis}.
As a final note, we should mention that our initial ideas about let-rewriting were somehow inspired by  \cite{LS01lpar,Jaime04}
where indexed unions of set expressions ---a construction generalizing the idea of let-expressions ---were used
to express sharing in an extension of \crwl\ to deal with constructive failure.

The rest of the paper is organized as follows. Section \ref{prelim} presents
some preliminaries about term rewriting and the \crwl\ framework; although with them the paper becomes almost self-contained, some familiarity with the basic notions of TRS certainly help to read the paper. Section
\ref{discussion} contains a first discussion
about failed or partial solutions to the problem of expressing non-strict call-time choice by a simple notion of rewriting.
Section \ref{let-rewriting} is the central part of the paper. First, it introduces local
bindings in the syntax to express sharing, defines let-rewriting as an
adequate notion of rewriting for them and proves some intrinsic properties of let-rewriting.
After that, in Section \ref{crwllet}, we extend the {CRWL}-logic to a new $\crwllet$-logic able to deal with lets in programs and expressions, and we investigate in depth
the properties of the induced semantics, mostly through the notion of \emph{hypersemantics}. Finally, in Section \ref{equivalence} we prove results of soundness and completeness of let-rewriting with respect to $\crwllet$, which have as corollary the equivalence of both, and hence the equivalence of let-rewriting and \crwl\ for programs and expressions not containing lets, as the original \crwl\ ones are.
Section \ref{SemEqs} aims at showing the power of having reduction and semantics as equivalent interchangeable tools for reasoning, including a remarkable case study.
In Section \ref{let-narrowing} we generalize the notion of let-rewriting to that of let-narrowing and give soundness and completeness results of the latter with respect to the former.
At the end of the section we give some hints on how computations can be organized according to known narrowing strategies.
Section \ref{letR-classR} addresses the relationship between
 let-rewriting and  classical term rewriting, proving in particular their
equivalence for  semantically deterministic programs.
Finally, Section
\ref{conclusions} analyzes our contribution and suggests further work. 
For the sake of readability, most of the (fully detailed) proofs have been moved to \ref{proofs}.


\section{Preliminaries}\label{prelim}
\subsection{Constructor based term rewriting systems}\label{trs}
We assume a fixed first order signature $\Sigma=CS\cup FS$, where $CS$ and $FS$
are two disjoint sets of constructor and defined function symbols respectively,
each of them with an associated
arity. We  write $CS^n$ and $FS^n$ for the set of constructor and function
symbols of arity $n$ respectively, and $\Sigma^n$ for $CS^n \cup FS^n$. As usual notations we
write $c,d,\ldots$ for constructors, $f,g,\ldots$ for functions and $X,Y, \ldots$ for
variables taken from a denumerable set $\var$. The notation
$\overline{o}$ stands for tuples of any kind of syntactic objects.

The set $Exp$ of {\it expressions} is defined as $Exp \ni e::= X \mid h(e_1,\ldots,e_n)$, where $X \in \var$, $h\in \Sigma^n$ and
$e_1,\ldots,e_n\in Exp$. The set $CTerm$ of {\it constructed terms} (or {\it
  c-terms}) has the same definition of $Exp$, but with $h$ restricted to $CS^n$
(so $CTerm \varsubsetneq Exp$). The intended meaning is that $Exp$ stands for
evaluable expressions, i.e., expressions that can contain (user-defined)
function symbols, while $CTerm$ stands for data terms representing values. We
will write $e,e',\ldots$ for expressions and $t,s,p,t', \ldots$ for c-terms. The set
of variables occurring in an expression $e$ will be denoted as $var(e)$.

\emph{Contexts} (with one hole) are defined by $Cntxt \ni {\cal C}::=[\ ]\mid
h(e_1,\ldots,{\cal C},\ldots,e_n)$, where
$h\in \Sigma^n$.
The application of a context ${\cal C}$ to an expression $e$, written as ${\cal C}[e]$, is defined inductively 
as follows:
$$
\begin{array}{rcl}
[\ ][e] & = & e \\
h(e_1,\ldots,{\cal C},\ldots,e_n)[e] & = & h(e_1,\ldots,{\cal C}[e],\ldots,e_n)
\end{array}
$$

{\em Substitutions} are finite mappings $\sigma:\var \longrightarrow Exp$ which extend naturally to
$\sigma : Exp \longrightarrow Exp$. We write $e\sigma$ for the application of the substitution
$\sigma$ to $e$. The domain and variable range of a substitution $\sigma$ are defined as $dom(\sigma) =
\{X\in \var \mid X\sigma \neq X\}$ and $\vran(\sigma) = \bigcup_{X\in dom(\sigma)}var(X\sigma)$.
By $[X_1/e_1, \ldots, X_n/e_n]$ we denote the substitution $\sigma$ such that $Y\sigma = e_i$ if $Y \equiv X_i$ for some $X_i \in \{X_1, \ldots, X_n\}$, and $Y\sigma = Y$ otherwise.
Given a set of variables $D$, the notation $\sigma|_{D}$ represents the
substitution $\sigma$ restricted to $D$ and $\sigma|_{\backslash D}$ is a
shortcut for $\sigma|_{({\cal V} \backslash D)}$. A {\em c-substitution} is a
substitution $\theta$ such that $X\theta \in CTerm$ for all $X \in
dom(\theta)$. We write $Subst$ and $CSubst$ for the sets of substitutions and
c-substitutions respectively.

A {\em term rewriting system} 
is any set of rewrite rules of the form $l \tor r$ where $l, r \in Exp$ and $l \not\in \var$. A {\em constructor based rewrite rule} or {\em program rule} has the form $f(p_1, \ldots, p_n) \tor r$
where $f\in FS^n$, $r \in Exp$ and $(p_1, \ldots, p_n)$ is a linear tuple of c-terms, where linear
means that no variable occurs twice in the tuple.
Notice that we allow $r$ to have extra variables (i.e., variables not occurring in the left-hand side). To be precise, we say that $X$ is an extra variable in the rule $l \tor r$ iff $X \in var(r) \setminus var(l)$, and by $vExtra(R)$ we denote the set of extra variables in a rule $R$.
Then a {\em constructor system} 
or \emph{program} $\prog$ is any set of program rules, i.e., a term rewriting system composed only of program rules.

Given a program ${\cal P}$, its associated \emph{rewrite relation} $\rw_{\cal P}$ is defined as ${\cal C}[l\sigma] \rw_{\cal P} {\cal C}[r\sigma]$ for any context ${\cal C}$, rule $l \tor r \in {\cal P}$ and $\sigma \in Subst$.
There, the subexpression $l\sigma$ is called the redex used in that \emph{rewriting step}.
Notice that $\sigma$ can instantiate extra variables to any expression.
For any binary relation ${\cal R}$ we write ${\cal R}^*$ for the reflexive and transitive closure of ${\cal R}$, and ${\cal R}^n$ for the composition of ${\cal R}$ with itself $n$ times.
We write $e_1 \stackrel{*}{\rw_{\cal P}} e_2$  for a term rewriting \textit{derivation} or \textit{reduction} from $e_1$ to $e_2$, and  $e_1 \stackrel{n}{\rw_{\cal P}} e_2$ for a $n$-step reduction.  $e_2$ is a \emph{normal form} wrt. $\rw_{\cal P}$, written as $\nf{{\cal P}}\!\! e_2$, if 
there is not any
$e_3$ such that $e_2 \ {\rw_{\cal P}}\ e_3$;
and $e_2$ is a normal form for $e_1$ wrt. $\rw_{\cal P}$, written as $e_1 \nf{{\cal P}}\!\! e_2$, iff $e_1 \stackrel{*}{\rw_{\cal P}} e_2$ and $e_2$ is a normal form.
When presenting derivations,  we will sometimes underline the redex used at each rewriting step.
In the following, we will usually omit the reference to ${\cal  P}$ when writing $e_1 \rw_{\cal P} e_2$, or denote it by $\prog \vdash e_1 \rw e_2$. 

A program ${\cal P}$ is confluent if for any $e,e_1,e_2\in
Exp$ such that $e\ra^*_{\cal P} e_1$, $e\ra^*_{\cal P} e_2$ there exists $e_3\in Exp$ such that both
$e_1\ra^*_{\cal P} e_3$ and $e_2\ra^*_{\cal P} e_3$.

\subsection{The \crwl\ framework}\label{sectPrelimCRWL}
We present here a simplified version of the \crwl\ framework \cite{GHLR96,GHLR99}.
The original \crwl\/ logic  considered also the possible presence
of \emph{joinability} constraints as conditions in rules in order to give a better treatment of
strict equality as a built-in, a subject orthogonal to the aims of this work.
Furthermore,  it is possible to replace
conditions by the use of an \emph{if\_then} function, as has been technically
proved in \cite{Jaime04} for \crwl\ and in \cite{AntoyJSC05} for term rewriting.
Therefore,
we consider only unconditional program rules.

In order to deal with non-strictness at the semantic level, we enlarge $\Sigma$
with a new constant (i.e., a 0-ary constructor symbol) $\perp$ that stands for the undefined value. The sets $Exp_\perp$, $CTerm_\perp$, $Subst_\perp$,  $CSubst_\perp$
of partial expressions, etc., are defined naturally. Notice that $\perp$ does not appear in programs.
Partial expressions are ordered by the {\em approximation} ordering $\sqsubseteq$ defined as the least
partial ordering satisfying $\perp \sqsubseteq e$ and $e \sqsubseteq e'
\Rightarrow {\cal C}[e] \sqsubseteq {\cal C}[e']$
for all $e,e' \in Exp_\perp, {\cal C} \in {\it Cntxt}$. This partial ordering
can be extended to substitutions: given $\theta,\sigma\in Subst_\bot$ we say
$\theta\sqsubseteq\sigma$ if $X\theta\sqsubseteq X\sigma$ for all $X\in\var$.

The semantics of a program ${\cal P}$ is determined in  \crwl\ by means of a proof calculus able to derive
reduction statements of the form $e \rightarrowtriangle t$, with $e \in Exp_\perp$ and $t \in CTerm_\perp$,
meaning informally that $t$ is (or approximates to) a possible value of $e$, obtained by
evaluating $e$ using ${\cal P}$ under call-time choice.

The \crwl-proof calculus is presented in Figure
\ref{fig:crwl}. Rule {\bf (B)} allows any expression to be undefined or not
evaluated (non-strict semantics).
Rule {\bf (OR)} expresses that to evaluate a function call we must choose a compatible
program rule, perform parameter passing (by means of a c-substitution $\theta$)
and then reduce the right-hand side. The use of c-substitutions in {(OR)} is essential
to express call-time choice; notice also that by the effect of $\theta$ in  {(OR)},
extra variables in the right-hand side of a rule can
be replaced by any partial c-term, but not by any expression as in ordinary \nuevo{term} rewriting $\ra_{\cal P}$.
We write ${\cal P} \conscrwl e \crwlto t$ to express that $e\crwlto t$ is derivable in the \crwl-calculus using
the program ${\cal P}$,
but in many occasions we will omit the mention to ${\cal P}$, writing simply $e \crwlto t$.

%
%
%

\begin{figure}
\framebox[.85\textwidth]{
\begin{minipage}{0.8\textwidth}
\begin{center}
  { 
    \begin{tabular}{l}
      {\bf (B)}$\ $
      $\begin{array}{c}
        \\[-.2cm]
        \hline
        \\[-.6cm]
        e\crwlto\bot
      \end{array}$
      \qquad\qquad
      {\bf (RR)}$\ $
      $\begin{array}{c}
        \\[-.2cm]
        \hline
        \\[-.6cm]
        X\crwlto X
      \end{array}$ $\qquad  X\in \var$\\[0.5cm]

      {\bf (DC)}$\ $
      $\begin{array}{c}
        e_1\crwlto t_1\ \ldots\ e_n\crwlto t_n\\[-.2cm]
        \hline
        \\[-.6cm]
        c(e_1,\ldots,e_n)\crwlto c(t_1,\ldots,t_n)
      \end{array}$ $\qquad c\in CS^n$ \\[0.5cm] 

      {\bf (OR)}$\ $
      $\begin{array}{c}
        e_1\crwlto p_1\theta \ldots\ e_n\crwlto p_n\theta~~r\theta \crwlto t\\[-.2cm]
        \hline
        \\[-.6cm]
        f(e_1,\ldots,e_n)\crwlto t
      \end{array}$ \qquad
      $
      \begin{array}{l}
        (f(p_1,\ldots,p_n) \tor r) \in \prog\\
        \theta \in CSubst_\perp
      \end{array}$\\[0.5cm]
    \end{tabular}
    } 
\end{center}
\end{minipage}
} 
    \caption{Rules of \crwl}
    \label{fig:crwl}
\end{figure}

\begin{definition}[\crwl-denotation]
Given a program
$\mathcal{P}$, the \emph{CRWL-denotation} of an expression $e \in Exp_\perp$ is
defined as $$\den{e}^{\mathcal{P}}_{\it CRWL}=\{t\in CTerm_\perp \mid \mathcal{P} \conscrwl e\crwlto t\}$$
We will usually omit the subscript \crwl\ and/or the superscript $\prog$ when implied by the context.
\end{definition}

%
%
%

As an example, Figure \ref{fig:crwlder} shows a {\em \crwl-derivation} or {\em \crwl-proof} for the
statement $heads(repeat(coin)) \crwlto (0,0)$, using the
program of Figure \ref{coin}.
Observe that in the derivation there is only one reduction statement for $coin$ (namely $coin \crwlto 0$), and the
obtained value $0$ is then {\em shared} in  the whole derivation, as corresponds to call-time choice.
\begin{figure*}
\begin{center}
\fbox{
{\scriptsize
\infer[{OR}]{heads(repeat(coin)) \crwlto (0,0)}
 {\infer[{OR}]{\hh{20} repeat(coin) \crwlto 0:0:\perp\hh{20}}
   {\infer[{OR}]{\mathbf{coin \crwlto 0}}{\infer[{DC}]{0 \crwlto 0}{}}\hh3
    \infer[{DC}]{\hh{20} 0:repeat(0) \crwlto 0:0:\perp\hh{20}}
       {\infer[{DC}]{0 \crwlto 0}{}\hh2
        \infer[{OR}]{\hh{15} repeat(0) \crwlto 0:\perp\hh{15}}
         {\infer[{DC}]{0 \crwlto 0}{}\hh2
          \infer[{DC}]{0:repeat(0) \crwlto 0:\perp}
             {\infer[{DC}]{0 \crwlto 0}{}\hh2
              \infer[{B}]{repeat(0) \crwlto \perp}{}
             }
         }
      }
   }\hh4
  \!\!\!\!\!\!\!\!\!\!\!\!\!\!\!\!\infer[{DC}]{(0,0) \crwlto (0,0)}{\infer[{DC}]{0 \crwlto 0}{}\ \infer[{DC}]{0 \crwlto 0}{}}
 }
} 
}
\end{center}
    \caption{A \crwl-derivation for $heads(repeat(coin)) \crwlto (0,0)$}
    \label{fig:crwlder}
\end{figure*}

In  alternative derivations, $coin$ could be reduced to $1$ (or to $\perp$).
It is easy to check that: 
%
$$
\den{heads(repeat(coin))} = \{(0,0),(1,1),(\perp,0),(0,\perp),(\perp,1),(1,\perp),(\perp,\perp),\perp\}
$$
Note that $(1,0), (0,1) \not\in \den{heads(repeat(coin))}$.

~\\
\indent We stress the fact that the \crwl-calculus {\em is not} an operational mechanism
for executing programs, but a way of describing the logic of programs.
In particular,  the rule (B) is a semantic artifact to reflect in a \crwl-proof of a statement $e\crwlto t$ the fact that, for obtaining $t$
as value of $e$, one does not need to know the value of a  certain subexpression $e'$ (to which the rule (B) is applied). But the calculus comes with no indication of when to apply (B) in a successful proof.
At the operational
level, the \crwl\ framework is accompanied  with various lazy narrowing-based goal-solving calculi
\cite{GHLR99,vado03} not considered in this paper.


~\\
One of the most important
properties of \crwl\
is its compositionality, a property very close to the DET-additivity property for algebraic specifications of \cite{hussmann93}  or the referential transparency property of \cite{Sondergaard90}. Compositionality shows that the \crwl-denotation of any expression placed in a context only depends on the \crwl-denotation of that expression. This implies that the semantics of a whole expression depends only on the semantics of its constituents, as shown by the next result, which is an adaptation of a similar one proved for the higher order case in \cite{LRSflops08}.
\begin{theorem}[Compositionality of \crwl]\label{thCompoCrwl}
For any $\con \in Cntxt$, $e, e' \in Exp_\perp$
$$
\den{\con[e]} = \bigcup\limits_{t \in \den{e}} \den{\con[t]}
$$
As a consequence: $\den{e} = \den{e'} \Leftrightarrow \forall \con \in Cntxt. \den{\con[e]} = \den{\con[e']}$
\end{theorem}


According to this result we can express for example
\begin{center}
$\den{heads(repeat(coin))}=\bigcup_{{t\in\den{coin}}}\bigcup_{s\in\den{repeat(t)}}
\den{heads(s)}$
\end{center}
The right hand side has an intuitive reading that reflects call-time choice:
get a
value $t$ of $coin$, then get a value $s$ of $repeat(t)$ and then get a
value of $heads(s)$.

In Theorem \ref{thCompoCrwlbis} we give an alternative formulation to the compositionality property.
Although it is  essentially equivalent to Theorem \ref{thCompoCrwl}, it is a somehow more abstract statement, based on the
notion of \emph{denotation of a context} introduced in Definition \ref{def:denContxt}. Our main reason for developing such alternative is to give good insights for the compositionality results
of the extension of \crwl\ to be presented in Section \ref{sectEqLetrwCrwlLet}.

We will use sometimes $\ds$ as an alias for $\partes{CTerm_\perp}$, i.e, for the kind of objects that are \crwl-denotations of expressions%
\footnote{$\ds$ is indeed a superset of the set of actual denotations, which are particular elements of $\partes{CTerm_\perp}$, namely \emph{cones} ---see \cite{GHLR99}---. But this is not relevant to the use we make of $\ds$.}.
We define the denotation of a context $\con$ as  a denotation transformer that reflects call-time choice.

\begin{definition}[Denotation of a context]\label{def:denContxt}
Given $\con \in Cntxt$, its denotation is a function $\den{\con}:\ds \rightarrow \ds$ defined as
\[\den{\con}\del = \bigcup_{t\in \del}\den{\con[t]}, ~\forall \del\in \ds\]
\end{definition}

With this notion, compositionality can be trivially re-stated as follows:

\begin{theorem}[Compositionality of \crwl, version 2]\label{thCompoCrwlbis}
For any $\con \in Cntxt$ and $e, e' \in Exp_\perp$
$$
\den{\con[e]} = \den{\con}\den{e}
$$
As a consequence: $\den{e} = \den{e'} \Leftrightarrow \forall \con \in Cntxt. \den{\con[e]} = \den{\con[e']}$
\end{theorem}


The formulation of compositionality given by Theorem \ref{thCompoCrwlbis} makes
even more apparent than Theorem \ref{thCompoCrwl} the fact that the syntactic
decomposition of an expression $e$ in the form $\con[e']$ has a direct semantic
counterpart, in the sense that the semantics of $e$ is determined by the
semantics of its syntactic constituents $\con$ and $e'$.
However, Theorems \ref{thCompoCrwl} and \ref{thCompoCrwlbis} are indeed of the same strength, since each of them can be easily proved from the other.

\section{\crwl\ and rewriting: a first discussion}\label{discussion} 
Before presenting let-rewriting we find interesting to discuss a couple of (in principle) shorter solutions to the problem of expressing non-strict call-time choice semantics
by means of a simple  one-step reduction relation.
A first question is whether a new  relation is needed at all: maybe call-time choice can be expressed by ordinary \nuevo{term} rewriting via a suitable program transformation. The next result  shows that in a certain technical sense this is not possible: due to different closedness under substitution and compositionality properties of call-time choice and term rewriting, none of them can be naturally simulated by each other.

\begin{proposition}\label{exNoTransTrsCrwl} 
There is a program $\prog$ for which the following two conditions hold:
\begin{itemize}
\item[i)] no term rewriting system (constructor based or not) ${\cal P}'$ verifies
$$
\prog \vdash_{\crwl} e \clto t \mbox{ iff } \prog' \vdash e \rw^* t  ~\mbox{, for all $e \in Exp, t \in CTerm$}
$$
\item[ii)] no program ${\cal P}'$ verifies
$$
 \prog \vdash e \rw^* t \mbox{ iff }  \prog' \vdash_{\crwl} e \clto t ~\mbox{, for all $e \in Exp, t \in CTerm$}
$$
\end{itemize}
\end{proposition}

\begin{proof}
The following simple program ${\cal P}$ suffices:
\begin{center}
$~~f(X) \ra c(X,X)$\hh5
$~~coin \ra 0$\hh5
$~~coin \ra 1$
\end{center}

\emph{i)}~ We reason by contradiction. Assume there is a term rewriting system ${\cal P}'$ such that:   ${\cal P} \conscrwl e \crwlto t$
$\Leftrightarrow$ $e \ra_{{\cal P}'}^* t$, for all $e,t$.
Since  ${\cal P} \conscrwl f(X) \crwlto c(X,X)$, we must have $f(X) \ra
_{{\cal P}'}^*c(X,X)$. Now, since $\ra_{{\cal P}'}^*$ is closed under
substitutions \cite{BaaderNipkow-98},  we have  $f(coin)  \ra _{{\cal P}'}^*c(coin,coin)$, and therefore
$f(coin)  \ra _{{\cal P}'}^*c(coin,coin) \ra _{{\cal P}'}^* c(0,1)$.
But it is easy to see that ${\cal P}   \conscrwl f(coin) \crwlto c(0,1)$ does not hold.

\emph{ii)}~  Assume now there is a program ${\cal P}'$ such that:   $\prog \vdash e \rw^* t$ $\Leftrightarrow$  $\prog' \vdash_{\crwl} e \clto t$, for all $e,t$.
Since $\prog \vdash f(coin) \rw^* c(0,1)$, we have
$\prog' \vdash_{\crwl} f(coin) \clto  c(0,1)$.
By compositionality of call-time choice (Theorem \ref{thCompoCrwl}), there must exist a possibly partial
$t \in CTerm_\perp$ such that $\prog' \vdash_{\crwl} coin \clto  t$ and
$\prog' \vdash_{\crwl} f(t) \clto c(0,1)$. Now we distinguish cases on the value of $t$:
\begin{itemize}
\item[(a)] If $t\equiv \perp$, then
monotonicity of $CRWL$-derivability ---see \cite{GHLR99} or Proposition \ref{propCrwlletPolar} below--- proves that $\prog' \vdash_{\crwl} f(s) \clto c(0,1)$
for any $s\in CTerm_\perp$,
in particular $\prog' \vdash_{\crwl} f(0) \clto c(0,1)$. Then, by the assumption on $\prog'$, it should be  $\prog \vdash f(0) \rw^* c(0,1)$, but this is not true.
\item[(b)] If $t\equiv 0$, then  $\prog' \vdash f(0) \clto c(0,1)$ as before.
The cases $t\equiv 1$, $t\equiv Y$ or  $t\equiv d(\overline{s})$ for a constructor $d$ different
from $0,1$ lead to similar contradictions.
\end{itemize}
\end{proof}


 Notice that Proposition \ref{exNoTransTrsCrwl} does not make any assumption about signatures: in any of \emph{i) or ii)}, no extension of the signature can lead to a simulating $\prog'$. This does not contradict Turing completeness of term rewriting systems. Turing completeness arguments typically rely on encodings not preserving the structure of data, something not contemplated in Proposition \ref{exNoTransTrsCrwl}.

In a second trial, requiring minimal changes over ordinary term rewriting, we impose  that the substitution $\theta$ in a rewriting step must be a c-substitution,
as in the rule {(OR)} of \crwl. 
This is done in the one-step rule {\bf (OR$^{rw}$)} in Figure \ref{br}. According to it, the step $heads(repeat(coin)) \ra heads(coin:repeat(coin))$ in the introductory example of Figure \ref{coin} would not be legal anymore.
However, (OR$^{rw}$) corresponds essentially to innermost evaluation, and  is not enough to deal with  non-strictness, as the following example shows:

\begin{example}
Consider the rules $f(X) \ra 0$ and ${\it
  loop}\ra {\it loop}$.
With a non-strict semantics $f(loop)$ should be reducible to $0$.
But (OR$^{rw}$)
does not allow the step $f(loop) \ra 0$;  only
$f(loop) \ra f(loop) \ra \ldots$ is a valid (OR$^{rw}$)-reduction, thus leaving $f(loop)$ semantically undefined,
as would correspond to a strict semantics.
\end{example}

At this point, the  rule {(B)} of \crwl\ is a help, since it allows to discard the evaluation of any (sub)-expression by reducing it to $\perp$.
The result of this discussion is the one-step reduction relation $\rightarrowtail$ given in Figure \ref{br}.

\begin{figure*}
\framebox[.95\textwidth]{
\hspace{-1cm}
\begin{minipage}{.925\textwidth}
\begin{tabular}{@{$\hspace{0.5cm}$}lllll @{$\hspace{0.5cm}$}}
{\bf (B$^{rw}$)}  &${\cal C}[e]$ & $\rightarrowtail$ & ${\cal C}[\perp]$\hh3 & $\forall \con \in {\it Cntxt}, e \in Exp_\perp$\\
{\bf (OR$^{rw}$)}\ &${\cal C}[f(t_1\theta,\ldots,t_n\theta)]$ & $\rightarrowtail$ & ${\cal C}[r\theta]$ & $\forall \con \in {\it Cntxt}$,
\hspace{-.35cm} $\begin{array}[t]{l}f(t_1,\ldots,t_n)\ra r \in \mathcal{P},\\ \theta\in CSubst_{\bot}\end{array}$\\
\end{tabular}
\end{minipage}
} 
\caption{A one-step reduction relation for non-strict call-time choice}\label{br}\label{fig:CrwlRw}
\end{figure*}

This relation satisfies our initial goals to a partial extent, as it
is not difficult to prove the following equivalence result.

\begin{theorem}\label{brcrwl}
Let ${\cal P}$ be a \crwl-program, $e \in Exp_\perp$ and $t \in CTerm_\perp$. Then:
$${\cal P} \conscrwl e \crwlto t\ \mathit{iff}\ e \rightarrowtail^*_{\cal P} t$$
\end{theorem}

This result has
an interesting reading: non-strict call-time choice can be achieved via innermost
evaluation if at any step one has the possibility of reducing a subexpression to $\perp$ (then, we could speak also of \emph{call-by-partial value}). For instance,
 a $\rightarrowtail$-rewrite sequence with the example of Figure \ref{coin} would be:
\[\begin{array}{l}
heads(repeat(coin)) \rightarrowtail  heads(repeat(0))\rightarrowtail\\
 heads(0:repeat(0)) \rightarrowtail heads(0:0:repeat(0)) \rightarrowtail\\
 heads(0:0:\perp) \rightarrowtail (0,0)
\end{array}\]

This gives useful intuitions about non-strict call-time choice and can actually serve for a very easy implementation of it, but has a major drawback:
in general, reduction  of a subexpression $e$ requires a \emph{don't know} guessing between
{(B$^{rw}$)} and {(OR$^{rw}$)}, because at the moment of reducing $e$ it is not known whether its value will be needed or not later on in the computation.
Instead of reducing to  $\perp$,  let-rewriting will create a  let-binding \emph{let U=e in \ldots}, which does not imply any guessing and keeps $e$  for its eventual  future use.

\section{Rewriting with local bindings}

\label{let-rewriting}


Inspired by
\cite{AriolaA95,AriolaFMOW95,AriolaF97,MaraistOW98,Plump98,Jaime04},
let-rewriting extends the syntax of expressions by adding local bindings to
express sharing and call-time choice. Formally the syntax for {\it let-expressions} is:
$$
LExp \ni e ::= X~|~h(e_1, \ldots, e_n)~|~let~X=e_1~in~e_2
$$
where $X \in {\cal V}$, $h \in \Sigma^n$, 
and $e_1, e_2, \ldots, e_n \in LExp$. The intended behaviour of $let~X=e_1~in~e_2$ is that the expression $e_1$ will be reduced only once (at most) and then its corresponding value will be shared within $e_2$. For $let~X=e_1~in~e_2$ we call $e_1$ the {\em definiens} of $X$, and $e_2$ the {\em body} of the let-expression.

The sets $FV(e)$ of {\em free} and $BV(e)$ {\em bound} variables of $e \in LExp$ are defined as:
$$
\begin{array}{c}
  \begin{array}{l}
    FV(X)=\{X\}\\
    FV(h(\overline{e}))= \bigcup_{e_i\in\overline{e}}FV(e_i)\\
    FV(let~X=e_1~in~e_2)=FV(e_1)\cup(FV(e_2)\backslash\{X\})\\[.3cm]
    BV(X)=\emptyset\\
    BV(h(\overline{e}))= \bigcup_{e_i\in\overline{e}}BV(e_i)\\
    BV(let~X=e_1~in~e_2)=BV(e_1)\cup BV(e_2)\cup \{X\}
  \end{array}
\end{array}
$$
Notice that with the given definition of $FV(let~X=e_1~in~e_2)$ there are not recursive let-bindings in the language
since the possible occurrences of $X$ in $e_1$ are not considered bound
and therefore refer to a `different' $X$. For example, the expression
$let~X=f(X)~in~g(X)$ can be equivalently written as $let~Y=f(X)~in~g(Y)$.
This is similar to what is done in \cite{MaraistOW98,AriolaFMOW95,AriolaF97}, but not in
\cite{AHHOV05,Lau93}.
Recursive lets have their own interest but there is not a general
consensus in the functional logic community about their meaning in presence of
non-determinism.
We remark also that the let-bindings introduced by let-rewriting derivations to
be presented in Section \ref{sect:letRwRelation}
are not recursive.
Therefore, recursive lets  are not considered in this work.

We will use the notation $let~\overline{X = a}~in~e$ as a shortcut for
$let~X_{1}=a_1~in\ldots in~let$ $X_{n}=a_n$ $in$ $e$.
The notion of {\it one-hole context} is also extended  to the new syntax:
$$
{\cal C} ::= [\ ]~|~let~X={\cal C}~in~e~|~let~X=e~in~{\cal C}~|~h(\ldots, {\cal C}, \ldots)
$$
By default, we will use contexts with lets  from now on.

Free variables of contexts are defined as for expressions, so that $FV(\con)=FV(\con [h])$, for any $h \in \Sigma^0$.
However, the set $BV(\con)$ of 
\emph{variables bound by a context} is defined quite differently because it
consists only of those let-bound variables visible from the hole of $\con$.
Formally:
\begin{center}
\begin{tabular}{l}
$BV([\ ]) = \emptyset$\\
$BV(h(\ldots,\con,\ldots)) = BV(\con)$\\
$BV(let~X = e~in~\con) = \{X\} \cup BV(\con)$\\
$BV(let~X = \con~in~e) = BV(\con)$
\end{tabular}
\end{center}

As a noticeable difference with respect to \cite{ppdp2007}, from now on we will
allow to use lets in any program, so our program rules have the shape $f(p_1, \ldots, p_n)\to r$, for $f\in FS^n$, $(p_1, \ldots, p_n)$ a linear tuple of c-terms, and $r\in LExp$. 
Notice, however, that the notion of c-term does not change: c-terms do not contain function symbols nor lets, although they can contain bound variables when put in an appropriate context as happens for example with the subexpression $(X,X)$ in the expression $let~X=coin~in~(X,X)$.
%

As usual with syntactical binding constructs, we assume a variable convention
according to which bound variables can be consistently renamed as to ensure that
the same variable symbol does not occur free and bound within an
expression. Moreover, to keep simple the management of substitutions, we assume
that whenever  $\theta$ is applied to an expression $e\in LExp$, the necessary
renamings are done in $e$ to ensure $BV(e) \cap (dom(\theta)\cup
\vran(\theta)) = \emptyset$. With all these conditions the rules defining
application of substitutions are simple while avoiding variable capture:
$$
\begin{array}{l}
X\theta = \theta(X), \textrm{ for $X\in\var$}\\
h(e_1, \ldots, e_n)\theta = h(e_1\theta, \ldots, e_n\theta), \textrm{
  for $h\in \Sigma^n$}\\ 
(let~X=e_1~in~e_2)\theta = let~X=e_1\theta~in~e_2\theta
\end{array}
$$

The following example illustrates the use of these conventions.
$$
\begin{array}{l}
(let~X=c(X)~in~let~Y=z~in~d(X, Y))[X/c(Y)] \\
= (let~U=c(X)~in~let~V=z~in~d(U, V))[X/c(Y)] \\
= let~U=c(c(Y))~in~let~V=z~in~d(U, V)
\end{array}
$$

The following substitution lemma will be often a useful technical tool:

\begin{lemma}[Substitution lemma for let-expressions] \label{auxBind}
Let $e, e'\in LExp_{\perp}$, $\theta \in Subst_{\perp}$ and $X \in \var$ such that $X \not\in dom(\theta) \cup \vran(\theta)$.
Then:
$$(e[X/e'])\theta \equiv e\theta[X/e'\theta]$$
\end{lemma}

\subsection{The let-rewriting relation}\label{sect:letRwRelation}

Let-expressions can be reduced step by step by means of the
{\it let-rewriting} relation $\f$, shown in Figure \ref{letrcalc}.
Rule \crule{\textbf{Contx}} allows us to use any subexpression as redex in the derivation. \crule{\textbf{Fapp}} performs a rewriting step in the proper sense, using a program rule. Note that only c-substitutions are allowed, to avoid copying of unevaluated expressions which would destroy sharing and call-time choice.
To  prevent that the restriction of \crule{\textbf{Fapp}} to total c-substitutions results in a strict semantics, we also provide the rule \crule{\textbf{LetIn}} that suspends the evaluation of a subexpression by introducing a let-binding.
If its value is needed later on,  its evaluation can be performed by some \crule{\textbf{Contx}} steps and the result propagated by \crule{\textbf{Bind}}. This latter rule is safe \wrt\ call-time choice because it only propagates c-terms, that is, either completely defined values (without any bound variable) or partially computed values with some suspension (bound variable) on it, which will be safely managed by the calculus. On the other hand, if the bound variable disappears from the body of the let-binding during evaluation, rule \crule{\textbf{Elim}} can be used for garbage collection.
This rule is useful to ensure that normal forms corresponding to values are c-terms.
Finally, \crule{\textbf{Flat}}  is needed for flattening nested lets; otherwise
some reductions could become wrongly blocked or forced to diverge. Consider for
example the program $\{loop \tor loop, g(s(X)) \tor 1\}$ and the expression
$g(s(loop))$, which can be reduced to $let~X=(let~Y=loop~in~s(Y))~in~g(X)$ by
applying \crule{\textbf{LetIn}} twice. Then, without \crule{\textbf{Flat}} we could only perform reduction
steps on $loop$, thus diverging; by using \crule{\textbf{Flat}}, we can obtain
$let~Y=loop~in~let~X=s(Y)~in~g(X)$, which can be finally
reduced to $1$ by applying \crule{\textbf{Bind}}, \crule{\textbf{Fapp}} and
\crule{\textbf{Elim}}. The condition $Y\not\in FV(e_3)$ in \crule{\textbf{Flat}} could be dropped by
the variable convention, but we have included it to keep the rules
independent of the convention. Quite different is the case of \crule{\textbf{Elim}}, where the
condition $X \not\in FV(e_2)$ is indeed necessary.


\begin{figure*}
\framebox{
\begin{minipage}{0.95\textwidth}
\begin{center}
  \begin{description}
     \item[(Fapp)] $f(p_1, \ldots, p_n)\theta~\f~r\theta$, ~~~ if $(f(p_1, \ldots, p_n) \tor r) \in \prog, \theta \in CSubst$\\[0.15cm]
     \item[(LetIn)] $h(\ldots, e, \ldots) \f let~X=e~in~h(\ldots, X, \ldots)$,\\
       if $h\in \Sigma$, $e\equiv f(\overline{e'})$ with $f\in FS$ or $e\equiv let~Y=e'~in~e''$, and $X$ is a fresh variable\\[0.15cm]
     \item[(Bind)] $let~X=t~in~e~\f~ e[X/t]$, ~~~ if $t \in CTerm$\\[0.15cm]
     \item[(Elim)] $let~ X = e_1~in~e_2 \f e_2$, ~~~ if  $X \not\in FV(e_2)$\\[0.15cm]
     \item [(Flat)] $let~X = (let~Y =e_1~in~e_2)~in~e_3 ~\f~ let~Y = e_1~in~(let~X=e_2~in~e_3)$\\
 if $Y \not\in FV(e_3)$ \\[0.15cm]
     \item[(Contx)] ${\cal C}[e] \f {\cal C}[e']$,\\
       if ${\cal C} \neq [\ ]$, $e  \f e'$ using any of the previous rules, and in case $e \f e'$ is a (Fapp) step using $(f(\overline{p}) \tor r) \in \prog$ and $\theta \in CSubst$, then $\vran(\theta|_{\setminus var(\overline{p})}) \cap BV({\cal C}) = \emptyset$.
\end{description}
\end{center}
\end{minipage}
}
\caption{Rules of the let-rewriting relation $\f$}
\label{letrcalc}
\end{figure*}

Note that, in contrast to \crwl\ or the relation $\rightarrowtail$ in  Section \ref{discussion}, let-rewriting does not need to use the semantic value $\perp$, which  does not appear in programs nor in computations.

\begin{example}\label{fig:letDer} 
Consider the program of Figure \ref{coin}. We can perform the following let-rewriting derivation for the expression $heads(repeat(coin))$, where in each step the corresponding redex has been underlined for the sake of readability.
$$
\begin{small}
\begin{array}{ll}
\underline{heads(repeat(coin))} & \crule{LetIn} \\
\f let~X=\underline{repeat(coin)}~in~heads(X) & \crule{LetIn} \\
\f \underline{let~X=(let~Y=coin~in~repeat(Y))~in~heads(X)} & \crule{Flat} \\
\f let~Y=coin~in~let~X=\underline{repeat(Y)}~in~heads(X) & \crule{Fapp} \\
\f let~Y=coin~in~let~X=\underline{Y:repeat(Y)}~in~heads(X) & \crule{LetIn} \\
\f let~Y=coin~in~\underline{let~X=(let~Z=repeat(Y)~in~Y:Z)~in~heads(X)} & \crule{Flat} \\
\f let~Y=coin~in~let~Z=repeat(Y)~in~\underline{let~X=Y:Z~in~heads(X)} & \crule{Bind} \\
\f let~Y=coin~in~let~Z=\underline{repeat(Y)}~in~heads(Y:Z) & \crule{Fapp}\\
\f let~Y=coin~in~let~Z=\underline{Y:repeat(Y)}~in~heads(Y:Z) & \crule{LetIn} \\
\f let~Y=coin~in~\underline{let~Z=(let~U=repeat(Y)~in~Y:U)~in~heads(Y:Z)} & \crule{Flat}\\
\f let~Y=coin~in~let~U=repeat(Y)~in~\underline{let~Z=Y:U~in~heads(Y:Z)} & \crule{Bind}\\
\f let~Y=coin~in~let~U=repeat(Y)~in~\underline{heads(Y:Y:U)} & \crule{Fapp}\\
\f let~Y=coin~in~\underline{let~U=repeat(Y)~in~(Y, Y)} & \crule{Elim} \\
\f let~Y=\underline{coin}~in~(Y, Y) & \crule{Fapp}\\
\f \underline{let~Y=0~in~(Y, Y)} & \crule{Bind}  \\
\f (0,0) \\
\end{array}
\end{small}
$$

Note that there is not a unique $\f$-reduction leading to $(0,0)$. The definition of $\f$, like traditional term rewriting, does not prescribe any particular strategy. The definition of on-demand evaluation strategies for let-rewriting is out of the scope of this paper, and is only informally discussed in Section  \ref{subsection:strategies}.

\end{example}

We study now some properties of let-rewriting that have intrinsic interest and will be useful when establishing a relation to CRWL in next sections.

The same example used in  Proposition \ref{exNoTransTrsCrwl} to show that \crwl\ is not closed under general substitutions shows also that the same applies to let-rewriting. However, let-rewriting is closed under c-substitutions, as expected in a semantics for call-time choice.

\begin{lemma}[Closedness under $CSubst$ of let-rewriting]\label{LRwCerr}
For any $e, e' \in LExp$, $\theta \in CSubst$ we have that $e \f^n e'$ implies $e\theta \f^n e'\theta$.
\end{lemma}

Another interesting matter is the question of
what happens in let-rewriting derivations in which the rule \crule{Fapp} is not used---and as a consequence, the program is ignored.

\begin{definition}[The $\fnf$ relation]\label{def:fnf}
The relation $\fnf$ is defined by the rules of Figure \ref{letrcalc} except \lrrule{Fapp}. 
As a consequence, for any program $\fnf\ \subseteq \f$. 
\end{definition}

We can think about any let-expression $e$ as an 
expression from $Exp$ in which some additional sharing information has been encoded using the let-construction. When we avoid the use of the rule \crule{Fapp} in derivations, we do not make progress in the evaluation of the implicit let-less expression corresponding to $e$, but we change the sharing-enriched representation of that expression in the let-rewriting syntax.
Following terminology from term graph rewriting ---as in fact a let-expression is a textual representation of a term graph--- all the rules of let-rewriting except \crule{Fapp} move between two isomorphic term graphs \cite{Plump01,Plump98}.
The $\fnf$ relation 
will be used to reason about these kind of derivations.

The first interesting property of $\fnf$ is that it is a terminating relation.
\begin{proposition}[Termination of $\fnf$]\label{termlr}
For any program $\prog$, the relation  $\fnf$ is terminating.
As a consequence, every $e \in LExp$ has at least one $\fnf$-normal form $e'$
(written as $e \nf{\it lnf} e'$).
\end{proposition}

However, for nontrivial signatures the relation $\fnf$ is not confluent (hence the relation $\f$ is not confluent either).
\begin{example}
Consider a signature such that $f, g \in FS^0, c \in CS^2$ and $f \not\equiv g$.
Then $c(f, g) \fnfe let~X=f~in~let~Y=g~in~c(X, Y)$ and $c(f, g) \fnfe let~Y=g~in~let~X=f~in~c(X, Y)$, but these expressions do not have a common reduct.
%
\end{example}

The lack of confluence of let-rewriting is alleviated
by a strong semantic property of $\fnf$ which, combined with the adequacy to \crwl\ of let-rewriting that we will see below, may be used as a substitute for confluence in some situations. These questions will be treated in detail in Section \ref{SectSoundLetRw}.

The next result characterizes $\fnf$-normal forms.
What we do in 
$\fnf$ derivations is exposing the computed part of $e$ ---its outer constructor
part--- concentrating it in the body of the resulting let, that is, the part
which is not a function application whose evaluation is pending. This is why we call it  \emph{`Peeling lemma'}.

%
%
%
\begin{lemma}[Peeling lemma]\label{T36}\label{lemma:bigPeeling}
%
%
For any $e, e' \in LExp$, 
if $e \nf{\it lnf} e'$ then $e'$ has the shape $e' \equiv let~\overline{X=f(\overline{t})}~in~e''$ such that $e'' \in \var$ or $e'' \equiv h(\overline{t'})$ with $h \in \Sigma$, $\overline{f} \subseteq FS$ and $\overline{t},\overline{t'} \subseteq CTerm$. \\
Moreover if $e \equiv h(e_1, \ldots, e_n)$ with $h \in \Sigma$, then
$$
e \equiv h(e_1, \ldots, e_n) \fnfe let~\overline{X=f(\overline{t})}~in~h(t_1, \ldots, t_n) \equiv e'
$$
under the conditions above, and verifying also that $t_i \equiv e_i$ whenever $e_i \in CTerm$.

\end{lemma}

The next property of $\f$ and $\fnf$ uses the notion of \emph{shell} $|e|$ of an expression $e$, that is the partial c-term corresponding to the outer constructor part of $e$. More precisely:

\begin{definition}[Shell of a let-expression]
$$
\begin{array}{lcll}
|X| & = & X & \mbox{ for } X \in \var\\
|c(e_1, \ldots, e_n)| & = & c(|e_1|, \ldots, |e_n|) & \mbox{ for } c \in CS\\
|f(e_1, \ldots, e_n)| & = & \perp & \mbox{ for } f \in FS \\
|let~X=e_1~in~e_2| & = & |e_2|[X/|e_1|] \\
\end{array}
$$
\end{definition}
Notice that in the 
case of a let-rooted expression, the information contained in 
the binding is taken into account for building up the shell of 
the whole expression: for instance $|c(let~X=2~in~s(X))| = c(s(2))$.

During a computation, the evolution of shells reflects the progress towards a value. The next result shows that shells never decrease. Moreover,  only \lrrule{Fapp} may change shells. As discussed above, `peeling' steps (i.e. $\fnf$- steps) just modify the representation of the implicit term graph corresponding to a let-expression; thus, they preserve the shell.

\begin{lemma}[Growing of shells]\label{LCascCrec}
\begin{enumerate}
    \item[i)] $e \fe e'$ implies $|e| \sqsubseteq |e'|$, for any $e, e' \in LExp$
    \item[ii)] $e \fnfe e'$ implies $|e| \equiv |e'|$, for any $e, e' \in LExp$
\end{enumerate}
\end{lemma}

\subsection{The {\it CRWL$_{\it let}$} logic}
\label{crwllet}


In this section we  extend the \crwl\ logic to deal with let-expressions, obtaining an enlarged framework that will be useful as a bridge to establish the connection between \crwl\ and let-rewriting.

As in the \crwl\ framework, we  consider partial let-expressions $e\in LExp_\perp$,
defined in the natural way. The approximation order $\ordap$ is also extended to $LExp_\perp$ but now using the notion of context for let-expressions, which in particular implies that $let~X=e_1~in~e_2 \ordap let~X=e'_1~in~e'_2$ iff $e_1 \ordap e'_1$ and $e_2 \ordap e'_2$.
The \crwll\ logic results of adding the following rule \crule{\textbf{Let}} to the \crwl\ logic of Section \ref{sectPrelimCRWL}:


$${\bf (Let)}\
\begin{array}{c}
  e_1\crwlto t_1~~~e_2[X/t_1]\crwlto t\\[-.2cm]
  \hline
  \\[-.6cm]
  let~X=e_1~in~e_2\crwlto t
\end{array}
$$



We write ${\cal P}
\conscrwllet e \crwlto t$ if
$e\crwlto t$ is derivable in the \crwll-calculus using the program ${\cal
  P}$.
In many occasions, we will omit ${\cal P}$.

\begin{definition}[\emph{CRWL$_{\it let}$}-denotation]
Given a program
$\mathcal{P}$, the \emph{CRWL$_{\it let}$-denotation} of  $e \in
LExp_\perp$ is defined as: $$\den{e}^{\mathcal{P}}_{\it CRWL_{\it let}}=\{t\in
CTerm_\perp \mid \mathcal{P} \conscrwllet e\crwlto t\}$$
We will omit the sub(super)-scripts when they are clear by the context.
\end{definition}

There is an obvious relation between \crwl\ and \crwll\ for programs and expressions without lets:
\begin{theorem}[\crwl\ vs. \crwll]\label{thEquivCrwlCrwllet} For any program $\prog$ without lets, and any $e \in Exp_\perp$:
$$
\dc{e}^{\prog} = \dcl{e}^{\prog}
$$
\end{theorem}

This result allows us to skip the mention to \crwl\ or \crwll\ when  referring to the denotation $\den{e}$ of an expression: if some let-binding occurs in $e$ ---or in the program wrt. which the denotation is considered--- then $\den{e}$ can be interpreted only as $\dcl{e}$; otherwise, both denotations coincide.

The \crwll\ logic inherits from \crwl\ a number of useful properties.

\begin{lemma}\label{lemmashells} For any program  $e \in LExp_\perp$, $t,t' \in CTerm_\perp$:
  \begin{enumerate}
  \item[i)] $ t \clto t'$ iff $t' \ordap t$.
  \item[ii)] $|e|\in\den{e}$.
  \item[iii)] $\den{e} \subseteq (|e|\!\!\uparrow)\!\!\downarrow$, where for a given $E \subseteq LExp_\bot$ its upward closure is $E\!\!\uparrow = \{e' \in LExp_\bot |~\exists e \in E.~e \ordap e'\}$, its downward closure is $E\!\!\downarrow = \{e' \in LExp_\bot |~ \exists e \in E.~e' \ordap e\}$, and those operators are overloaded for let-expressions as $e\!\!\uparrow = \{e\}\!\!\uparrow$ and $e\!\!\downarrow = \{e\}\!\!\downarrow$.
  \end{enumerate}
\end{lemma}


The first part of the previous result shows that c-terms can only be reduced to smaller c-terms. The other parts express that the shell of an expression represents `stable' information contained in the expression in a similar way to Lemma \ref{LCascCrec}, as the shell is in the denotation by \emph{ii)}, and everything in the denotation comes from refining it by \emph{iii)}.

The following results are adaptations to \crwll\ of properties known for \crwl\  \cite{GHLR99,vado02}.
The first one states that if we can compute a value for an expression then from greater expressions we can reach smaller values. The second one says that \crwll-derivability is closed for partial c-substitutions.

\begin{proposition}[Polarity of \crwll]\label{propCrwlletPolar} For any  $e, e' \in LExp_\perp$, $t, t' \in CTerm_\perp$, if $e \ordap e'$ and $t' \ordap t$ then $ e \clto t$ implies $ e' \clto t'$ with a proof of the same size or smaller---where the size of a \crwll-proof is measured as the number of rules of the calculus used in the proof.
\end{proposition}

\begin{proposition}[Closedness under c-substitutions of \crwll]\label{closednessCSubst}
For any   $e \in LExp_\perp$, $t \in CTerm_\perp$, $\theta \in CSubst_\perp$, $t \in \den{e}$ implies $t\theta \in \den{e\theta}$.
\end{proposition}

Compositionality is a more delicate issue. Theorem \ref{thCompoCrwl} does not hold for \crwll, as shown by the following example: consider the program $\{f(0) \tor 1\}$, the expression $e \equiv f(X)$ and the context $\con \equiv let~X=0~in~[]$.
 $\con[e]$ can produce the value $1$. However, $f(X)$ can only be reduced to $\perp$, and  $\con[\perp]$ cannot reach the value $1$. The point in that example is that the subexpression $e$ needs some information from the context to produce a value that is then used by the context to compute the value for the whole expression $\con[e]$. This information may only be the definientia of some variables of $e$ that get bound when put in $\con$; with this idea in mind we can state the following weak
 compositionality result for \crwll.

\begin{theorem}[Weak Compositionality of \crwll]\label{thCompoCrwllet}\label{lem:weakcomp}
For any $\con \in Cntxt$, $e \in LExp_\perp$
$$
\den{\con[e]} = \bigcup\limits_{t \in \den{e}} \den{\con[t]}\qquad \mbox{if } BV(\con) \cap FV(e) = \emptyset
$$
As a consequence, $\den{let~X=e_1~in~e_2} = \bigcup_{t_1 \in \den{e_1}}\den{e_2[X/t_1]}$.
\end{theorem}

In spite of not being a fully general compositionality result, Theorem \ref{lem:weakcomp} can be used to prove new properties of
\crwll, like the following monotonicity property related to substitutions, that will be used later on.
It is formulated for the partial order $\ordap$ over $LSubst_\perp$ (defined naturally as it happened for $Susbt_\perp$) and the preorder $\dsord$ over $LSubst_\perp$, defined by $\sigma \dsord \sigma'$ iff $\forall X \in \var, \den{\sigma(X)} \subseteq \den{\sigma'(X)}$.
\begin{proposition}[Monotonicity for substitutions of \crwll]\label{PropMonSubstCrwlLet}
If $\sigma \ordap \sigma'$ or $\sigma \dsord \sigma'$ then $\den{e\sigma} \subseteq \den{e\sigma'}$, for any $e \in LExp_\perp$ and $\sigma, \sigma' \in LSubst_\perp$.
\end{proposition}
%
The limitations of Theorem \ref{thCompoCrwllet} make us yearn for another semantic notion for let-expressions with a better compositional behaviour. We have already seen that the problem with  $\crwllet$ is the possible loss of definientia when extracting an expression from its context. But in fact what bound variables need is access to the \emph{values} of their corresponding definientia, as it is done in the rule \crule{Let} where the value of the definiens is transmited to the body of the let-binding by applying a c-substitution replacing the bound variable by that value.
With these ideas in mind we define the stronger notion of  \emph{hyperdenotation} (sometimes we say \emph{hypersemantics}), which gives a more active role to variables in expressions:
in contrast to the denotation  of an expression $e$, which is a set of
c-terms, its hyperdenotation $\denn{e}$ is a function mapping c-substitutions to denotations, i.e., to sets of c-terms.

\begin{definition}[Hyperdenotation]\label{hyper}~
The  hyperdenotation of an expression $e \in LExp_{\perp}$ under a program $\prog$ is a function $\denn{e}^\prog : CSubst_\perp \rightarrow \ds$ defined by $\denn{e}^\prog\ \theta = \den{e\theta}^\prog$.
\end{definition}

As usual, in most cases we will omit the mention to $\prog$.
We will use sometimes $\hds$ as an alias for $CSubst_\perp \rightarrow \ds$, i.e, for the kind of objects that are hyperdenotations of expressions.


The notion of hyperdenotation is strictly more powerful than the notion of $\crwllet$\ denotation. Equality
of hyperdenotations implies equality of denotations ---because if $\denn{e} = \denn{e'}$ then $\den{e} = \den{e\epsilon} = \denn{e}\epsilon = \denn{e'}\epsilon = \den{e'\epsilon} = \den{e'}$---
but the opposite does not hold: consider the program $\{f(0) \tor 1\}$ and the expressions $f(X)$ and $\perp$; they have the same denotation (the set $\{\perp\}$) but different hyperdenotations, as $\denn{\perp}[X/0] \not\ni 1 \in \denn{f(X)}[X/0]$.
Hypersemantics are useful to characterize the meaning of expressions present in a context in which some of its variables
may get bound, like in the body of a let-binding or in the right hand side of a program rule.
Therefore are useful to reason about expressions put in arbitrary contexts, in which let-bindings may freely appear.

Most remarkably, hyperdenotations allow to recover strong compositionality
results for let-expressions similar to Theorems \ref{thCompoCrwl} and
\ref{thCompoCrwlbis}. We find it more intuitive to start the analog to the latter. Semantics of contexts were defined as denotation transformers (Definition \ref{def:denContxt}).
Analogously, the hypersemantics $\denn{\con}$ of a context $\con$ is a hyperdenotation transformer defined as follows:

\begin{definition}[Hypersemantics of a context]
Given $\con \in Cntxt$, its hyperdenotation is a function $\denn{\con}:\hds \rightarrow \hds$ defined by induction over the structure of $\con$ as follows:
\begin{itemize}
 \item $\denn{[]} \hde\theta = \hde\theta$
 \item $\denn{h(e_1, \ldots, \con, \ldots, e_n)}\hde\theta = \bigcup\limits_{t \in \denn{\con}\hde\theta} \den{h(e_1\theta, \ldots, t, \ldots, e_n\theta)}$
 \item $\denn{let~X=\con~in~e}\hde\theta = \bigcup\limits_{t \in \denn{\con}\hde\theta} \den{let~X=t~in~e\theta}$
 \item $\denn{let~X=e~in~\con}\hde\theta = \bigcup\limits_{t \in \denn{e}\theta} \denn{\con}\hde(\theta[X/t])$
\end{itemize}
\end{definition}

With this notion, our first version of strong compositionality
for hypersemantics looks like Theorem \ref{thCompoCrwlbis}.

\begin{theorem}[Compositionality of hypersemantics]\label{CompHipSem}
For all $\con \in Cntxt$, $e \in LExp_\perp$
$$\denn{\con[e]} = \denn{\con}\denn{e}$$
As a consequence: $\denn{e} = \denn{e'} \Leftrightarrow \forall \con \in Cntxt. \denn{\con[e]} = \denn{\con[e']}$.
\end{theorem}

This result implies that in any context we can replace any  subexpression by another one having
the same hypersemantics (and therefore also the same semantics) without changing the
hypersemantics (hence the semantics) of the global expression.

In Theorems \ref{thCompoCrwlbis} and \ref{CompHipSem} the role of call-time choice is hidden in the definition of semantics and hypersemantics of a context, respectively.
To obtain a version of strong compositionalty of hypersemantics closer to Theorem \ref{thCompoCrwl} and \ref{thCompoCrwllet}, we need some more notions and notations about hyperdenotations or, more generally, about functions in $\hds$.
Since they are set-valued functions, many usual set operations and relations can be lifted naturally in a pointwise manner to $\hds$. The precise definitions become indeed clearer if we give them for general sets, abstracting away the details about $\hds$. We introduce also some notions about decomposing set-valued functions that will be useful for hyperdenotations.
We use freely $\lambda$-notation  to write down a function in the mathematical sense; we may write $\lambda x\in A$ to indicate its domain $A$, if it not clear by the context.

\begin{definition}[Operations and relations for set-valued functions]\label{def:setFunOp}
Let  $A,B$ be two sets,  ${\cal F}$  the set of functions $A \rightarrow\partes{B}$,  and $f,g\in {\cal F}$. Then:
\begin{itemize}
\item[i)]  The \emph{hyperunion} of $f,g$ is defined as $f \uhs g = \lambda x\in A. f(x) \cup g(x)$.
\item[ii)] More generally, the \emph{hyperunion of a family}  ${\cal I} \subseteq {\cal F}$,
written indistinctly as  $\Uhs {\cal I}$ or $\Uhss{f\in {\cal I}}f$, is defined as
\[ \Uhs {\cal I} \equiv\! \Uhss{f\in {\cal I}}f ~=_{\it def}~ \lambda x\in A.\bigcup_{f\in {\cal I}}f(x) \]
Notice that $f \uhs g = \Uhs\{f,g\}$.
\item[iii)] We say that $f$ is \emph{hyperincluded} in $g$, written $f \ohs g$, iff $\forall x\in A.f(x) \subseteq g(x)$.
\item[iv)] A \emph{decomposition} of $f$ is any ${\cal I} \subseteq {\cal F}$ such that $\Uhs {\cal I} = f$.
\item[v)] The \emph{elemental decomposition} of $f$ is the following set of functions of ${\cal F}$:
\[\sd{f} = \{\lambda x\in A.\left\{\begin{array}{l}\{b\} \mbox{~if~} x=a\\\emptyset \mbox{~otherwise}\end{array}\right.~\mid~a\in A,b\in f(a)\}\]
Or, using the abbreviation $\hl a.\{b\}$ as a shorthand for $\lambda x.\left\{\begin{array}{l}\{b\} \mbox{~if~} x=a\\\emptyset \mbox{~otherwise}\end{array}\right.$,
\[\sd{f} = \{\hl a.\{b\}~\mid~a\in A,b\in f(a)\}\]
\end{itemize}
\end{definition}
Decompositions are used to split set-valued functions into  smaller pieces; elemental decompositions do it with minimal ones.
For instance, if $f:\{a,b\} \rightarrow \partes{\{0,1,2\}}$ is given by $f(a)\! =\! \{0,2\}$ and $f(b)\! =\! \{1,2\}$, then
$\sd{f}\!=\! \{\hl a. \{0\},\hl a. \{2\},\hl b. \{1\},\hl b. \{2\}\}$.

Hyperinclusion and hyperunion share many properties of standard set inclusion and union. Some of them are collected in the next result, that refer also to decompositions:

\begin{proposition}\label{HipSemDecUnion}
Consider two sets $A,B$, and let ${\cal F}$ be the set of functions $A \rightarrow\partes{B}$. Then:
\begin{enumerate}
    \item[i)] $\ohs$ is indeed a partial order on ${\cal F}$, and $\sd{f}$ is indeed a decomposition of $f\in {\cal F}$, i.e., $\Uhs{(\sd{f})} = f$.
    \item[ii)] Monotonicity of hyperunion wrt. inclusion: for any ${\cal I}_1,{\cal I}_2 \subseteq {\cal F}$
$$
{\cal I}_1 \subseteq {\cal I}_2 \mbox{ implies } \Uhs{\cal I}_1 \ohs \Uhs{\cal I}_2
$$
    \item[iii)] Distribution of unions: for any ${\cal I}_1,{\cal I}_2 \subseteq {\cal F}$
$$
\Uhs{({\cal I}_1 \cup {\cal I}_2)} = (\Uhs{\cal I}_1) \uhs (\Uhs{\cal I}_2)
$$
    \item[iv)] Monotonicity of decomposition wrt. hyperinclusion: for any $f_1, f_2 \in {\cal F}$
$$
f_1 \ohs f_2 \mbox{ implies } \sd{f_1} \subseteq \sd{f_2}
$$
\end{enumerate}
\end{proposition}

We will apply all these notions, notations and properties  to the case when $A\equiv CSubst_\perp$ and $B\equiv CTerm_\perp$ (i.e. $\partes{B}\equiv \ds$ and therefore ${\cal F} ~\equiv~\hds$). Therefore, we can speak of the hyperunion of two hyperdenotations, or of a family of them, we can elementarily decompose a hyperdenotation, etc.

\begin{proposition}[Distributivity under context of hypersemantics unions]\label{HipSemDistCntx}
$$
\denn{\con}(\Uhs{\hdes}) = \Uhss{\hde\!\in\!\hdes}\denn{\con}\hde
$$

\end{proposition}

With this result we can  easily prove our desired new version of a strong compositionality result for hypersemantics, with a style closer to the formulations of Theorems \ref{thCompoCrwl} and \ref{thCompoCrwllet}.
This new form of compositionality will be used in Section \ref{bubbling} for building a straightforward proof of the adequacy of a transformation that otherwise becomes highly involved by using other techniques.

\begin{theorem}[Compositionality of hypersemantics, version 2]\label{CompHipSemDos}
For any $\con \in Cntxt$, $e \in LExp_\perp$:
$$
 \denn{\con[e]} = \Uhss{\hde\!\in\!\hdes}\denn{\con}\hde \mbox{, for any decomposition $\hdes$ of $\denn{e}$}
$$
In particular:
$
\denn{\con[e]} = \Uhss{\hde\!\in\!\sd{\denn{e}}}\denn{\con}\hde.
$

\noindent As a consequence: $\denn{e} = \denn{e'} \Leftrightarrow \forall \con \in Cntxt. \denn{\con[e]} = \denn{\con[e']}$.
\end{theorem}
\begin{proof}
$$
\begin{array}{ll}
\denn{\con[e]} = \denn{\con}\denn{e}                 & \mbox{ by compositionality \emph{v.1} (Theorem \ref{CompHipSem})} \\
               = \denn{\con}(\Uhs{\hdes})  & \mbox{ by definition of decomposition (Def. \ref{def:setFunOp} {\it iv)}}\\
= \Uhss{\hde\!\in\!\hdes}\denn{\con}\hde & \mbox{ by distributivity (Proposition \ref{HipSemDistCntx})}
\end{array}
$$
\end{proof}


As happened with Theorems \ref{thCompoCrwl} and \ref{thCompoCrwlbis} with respect to denotations, Theorems \ref{CompHipSem} and \ref{CompHipSemDos} are different aspects of the same property, which shows that the hypersemantics of a whole let-expression depends only on the hypersemantics of its constituents; it also allows us to interchange in a context any pair of expressions with the same hypersemantics. This is reflected on the fact  that we have attached $\denn{e} = \denn{e'} \Leftrightarrow \forall \con \in Cntxt. \denn{\con[e]} = \denn{\con[e']}$ as a trivial consequence both in Theorem \ref{CompHipSem} and Theorem \ref{CompHipSemDos}. Moreover, Theorem \ref{CompHipSem} can also be proved by a combination of Theorem \ref{CompHipSemDos} and Propositions \ref{HipSemDecUnion} {\it i)} and \ref{HipSemDistCntx}, in a similar way to the proof for Theorem \ref{CompHipSemDos} above.
$$
\begin{array}{ll}
\denn{\con[e]} = \Uhss{\hde\!\in\!\hdes}\denn{\con}\hde & \mbox{ by compositionality \emph{v.2} (Theorem \ref{CompHipSemDos})} \\
= \denn{\con}(\Uhs{(\sd{\denn{e}})}) & \mbox{ by distributivity (Proposition \ref{HipSemDistCntx})} \\
= \denn{\con}\denn{e}  & \mbox{ because $\sd{\denn{e}}$ decomposes $e$ (Proposition \ref{HipSemDecUnion} {\it i)})} \\
\end{array}
$$
Therefore Theorems \ref{CompHipSem} and \ref{CompHipSemDos} are results with  the same strength, two sides of the same coin that will be useful tools for reasoning with hypersemantics.


~\\
To conclude, we present the following  monotonicity property under contexts of hypersemantics, which will be useful in the next section.
\begin{lemma}[Monotonicity under contexts of hypersemantics]\label{LemMonConHipSemNu} For any $\con \in Cntxt, \hde_1,\hde_2 \in \hds$:
$$\hde_1 \ohs \hde_2 \mbox{ implies that } \denn{\con}\hde_1 \ohs \denn{\con}\hde_2$$
\end{lemma}
\begin{proof}
Assume $\hde_1 \ohs \hde_2$. Then:
$$
\begin{array}{ll}
\denn{\con}\hde_1 = \denn{\con}(\Uhs{(\sd{\hde_1})}) & \mbox{ by Proposition \ref{HipSemDecUnion} {\it i)}}
\\
= \denn{\con}(\Uhs{\{\hl\mu.\{t\}~|~\mu \in CSubst_\perp, t \in \hde_1\mu\}}) & \mbox{ by definition of $\sd{}$}
\\
\ohs \denn{\con}(\Uhs{\{\hl\mu.\{t\}~|~\mu \in CSubst_\perp, t \in \hde_2\mu\}}) & \mbox{ by Proposition \ref{HipSemDecUnion} {\it ii)}} \\
= \denn{\con}(\Uhs{(\sd{\hde_2})}) & \mbox{ by definition of $\sd{}$}
\\
= \denn{\con}\hde_2 & \mbox{ by Proposition \ref{HipSemDecUnion} {\it i)}} \\
\end{array}
$$
\end{proof}

We have now the tools needed to tackle the task of formally relating \crwl\ and let-rewriting.


\subsection{Equivalence of let-rewriting to \crwl\ and {\it CRWL}$_{let}$}\label{sectEqLetrwCrwlLet}\label{equivalence}

In this section we  prove  soundness and completeness results of let-rewriting  with respect to $\crwllet$ and \crwl.

\subsubsection{Soundness}\label{SectSoundLetRw}

Concerning  soundness we want to prove that $\f$-steps do not create new {CRWL}-semantic values. More precisely:

\begin{theorem}[Soundness of let-rewriting]\label{T14}\label{C1}
For all $e, e' \in LExp$, if $e \fe e'$ then $\den{e'} \subseteq \den{e}$.
\end{theorem}

Notice that because of non-determinism
$\subseteq$ cannot be replaced by $=$ in this theorem. For example, with the program $\prog = \{coin \tor 0, coin \tor 1\}$ we can perform the step $coin \f 0$, for which $\den{0} = \{0, \perp\}$, $\den{coin} =  \{0, 1, \perp\}$. 

~\\
\noindent It is interesting to explain why a direct reasoning with denotations fails to prove Theorem \ref{T14}.

A proof could proceed straightforwardly by a case distinction on the rules for $\f$ to prove the soundness of a single $\f$ step. The problem is that the case for a \crule{Contx} step would need the following monotonicity property under context of $\crwllet$ denotations:
$$
\den{e} \subseteq \den{e'} \mbox{ implies } \den{{\cal C}[e]} \subseteq \den{{\cal C}[e']}
$$
Unfortunately, the property is false, for the same reasons that already explained the weakness of Theorem \ref{thCompoCrwllet}: the possible capture of variables when switching from $e$ to $\con[e]$.

\begin{counterexample}
Consider the program $\prog = \{f(0)\tor 1\}$. We have $\den{f(X)} = \{\bot\} \subseteq \{\bot,0\} = \den{0}$, but when these expressions are placed within the context $let~X=0~in~[\ ]$ we obtain $\den{let~X=0~in~f(X)} = \{\bot,1\} \not\subseteq \{\bot,0\} = \den{let~X=0~in~0}$.
\end{counterexample}

The good thing is that we can overcome these problems by using hypersemantics.
Theorem \ref{T14} will be indeed an easy corollary of the following generalization to hypersemantics.
\begin{theorem}[Hyper-Soundness of let-rewriting]\label{T27}
For all $e, e' \in LExp$, if $e \fe e'$ then $\denn{e'} \ohs \denn{e}$.
\end{theorem}

And, in order to prove this generalized theorem, we also devise a generalization of the faulty monotonicity property of $\crwllet$\ denotations above mentioned. That generalization is an easy consequence of the compositionality and monotonicity under contexts of hypersemantics. 
\begin{lemma}\label{T25}
For all $e, e' \in LExp_{\perp}$ and $\con \in Cntxt$, if $\denn{e} \ohs \denn{e'}$ then $\denn{\con[e]} \ohs \denn{\con[e']}$.
\end{lemma}
\begin{proof}
$$
\begin{array}{ll}
\denn{\con[e]} = \denn{\con}\denn{e} & \mbox{ by Theorem \ref{CompHipSem}} \\
\ohs \denn{\con}\denn{e'} & \mbox{ by Lemma \ref{LemMonConHipSemNu}, as $\denn{e} \ohs \denn{e'}$} \\
= \denn{\con[e']}  & \mbox{ by Theorem \ref{CompHipSem}} \\
\end{array}
$$
\end{proof}

With the help of Lemma \ref{T25}, we  can now prove Theorem \ref{T27} by a simple case distinction on the rules for $\f$ and a trivial induction on the length of the derivation. Now, Theorem \ref{T14} follows as an easy consequence.
\begin{proof}[Proof for Theorem \ref{T14}]
Assume $e \fe e'$. By Theorem \ref{T27} we have $\denn{e'} \ohs \denn{e}$, and therefore $\den{e'\theta} \subseteq \den{e\theta}$ for every $\theta \in CSubst_\perp$. Choosing $\theta = \epsilon$ (the empty substitution) we obtain $\den{e'} \subseteq \den{e}$ as desired.
\end{proof}

The moral then is that \emph{when reasoning about the semantics of expressions and programs with {\it lets} it is usually better to lift the problem to the hypersemantic world}, and then particularize to semantics the obtained result.
This is done, for instance, in the following result:

\begin{proposition}[The $\fnf$ relation preserves hyperdenotation]
\label{propFnfPreservHipSem}~
For all $e, e' \in LExp$, if $e ~\fnfe~ e'$ then $\denn{e} = \denn{e'}$---and therefore $\den{e} = \den{e'}$.
\end{proposition}

This result mirrors semantically the fact that $\fnf$ performs transitions between let-expressions corresponding to the same implicit term graph. Proposition \ref{propFnfPreservHipSem} in some sense lessens the importance of the lack of confluence for the $\fnf$ relation seen in Section \ref{sect:letRwRelation}.
Preservation of hyperdenotation may be used in some situations as a substitute for confluence, specially taking into account that let-rewriting and $\crwllet$ enjoy a really strong equivalence, as it is shown in this section.


~\\
Finally, we combine the previous results in order to get our main result concerning the soundness of let-rewriting with respect to the $\crwllet$ 
 calculus:
\kk
\begin{theorem}[Soundness of let-rewriting]
\label{LSound}
For any program $\prog$ and $e \in LExp$ we have:
\begin{enumerate}
\item[i)] $e \f^* e'$ implies $\gl e \crwlto |e'|$, for any  $e' \in LExp$.
 \item[ii)] $e \f^* t$ implies $\gl e \crwlto t$, for any  $t \in CTerm$.
\end{enumerate}
Furthermore, if neither $\prog$ nor $e$  have lets then we also have:
\begin{enumerate}
\item[iii)] $e \f^* e'$ implies $\cl e \crwlto |e'|$, for any  $e' \in LExp$.
\item[iv)] $e \f^* t$ implies $\cl e \crwlto t$, for any  $t \in CTerm$.
\end{enumerate}
\end{theorem}
\begin{proof}
\begin{enumerate}
 \item[i)] Assume $e \f^* e'$. Then, by 
Theorem \ref{C1} we have $\dcl{e'} \subseteq
\dcl{e}$. Since $|e'| \in \dcl{e'}$ by Lemma \ref{lemmashells}, we get $|e'| \in \dcl{e}$,
which means 
$ e\crwlto |e'|$.
 \item[ii)] Trivial by {\it (i)}, since $|t| = t$ for any $t\in CTerm$.
 \item[iii)] Just combining {\it i)} and Theorem \ref{thEquivCrwlCrwllet}.
 \item[iv)] Just combining {\it ii)} and Theorem \ref{thEquivCrwlCrwllet}.
\end{enumerate}
\end{proof}
\ekk

\subsubsection{Completeness}\label{SectCompLetRw}
Now we look for the reverse implication of Theorem \ref{LSound}, that is, the completeness of let-rewriting as its ability to compute, for any given expression, any value that can been computed by the \crwl-calculi. 
With the aid of the Peeling Lemma \ref{T36}
we can  prove the following strong completeness result for let-rewriting, which still has a certain technical nature.

\begin{lemma}[Completeness lemma for let-rewriting]
\label{T32}
For any $e \in LExp$ and $t \in CTerm_{\perp}$ such that $t \not\equiv \perp$,
$$
e \clto t \mbox{ implies } e \fe let~\overline{X = a}~in~t'
$$
for some $t' \in CTerm$ and  $\overline{a} \subseteq LExp$ such that $t\sqsubseteq |let~\overline{X=a}~in~t'|$ and $|a_i| = \perp$ for every $a_i \in \overline{a}$. As a consequence, $t \sqsubseteq t'[\overline{X/\perp}]$.
\end{lemma}

Note the condition $t \not\equiv \perp$ is essential for this lemma to be true, as we can see by taking $\prog = \{loop \tor loop\}$ and $e \equiv loop$: while $loop \clto \perp$, the only $LExp$ reachable from $loop$ is $loop$ itself.\\

Our main result concerning  completeness of let-rewriting follows easily from
Lemma \ref{T32}. It shows that any c-term computed by \crwl\ or $\crwllet$ for
an expression can be refined by a let-rewriting derivation; moreover, if the c-term is total, then it can be exactly reached by let-rewriting.

\begin{theorem}[Completeness of let-rewriting]\label{LComp}\label{LCompTot}
For any program $\prog$, $e \in LExp$, and $t \in CTerm_{\perp}$ we have:
\begin{enumerate}
    \item[i)] $\gl e \crwlto t$ implies $e \f^* e'$ for some $e' \in LExp$ such that $t \sqsubseteq |e'|$
        \item[ii)] Besides, if $t \in CTerm$ then $\gl e \crwlto t$ implies $e \f^* t$
\end{enumerate}
Furthermore, if neither $\prog$ nor $e$ have lets then we also have
\begin{enumerate}
    \item[iii)] $\cl e \crwlto t$ implies $e \f^* e'$ for some $e' \in LExp$ such that $t \sqsubseteq |e'|$
        \item[iv)] Besides, if $t \in CTerm$ then $\cl e \crwlto t$ implies $e \f^* t$
\end{enumerate}

\end{theorem}
\begin{proof}
Regarding  part \emph{i)}, if $t \equiv \perp$ then we are done with $e \f^0 e$ as $\forall e, \perp \sqsubseteq |e|$. On the other hand, if $t \not\equiv \perp$ then by Lemma \ref{T32} we have $e \f^* let~\overline{X}=\overline{a}~in~t'$ such  that $t \sqsubseteq |let~\overline{X}=\overline{a}~in~t'|$.

To prove part \emph{ii)}, assume $\gl e \crwlto t$. Then, by Lemma \ref{T32}, we get  $e \f^* let~\overline{X}=\overline{a}~in~t'$ such that $t \sqsubseteq |let~\overline{X}=\overline{a}~in~t|\equiv t'[\overline{X}/\overline{\perp}] $, for some $t' \in CTerm, \overline{a} \subseteq LExp$. As $t \in CTerm$ then $t$ is maximal wrt. $\sqsubseteq$, so
$t \sqsubseteq  t'[\overline{X}/\overline{\perp}]$ implies $t'[\overline{X}/\overline{\perp}] \equiv t$, but then $t'[\overline{X}/\overline{\perp}] \in CTerm$ so it must happen that $FV(t') \cap \overline{X} = \emptyset$ and therefore $t' \equiv t'[\overline{X}/\overline{\perp}] \equiv t$. But then $let~\overline{X}=\overline{a}~in~t' \f^* t' \equiv t$ by zero or more steps of \lrrule{Elim}, so $e \f^* let~\overline{X}=\overline{a}~in~t' \f^* t$, that is $e \f^* t$.

Finally, parts \emph{ii)} and \emph{iv)} follow from  \emph{ii)}, \emph{iii)} and Theorem \ref{thEquivCrwlCrwllet}.
\end{proof}





As an immediate corollary of this completeness result and soundness (Theorem
\ref{LSound}), we obtain the following result relating let-rewriting to \crwl\
and \crwll\ for total c-terms, which gives a clean and easy way to understand
the formulation of the adequacy of let-rewriting.

\begin{corollary}[Equivalence  of $\crwllet$ and let-rewriting for total values]\label{LCorrTot}
For any program $\prog$, $e \in LExp$, and $t \in CTerm$ we have
\begin{center}
$\gl e \crwlto t \mbox{ iff } e \f^* t$.
\end{center}

\noindent Besides if neither $\prog$ nor $e$ have lets then we also have

\begin{center}
$\cl e \crwlto t \mbox{ iff } e \f^* t$.
\end{center}
\end{corollary}

As final consequence  of Theorems \ref{LSound} and \ref{LComp}  we obtain another strong equivalence result for both formalisms, this time expressed in terms of semantics and hypersemantics.

\begin{theorem}[Equivalence  of $\crwllet$ and let-rewriting]\label{TheorEquivLetRwDown}
For any program $\prog$ and $e \in LExp$:
\begin{enumerate}
    \item[i)] $\den{e} = \{|e'|~|~ e \fe e'\}\!\!\downarrow $
    \item[ii)] $\denn{e} = \lambda \theta\in CSubst_\bot.(\{|e'|~|~ e \fe e'\}\!\!\downarrow)$
\end{enumerate}
where $\downarrow $ is the downward closure operator defined in Lemma \ref{lemmashells}.
\end{theorem}
\begin{proof}
\begin{enumerate}
    \item[i)] We prove both inclusions. Regarding $\den{e} \subseteq \{|e'|~|~ e \fe e'\}\!\!\downarrow$, assume $t \in \den{e}$. By Theorem \ref{LComp} there must exist some $e' \in LExp$ such that $e \fe e'$ and $t \ordap |e'|$, therefore $|e'| \in \{|e'|~|~ e \fe e'\}$. But this, combined with $t \ordap |e'|$, results in $t \in \{|e'|~|~ e \fe e'\}\!\!\downarrow$.

Regarding the other inclusion, consider some $t \in \{|e'|~|~ e \fe e'\}\!\!\downarrow$. By definition of the $\downarrow$ operator, there must exist some $e' \in LExp$ such that $t \ordap |e'|$ and $e \fe e'$. But that implies $|e'| \in \den{e}$, by Theorem \ref{LSound}, which combined with $t \ordap |e'|$ and the polarity property (Proposition \ref{propCrwlletPolar}) gives us that $t \in \den{e}$.

    \item[ii)] Trivial by applying the previous item and the definition of hypersemantics of an expression.
\end{enumerate}
\end{proof}

\section{Semantic reasoning}\label{SemEqs}
Having equivalent  notions of semantics and
reduction allows to reason interchangeably at the rewriting and  semantic levels.
In this section we show the power of such technique in different situations.
We start with a concrete example, adapted from \cite{LRSrta09},  where semantic reasoning leads easily to conclusions non-trivially achievable when thinking directly in operational terms.
\begin{example}
Imagine a program using constructors $a,b \in \ctrs^0,c\in \ctrs^1,d\in \ctrs^2$ and defining a function $f \in FS^1$ for which we know that $f(a)$ can be let-rewritten  to $c(a)$ and $c(b)$ but no other c-terms. Consider also an expression $e$ having $f(a)$ as subexpression, i.e., $e$ has the shape $\con[f(a)]$.
We are interested now in the following question: can we safely replace in $e$ the subexpression $f(a)$ by any other ground expression $e'$ let-reducible to the same set of values\footnote{More precisely, to the same set of shells in the sense of Theorem \ref{TheorEquivLetRwDown} part $i)$.}? By safely we mean not changing the values reachable from $e$.

The question is less trivial than it could appear. For instance, if reductions were made with term rewriting instead of let-rewriting ---i.e., considering run-time instead of call-time choice--- the answer is negative \cite{LRSrta09}. To see that, consider the program
$$
\begin{array}{lll}
f(a) \tor c(a) & g \tor a  & h(c(X)) \tor d(X,X)
\\
f(a) \tor c(b) & g \tor b  &
\end{array}
$$
and the expressions  $e \equiv h(f(a))$  and $e' \equiv c(g)$.
All this is compatible with the assumptions of our problem. However, $e$ is reducible by term rewriting
only to $d(a,a)$ and $d(b,b)$, while replacing $f(a)$ by $e'$ in $e$ gives $h(c(g))$, which is reducible by term rewriting
to two additional values, $d(a,b)$ and $d(b,a)$; thus, the replacement of $f(a)$ by $e'$ has been unsafe.

However, the answer to our question is affirmative in general for let-rewriting,  as it is very easily proved by a semantic reasoning using compositionality of \crwll: the assumption on $f(a)$ and $e'$ means that they have the same denotation
$\den{f(a)} = \den{e'} = \{c(a),c(b)\}\downarrow$ and, since they are ground, the same hyperdenotation
$\denn{f(a)} = \denn{e'} = \lambda\theta.\{c(a),c(b)\}\downarrow$. By compositionality of hypersemantics,  $\con[f(a)]$ and $\con[e']$ have the same (hyper)denotation, too. By equivalence of \crwll\ and let-rewriting this implies that both expressions reach the same value by let-rewriting.

Despite its simplicity, the example raises naturally interesting questions about replaceability,  for which semantic methods could be simpler than direct reasonings about reduction sequences. This is connected to the \emph{full abstraction} problem that we have investigated for run-time and call-time choice in \cite{LRSrta09,LR10}.
\end{example}


Semantic methods can be also used to prove the correctness of new  operational rules  not directly provided by our set of let-rewriting rules. Such  rules can be useful for different purposes: to make computations simpler, for program transformations, to obtain new properties of the framework, \ldots
Consider for instance the following generalization of the (LetIn) rule in Figure \ref{letrcalc}:
\begin{center}
\textbf{(CLetIn)}~~$\con[e] \f let~X=e~in~\con[X]$,~~ if $BV(\con) \cap FV(e) = \emptyset$ and $X$ is fresh
\end{center}
This rule allows to create let-bindings in more situations and to put them in outer positions than the original (LetIn) rule. If we have not considered it in the definition of let-rewriting is because it would destroy the strong termination property of Proposition \ref{termlr}, as it is easy to see.
However, this rule may shorten derivations. For instance, the derivation in Example \ref{fig:letDer} could be shortened to:

$$
\begin{small}
\begin{array}{ll}
\underline{heads(repeat(coin))} & \crule{CLetIn}  \\
\f let~C=coin~in~heads(\underline{repeat(C)}) & \crule{Fapp} \\
\f let~C=coin~in~heads(C:\underline{repeat(C)}) & \crule{Fapp} \\
\f let~C=coin~in~\underline{heads(C:C:repeat(C))} & \crule{CLetIn} \\
\f let~C=coin~in~let~X=repeat(C)~in~\underline{heads(C:C:X)} & \crule{Fapp} \\
\f let~C=coin~in~\underline{let~X=repeat(C)~in~(C, C)}  & \crule{Elim} \\
\f let~C=\underline{coin}~in~(C, C)  & \crule{Fapp} \\
\f \underline{let~C=0~in~(C,C)} & \crule{Bind}\\
\f (0,0)
\end{array}
\end{small}
$$

Reasoning the correctness of (CLetIn) rule is not difficult by means of semantic methods. We only need to prove that the rule preserves hypersemantics. 

\begin{lemma}\label{lemCLetInPreserv}
If $BV(\con) \cap FV(e) = \emptyset$ and $X$ is fresh, then
$
\denn{\con[e]} = \denn{let~X=e~in~\con[X]}
$.
\end{lemma}
\begin{proof}
Assume an arbitrary $\theta \in CSubst_{\perp}$:
$$
\begin{array}{ll}
\denn{let~X = e~in~\con[X]}\theta = \den{(let~X = e~in~\con[X])\theta} \\
= \den{let~X = e\theta~in~\con\theta[X]} & \mbox{as $X$ is fresh} \\
= \bigcup\limits_{t \in \den{e\theta}}\den{(\con\theta[X])[X/t]} & \mbox{by Theorem  \ref{lem:weakcomp}} \\
= \bigcup\limits_{t \in \den{e\theta}}\den{\con\theta[t]} & \mbox{as $X$ is fresh} \\
= \den{\con\theta[e \theta]} & \mbox{by Theorem  \ref{lem:weakcomp} } \\ 
= \den{(\con[e])\theta} = \denn{\con[e]}\theta 
\end{array}
$$
\end{proof}

The rule (CLetIn) is indeed used in some of the proofs in \nuevo{the online appendix}, together with another derived rule:
\begin{center}
\textbf{(Dist)}~~$\con[let~X = e_1~in~e_2] \f let~X =
e_1~in~\con[e_2]$,\\
~~~~~~ if $BV(\con) \cap FV(e_1) = \emptyset$ and $X \not\in
FV(\con)$
\end{center}
which also preserves hypersemantics:
\begin{lemma}\label{lDistHD}
If $BV(\con) \cap FV(e_1) = \emptyset$ and $X \not\in
FV(\con)$ then $\denn{\con[let~X = e_1~in~e_2]} = \denn{let~X =
e_1~in~\con[e_2]}$.
\end{lemma}

These ideas can be made more general. Consider the equivalence relation $e_1 \eqehs e_2$ iff $\denn{e_1} = \denn{e_2}$. This relation is especially relevant because $e_1 \eqehs e_2$ iff  $\forall \con \in Cntxt. \denn{\con[e]} = \denn{\con[e']}$, by Theorem \ref{CompHipSem}. We can contemplate $\eqehs$ as an abstract, although non-effective, reduction relation, of which the relations $\fnf$ of Section \ref{let-rewriting} and the rules (CLetIn) and (Dist) are particular subrelations.
It is trivial to check that, by construction, the combined relation $\f \cup \eqehs$ is sound and complete \wrt\  $\crwllet$. We can use that relation to reason about the meaning or equivalence of let-expressions and programs. We  could also employ it in the definition of on-demand evaluation strategies for let-rewriting. As any subrelation of $\f \cup \eqehs$ is sound \wrt\  $\crwllet$, an approach to strategies for let-rewriting could consist in defining a suitable operationally effective subrelation of $\f \cup \eqehs$ and then proving its completeness and optimality  (if it is the case).


\subsection{A case study: correctness of bubbling}\label{bubbling}
\nc{\equiva}{\sim}
\nc{\hden}[1]{[\![\![#1]\!]\!]}
\nc{\fs}{\ra^{l^*}}
\nc{\rac}{\ra_{CRWL}}
We develop here another nice application of the `semantic route', where let-rewriting provides a good level of abstraction to formulate a new operational rule ---{\em bubbling}--- while the semantic point of view is appropriate for proving its correctness.

Bubbling, proposed in \cite{AntoyBrownChiangENTCS06}, is an operational rule
devised to improve the efficiency of functional logic computations. Its correctness was
formally studied in \cite{AntoyBrownChiang06RTA} in the framework
of a variant \cite{EchahedJanodet98JICSLP} of term graph rewriting.
%
The idea  of bubbling  is to concentrate all non-determinism of a system
into a \emph{choice} operation $?$ defined by the rules
$X~?~Y \tor X$ and $X~?~Y \tor Y$,
and to lift applications of  $?$ out of their surrounding context,
as illustrated by the following graph transformation taken from \cite{AntoyBrownChiang06RTA}:

\begin{center}
  \includegraphics[scale=0.22]{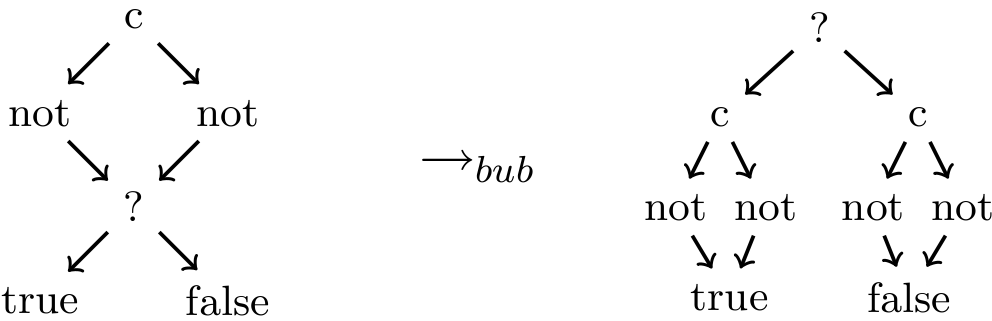}
\end{center}

As it is shown in \cite{AntoyBrownChiangENTCS06}, bubbling can be implemented in such a way
that many functional logic programs become more efficient, but we will not deal with these issues here.

Due to the technical particularities of term graph rewriting, not only the proof of correctness,
but even the definition of bubbling in \cite{AntoyBrownChiangENTCS06,AntoyBrownChiang06RTA}
are involved and need subtle care concerning the appropriate
contexts over which choices can be bubbled.
In contrast, bubbling can be expressed within
our framework in a remarkably easy and abstract way
as a new rewriting rule:
\begin{center}
{\bf (Bub)}~~ $\cnn{e_1 ? e_2} \ra^{bub} \cnn{e_1} ? \cnn{e_2}$, for $e_1,e_2 \in LExp$
\end{center}

With this rule, the bubbling step corresponding to the graph transformation of
the example above is:
$$
\begin{array}{l}
let~X=true~?~false~in~c(not(X), not(X)) \rw^{bub} \\
let~X=true~in~c(not(X), not(X))~?~let~X=false~in~c(not(X), not(X))
\end{array}
$$

Notice that the effect of this bubbling step is not a shortening of any existing
let-rewriting derivation; bubbling is indeed a genuine new rule, the
correctness of which must be therefore subject of proof. Call-time choice is
essential, since bubbling is not correct with respect to ordinary term rewriting, i.e., run-time choice. 
\begin{counterexample}[Incorrectness of bubbling for run-time choice]
Consider a function $pair$ defined by the rule $pair(X) \tor c(X,X)$ and the expression $pair(0~?~1)$ for $c \in CS^2$ and $0,1 \in CS^0$. Under term rewriting/run-time choice the derivation
$$
\begin{array}{l}
pair(0~?~1) \rw c(0~?~1, 0~?~1) \rw c(0, 0~?~1) \rw c(0, 1)
\end{array}
$$
is valid. But if we performed the bubbling step $$pair(0~?~1) \rw^{bub} pair(0)~?~pair(1)$$ then the c-term $c(0, 1)$ would not be reachable anymore by term rewriting from $pair(0)~?~pair(1)$.
\end{counterexample}


Formulating and proving the correctness of bubbling for call-time choice becomes easy by using semantics.
As we did before, we simply prove that bubbling steps preserve hypersemantics.
We need first a basic property of the (hyper)semantics of binary choice $?$.
Its proof stems almost immediately from the rules for $?$ and the definition of \crwl-(hyper)denotation.

\begin{proposition}[(Hyper)semantic properties of $?$]\label{propHypSemPropAlter} For any $e_1, e_2 \in LExp_\perp$
\begin{enumerate}
    \item[i)]
$
\den{e_1~?~e_2} = \den{e_1} \cup \den{e_2}
$
    \item[ii)]
$
\denn{e_1~?~e_2} = \denn{e_1} \uhs \denn{e_2}
$
\end{enumerate}
\end{proposition}
Combining this property with some of the powerful hypersemantic results from Section \ref{crwllet} leads to an appealing proof of the correctness of bubbling.

\begin{theorem}[Correctness of bubbling for call-time choice]\label{hypcorrbub}
If $e \ra^{bub} e'$ then $\denn{e} = \denn{e'}$, for any $e, e' \in LExp$.
\end{theorem}
\begin{proof}
If $e \ra^{bub} e'$ then $e = \con[e_1~?~e_2]$ and $e' = \con[e_1]~?~\con[e_2]$, for some $e_1,e_2$. Then:
$$
\begin{array}{ll}
\!\!\denn{\con[e_1~?~e_2]} = \denn{\con}\denn{e_1~?~e_2} & \mbox{by Theorem \ref{CompHipSem}} \\
= \denn{\con}(\denn{e_1} \uhs \denn{e_2}) & \mbox{by Proposition \ref{propHypSemPropAlter} {\it ii)}} \\
= \denn{\con}\denn{e_1} \uhs \denn{\con}\denn{e_2} & \mbox{by Proposition \ref{HipSemDistCntx}} \\
= \denn{\con[e_1]} \uhs \denn{\con[e_2]} & \mbox{by Theorem \ref{CompHipSem}} \\
= \denn{\con[e_1]~?~\con[e_2]} & \mbox{by Proposition \ref{propHypSemPropAlter} {\it ii)}} \\
\end{array}
$$
\end{proof}

This property was proved also for the HO case in \cite{LRSflops08}. But the proof given here is much more elegant thanks to the new semantic tools developed in Section \ref{crwllet}.




\section{Let-narrowing}\label{let-narrowing}
It is well known that  there are  situations in functional logic computations where rewriting is not enough and must
be lifted to some kind of {\em narrowing}, because the expression being reduced contains variables
for which different bindings might produce different evaluation results. This might happen either
because variables are already present in the initial expression to reduce, or due to the presence
of extra variables in the program rules.
 In the latter case let-rewriting certainly works,
but not in an effective way, since the parameter passing substitution
in the rule \crule{Fapp} of Figure 
\refp{letrcalc} `magically' guesses the appropriate values for those extra variables
(see Example \ref{ex:narrder} below).
Some works \cite{AntoyHanus06,DiosLopez07,BrasselH07} have proved that  guessing can be replaced by a systematic
non-deterministic generation of all (ground) possible values.
However, this does not cover all aspects of narrowing, which is able to produce non-ground
answers, while generator functions are not.
In this section we present \emph{let-narrowing}, a natural lifting of let-rewriting devised to effectively deal with free and extra variables. 

Using the notation of contexts, the standard definition of narrowing as a lifting of term rewriting in ordinary \trss\ is the following:
${\cal C}[f(\overline{t})] \leadsto_\theta {\cal C}\theta[r\theta]$, if $\theta$ is a mgu of
$f(\overline{t})$ and $f(\overline{s})$, where $f(\overline{s}) \rw r$ is a fresh variant of a rule
of the TRS. The requirement that the binding substitution $\theta$ is a mgu can be relaxed
to accomplish with certain narrowing strategies like needed narrowing \cite{AEH00}, which use  unifiers but
not necessarily  most general ones.

This definition of narrowing cannot be directly translated as it is to the case of let-rewriting,
for two reasons. First, binding substitutions must be
c-substitutions, as for the case of let-rewriting.
Second, let-bound variables
should not be narrowed, but their values  should be rather obtained by evaluation
of their binding expressions.
The following example illustrates some of the points above.

\begin{example}\label{ex:narr}
Consider the following program over Peano natural numbers:
$$
\begin{array}{ll}
0+Y \rw Y         &~~~~~~ even(X) \rw {\it if}\ (Y\!+\!Y==X)~then\ true\\
s(X)+Y \rw s(X+Y) &~~~~~~ {\it if}\ true\ then\ Y \rw Y\\
0 == 0 \rw true     &~~~~~~ s(X) == s(Y) \rw X == Y\\
0 == s(Y) \rw false &~~~~~~ s(X) == 0 \rw false\\
coin \rw 0          &~~~~~~ coin \rw s(0)\\
\end{array}
$$
Notice the extra variable $Y$ in the rule for \emph{even}. The evaluation of \emph{even(coin)} by let-rewriting could start as follows:
$$
\begin{array}{l}
even(coin) \f let~X=coin~in~even(X) \\
\f let~X=coin~in~{\it if}~(Y+Y==X)~then~true \\
\f^* let~X=coin~in~let~U=Y+Y~in~let~V=(U==X)~in~{\it if}~V~then~true \\
\f^* let~U=Y+Y~in~let~V=(U==0)~in~{\it if}~V~then~true
\end{array}
$$
Now, because all function applications involve variables, the evaluation cannot continue merely by
rewriting, and therefore narrowing is required instead.
We should not perform standard narrowing steps that bind already let-bound variables; otherwise, the syntax of let-expressions can be lost. For instance, narrowing at \emph{if V then true} generates the binding $[V/true]$ that, if applied  naively to the surrounding context, results in the syntactically illegal expression:
\begin{center}
\emph{let U=Y+Y in let true=(U==0) in true}
\end{center}

What is harmless is to perform narrowing at $Y+Y$ ($Y$ is
a free variable). This gives the substitution $[Y/0]$ and the result $0$ for the
subexpression $Y+Y$. Placing it in
its surrounding context,
the derivation continues as follows:
$$
\begin{array}{l}
let~U=0~in~let~V=(U==0)~in~{\it if}~V~then~true\\
\f let~V=(0==0)~in~{\it if}~V~then~true\\
\f let~V=true~in~{\it if}~V~then~true\\
\f {\it if}~true~then~true \f true
\end{array}
$$

\end{example}

The previous example shows that  let-narrowing  \emph{must
protect bound variables} against substitutions, which is the key observation for defining narrowing in presence of let-bindings.

The one-step let-narrowing relation $e \fnr_\theta e'$ (assuming a given program $\cal P$)
is defined in Figure \ref{letNarrcalc}.

\begin{figure*}
\framebox{
\begin{minipage}{0.95\textwidth}
\begin{center}
  \begin{description}
     \item[(X)] $e ~\fnrl_\epsilon~ e'$~~~~~if $e \f e'$ using {\boldmath $X$}$\in\! \{LetIn, Bind, Elim, Flat\}$ in Figure \refp{letrcalc}.\\[0.15cm]
    \item[(Narr)] $f(\overline{t})~\fnrl_\theta~r\theta$, for any fresh variant $(f(\overline{p}) \tor r) \in {\cal P}$ and $\theta \in CSubst$ such that $f(\overline{t})\theta \equiv f(\overline{p})\theta$.\\[0.15cm]
     \item[(Contx)] ${\cal C}[e] ~\fnrl_\theta~ {\cal C}\theta[e']$,~~for ${\cal C} \neq []$,
 if $e \fnrl_\theta e'$ by any of the previous rules, and if the step is  \crule{Narr} using $(f(\overline{p}) \tor r) \in {\cal P}$, then:\\
 \begin{tabular}[t]{rl}
 (i) & $dom(\theta) \cap BV({\cal C}) = \emptyset$\\
 (ii) & $vRan(\theta|_{\setminus var(\overline{p})}) \cap BV({\cal C}) = \emptyset$
 \end{tabular}
 \end{description}
\end{center}
\end{minipage}
}
\caption{Rules of the let-narrowing relation $\fnr$}
\label{letNarrcalc}
\end{figure*}

\begin{itemize}
\item The rule
{\bf \crule{X}} collects  \emph{\crule{Elim}, \crule{Bind}, \crule{Flat}, \crule{LetIn}} of 
$\f$, that remain the same in
$\fnr$,
except for the decoration with the empty substitution $\epsilon$.
\item The rule {\bf \crule{Narr}} performs a narrowing step in a proper sense.
To avoid unnecessary loss of generality or applicability of our approach, we do not impose $\theta$ to be a mgu.
\nota{The comment about soundness and completeness has been rephrased and moved to the new section about 'organizing computations'}
For the sake of readability, we will sometimes decorate \crule{Narr} steps with $\theta|_{FV(f(\overline{t}))}$ instead of $\theta$,
i.e., with the projection over the variables in the narrowed expression.

\item The rule {\bf \crule{Contx}} indicates how to use the previous rules in inner positions.
The condition ${\cal C} \not= [\ ]$ simply avoids trivial overlappings of \crule{Contx} with the previous rules.
The rest of the conditions are set to ensure that the combination of \crule{Contx}  with \crule{Narr} makes
a proper treatment of bound variables:
\begin{itemize}

\item {\em (i)} expresses the protection of bound variables
  against narrowing justified in Example \ref{ex:narr}.


\item {\em  (ii)} is a rather technical condition needed to prevent undesired
situations when the narrowing step has used a program rule with  extra variables
and a unifier $\theta$ which is not a mgu.
Concretely, the condition states that the bindings created by  $\theta$ for the extra variables
in the  program rule do not introduce variables that are bound by the surrounding context ${\cal C}$.
To see the problems that can arise without \emph{(ii)}, consider for instance the program rules
$f \to Y$ and $loop \to loop$ and the expression
  $let~X=loop~in~f$. A legal reduction for this expression, respecting condition \emph{(ii)}
could be the following:
$$
let~X=loop~in~f ~\fnr_\epsilon~ let~X=loop~in~Z
$$
by applying \crule{Narr} to $f$ with $\theta = \epsilon$ taking the fresh  variant rule $f\to Z$, and using \crule{Contx} for the whole expression. However, if we drop  condition $(ii)$ we could perform a similar derivation using the same fresh  variant of the rule for $f$, but now using the substitution $\theta=[Z/X]$:
$$
let~X=loop~in~f ~\fnr_\epsilon~ let~X=loop~in~X
$$
which is certainly not intended because the free variable $Z$ in the previous derivation appears now as a bound variable, i.e., we get an undesired capture of variables.

We remark that if the substitution $\theta$ in \crule{Narr}  is chosen to be a standard mgu\footnote{By standard mgu of $t,s$ we mean an idempotent mgu $\theta$ with $dom(\theta)\cup ran(\theta) \subseteq var(t)\cup var(s)$.} of $f(\overline{t})$ and $f(\overline{p})$ (which is always possible) then the condition \emph{(ii)} is always fulfilled.

\end{itemize}

\end{itemize}

The one-step relation $\fnr_\theta$ is extended in the natural way to the multiple-steps narrowing
relation $\fnre_\theta$, which is defined as the least  relation verifying:
$$e \fnre_\epsilon e~~~~e \fnr_{\theta_1} e_1 \fnr_{\theta_2} \ldots e_n \fnr_{\theta_n} e'\ \Rightarrow\ e \fnre_{\theta_1\ldots\theta_n} e'$$
We write $e \fnrc{n}_\theta e'$ for a n-steps narrowing sequence.


\begin{example}\label{ex:narrder}
Example \ref{ex:narr} essentially contains already a narrowing derivation. For the sake of clarity,
we repeat it 
here making explicit the rule of let-narrowing used at each step
(maybe in combination with \crule{Contx}, which is not written).
Besides, if the step uses \crule{Narr}, the narrowed expression is underlined.
$$
\begin{array}{ll@{\hspace{-0.1cm}}}
      even(coin) \fnr_\epsilon & \crule{LetIn}\\
      let~X=coin~in~\underline{even(X)}   \fnr_\epsilon  & \crule{Narr}\\

      let~X=coin~in~{\it if}~Y+Y==X~then~true   \fnrc{3}_\epsilon & \crule{LetIn^2,Flat}\\
      let~X=\underline{coin}~in~let~U=Y+Y~in\\

      \qquad let~V=(U==X)~in~{\it if}~V~then~true   \fnr_\epsilon & \crule{Narr}\\
      let~X=0~in~let~U=Y+Y~in\\
      \qquad let~V=(U==X)~in~{\it if}~V~then~true   \fnr_\epsilon & \crule{Bind}\\

      let~U=\underline{Y\!+\!Y}~in~let~V\!=\!(U\!==\!0)~in~{\it if}~V~then~true   \fnr_{[Y\!/0]} & \crule{Narr}\\
      let~U=0~in~let~V=(U==0)~in~{\it if}~V~then~true   \fnr_\epsilon & \crule{Bind}\\

      let~V=\underline{(0==0)}~in~{\it if}~V~then~true   \fnr_\epsilon & \crule{Narr}\\
      let~V=true~in~{\it if}~V~then~true   \fnr_\epsilon & \crule{Bind}\\
      {\it if}~true~then~true   \fnr_\epsilon & \crule{Narr}\\

      true  &

\end{array}
$$

Notice that all \crule{Narr} steps in the derivation except one have $\epsilon$ as
narrowing substitution (because of the projection over the variables of the
narrowed expression), so they are really rewriting steps.
An additional remark  that could help to further explain
the relationship between the let-narrowing relation $\fnrl$ and the let-rewriting relation $\f$ is the following:
since we have $even(coin) \fnr_\theta~ true$ for some $\theta$, but $even(coin)$ is ground,
Theorem \ref{SoundLNarr} in next section ensures that there
must be also a  successful let-rewriting derivation $even(coin) \f^*~ true$.
This derivation could have the form:
    $$\begin{array}{ll}
even(coin) \f & \crule{LetIn}\\
let~X=coin~in~ even(X)  \f  & \crule{Fapp}\\
let~X=coin~in~if~(0+0==X)~then~true   \f & \\
\ldots\ldots\ldots\ \f~ true& \\
\end{array}
$$


The indicated \crule{Fapp}-step in this let-rewriting derivation has used the substitution $[Y/0]$, thus anticipating and `magically guessing'
the correct value of the extra variable $Y$ of the rule of $even$.
In contrast, in the let-narrowing derivation the binding for $Y$ is not done while
reducing $even(X)$ but in a later \crule{Narr}-step over $Y+Y$. This corresponds closely to
the behavior of narrowing-based systems like Toy or Curry.

\end{example}

\subsection{Soundness and completeness of the let-narrowing relation $\fnrl$}
\label{soundComplLetNarr}
In this section we show the adequacy of let-narrowing wrt. let-rewriting. From now on we assume a fixed program $\cal P$.

As usual with narrowing relations, {soundness} results are not difficult to formulate and prove. The following \emph{soundness} result for $\fnrl$ states that we can mimic any $\fnrl$ derivation with $\f$ by applying over the starting expression the substitution computed by the original let-narrowing derivation.
\begin{theorem}[Soundness of the let-narrowing relation $\fnrl$]\label{SoundLNarr}
For any $e, e' \in LExp$, $e \fnre_{\theta} e'$ implies $e\theta \f^* e'$.
\end{theorem}
%
%

Completeness is more complicated to prove. The key result is a generalization
to  let-rewriting of Hullot's {\em lifting lemma} \cite{Hullot80} for classical term rewriting and narrowing.
It states that  any rewrite sequence for a particular instance of an expression can be generalized by a narrowing derivation.

\begin{lemma}[Lifting lemma for the let-rewriting relation $\f$]\label{lem:lifting}
Let $e,e' \in LExp$ such that $e\theta \fe e'$ for some $\theta \in CSubst$, and let
${\cal W}, {\cal B} \subseteq {\cal V}$ with $dom(\theta) \cup FV(e) \subseteq {\cal W}$, $BV(e) \subseteq {\cal B}$
and $(dom(\theta) \cup \vran(\theta)) \cap {\cal B} = \emptyset$, and for each \crule{Fapp} step of $e\theta \fe e'$  using a rule $R \in \prog$ and a substitution $\gamma \in CSubst$ then $\vran(\gamma|_{vExtra(R)}) \cap {\cal B} = \emptyset$.
Then there  exist a derivation $e ~\fnrl^*_{\sigma}~ e''$ and $\theta' \in CSubst$ such that:
$$ \mbox{(i)~} e''\theta' = e'
\qquad \mbox{(ii)~} \sigma\theta' = \theta[{\cal W}]
\qquad \mbox{(iii)~} (dom(\theta') \cup \vran(\theta')) \cap {\cal B} = \emptyset
$$
Besides, the let-narrowing derivation can be chosen to use mgu's at each \crule{Narr} step.
Graphically:
\vspace*{-0.2cm}
\begin{center}
\begin{tikzpicture}[scale=0.65, auto] 
    \node[nnarr] (e) at (-2,0) {$e$};
    \node[nnarr] (epp) at (2,0) {$e''$};
    \node[nnarr] (etheta) at (-2,-2) {$e\theta$};
    \node[nnarr] (ep) at (2,-2) {$e'$};
    \draw [->, apath, densely dotted, snake=snake, line after snake=2mm] (e) to node {$l^*$} node[swap] {$\sigma$} (epp);
    \draw [->, apath] (etheta) to node {$l^*$} (ep);
    \draw [|->, apath, shorten <=1pt] (e) to node [swap] {$\theta$} (etheta);
    \draw [|->, apath, shorten <=1pt, dashed] (epp) to node {$\theta'$} (ep);

\begin{pgfonlayer}{background}
   \filldraw [fondoTerm]
      (e.north -| e.west) rectangle (ep.south -| ep.east);
\end{pgfonlayer}

\end{tikzpicture}
\end{center}

\end{lemma}

With the aid of this lemma  we are now ready to state and prove the following strong completeness result for $\fnrl$.
\begin{theorem}[Completeness of the let-narrowing relation $\fnrl$]\label{TLNarrComp}
Let $e, e' \in LExp$ and $\theta \in CSubst$. If $e\theta \f^* e'$, then there exist
a let-narrowing derivation $e ~\fnrl^*_{\sigma}~ e''$ and $\theta' \in CSubst$ such that
$e''\theta' \equiv e'$ and $\sigma\theta' = \theta[FV(e)]$.
\end{theorem}
\begin{proof}
Applying Lemma \ref{lem:lifting} to $e\theta|_{FV(e)} \f^* e'$ with ${\cal W} = FV(e)$ and ${\cal B} = BV(e)$, as $e\theta|_{FV(e)}\equiv e\theta$ and the additional conditions over ${\cal B}$ hold by the variable convention.
\end{proof}



Finally, by combining  Theorems \ref{SoundLNarr} and  \ref{TLNarrComp}, we obtain a strong adequacy theorem for let-narrowing with respect to let-rewriting. 

\begin{theorem}[Adequacy of the let-narrowing relation $\fnrl$ wrt. $\f$]\label{TLNarrAdeq}
Let $e,e_1\in LExp$ and $\theta \in CSubst$, 
then:
$$
e\theta\f^* e_1\ \Leftrightarrow\
\begin{array}{l}
\textrm{there exist a let-narrowing
  derivation } e\fnre_{\sigma} e_2 \textrm{ and}\\
\textrm{some }
\theta'\in CSubst \textrm{ such that } \sigma\theta'=\theta[FV(e)],\ e_2\theta'\equiv e_1

\end{array}
$$
\end{theorem}
\begin{proof}
  \begin{description}
  \item[$(\Rightarrow)$] Assume $e\theta\fe e_1$. As $e\theta|_{FV(e)}\equiv
    e\theta$ then trivially $e\theta|_{FV(e)}\f^* e_1$. We can apply Lemma
    \ref{lem:lifting} taking ${\cal W}=FV(e)$ to get
    $e\fnre_{\sigma} e_2$ such that there exists $\theta'\in CSubst$ with
    $\sigma\theta'=\theta|_{FV(e)}[{\cal W}]$ and $e_2\theta'\equiv e_1$. But as
    ${\cal W}=FV(e)$ then $\sigma\theta'=\theta|_{FV(e)}[{\cal W}]$ implies
    $\sigma\theta'=\theta[FV(e)]$.

    We remark that the lifting lemma ensures that the
    narrowing derivation can be chosen to use mgu's at each {\clrule{Narr}} step.

  \item[$(\Leftarrow)$] Assume $e\fnre_{\sigma} e_2$ and $\theta'$ under the
    conditions above. Then by Theorem \ref{SoundLNarr} we have $e\sigma\f^* e_2$. As
    $\f$ is closed under c-substitutions (Lemma \ref{LRwCerr}) then $e\sigma\theta'\f^*
    e_2\theta'$. But as $\sigma\theta'=\theta[FV(e)]$, then $e\theta\equiv
    e\sigma\theta'\f^* e_2\theta'\equiv e_1$.
  \end{description}
\end{proof}

\subsection{Organizing computations}\label{subsection:strategies}

Deliberately, in this paper we have kept the definitions  of let-rewriting and narrowing apart from any particular computation strategy.
In this section we explain rather informally how the ideas of some known strategies for functional logic programming \cite{AntoyJSC05} can be adapted also to our formal setting.
For the sake of brevity we focus only on let-narrowing computations. As a running example, consider the program
\begin{center}
$\begin{array}{ll}
leq(0,Y) \ra true               & f(0) \ra 0 \\
leq(s(X),0) \ra false           & \\
leq(s(X),s(Y)) \ra leq(X,Y) \hh9 & \\
\end{array}
$
\end{center}
and the initial expression $leq(X,f(Y))$ to be let-narrowed using it.

As a first remark, when designing a strategy one can freely use `peeling' steps in a \emph{don't care} manner  using the relation $\fnf$ (Definition \ref{def:fnf}), since it is terminating and (hyper-)semantics-preserving. In our case one step suffices:
$leq(X,f(Y)) ~\fnrl_\epsilon let~U=f(Y)~in~leq(X,U)$.
After a peeling (multi-)step, a (Narr) step must be done. Where?
Certainly, the body $leq(\ldots)$ must be narrowed at some point. One \emph{don't know} possibility is narrowing at $leq(X,U)$ using the first rule for $leq$ that does not bind $U$:
$let~U=f(Y)~in~leq(X,U) ~\fnrl_{[X/0]} let~U=f(Y)~in~ true$.
A new peeling step leads to a first final result  $true$, with computed substitution $[X/0]$.

The second and third rules for $leq$ could lead to more results. Those rules have non-variable patterns as second arguments,  and then  the bound variable $U$ in $leq(X,U)$ inhibits a direct (Narr) step in that position. Typically it is said that $U$ is \emph{demanded} by those $leq$ rules. Therefore, we narrow  $f(Y)$ to get values for $U$, and then we `peel':

\begin{center}
$let~U=f(Y)~in~leq(X,U) ~\fnrl_{[Y/0]} let~U=0~in~leq(X,U) ~\fnrl_\epsilon leq(X,0)$ \hfill \emph{(1)}
\end{center}

The computation proceeds now by two don't know choices using the rules for $leq$, leading to two more solutions
$(true,[Y/0,X/0])$ and $(false,[Y/0,X/s(Z)])$.

This  implicitly applied strategy  can be seen as a translation to let-narrowing of \emph{lazy narrowing} \cite{MR92,AFIV03}. As a known drawback of lazy narrowing, notice that the second solution $(true,[Y/0,X/0])$ is redundant, since it is less general than the first one $(true,[X/0])$.
Redundancy is explained because we have narrowed the expression $f(Y)$ whose evaluation was demanded only by some of the rules for the outer function application $leq(X,f(Y))$, but after that we have used the rules not demanding the evaluation (the first rule for $leq$).
This problem is tackled successfully by \emph{needed narrowing} \cite{AEH00} which takes into account, when narrowing an inner expression, what are the rules for an outer function application demanding such evaluation. A needed narrowing step `anticipates' the substitution that will perform these rules when they are to be applied.
The ideas of needed narrowing can be adapted to our setting. In our example, we get the following derivation instead of (1):
$$
\begin{array}{lr}
  let~U=f(Y)~in~leq(X,U) ~\fnrl_{[X/s(Z),Y/0]} let~U=0~in~leq(s(Z),U) ~\fnrl_\epsilon & \\
leq(s(Z),0) ~\fnrl_\epsilon false  &
\end{array}\hfill \emph{(1')}$$

The first step does not use a mgu. This a typical feature of needed narrowing, and is also allowed by let-narrowing steps. Needed narrowing steps rely on \emph{definitional trees}  that structure demandness information from the rules of a given function. This information can be embedded also into a program transformation.
There are simple transformations for which the transformed program, under a lazy narrowing regime using mgu's, obtains the same solutions than the original program \cite{Zartmann97}, although it is not guaranteed that the number of steps is also preserved.
In our example, the definition of $leq$ can be transformed as follows:

\begin{center}
$\begin{array}{ll}
leq(0,Y) \ra true                & leqS(X,0) \ra false \\
leq(s(X),Y) \ra leqS(X,Y)   \hh9 & leqS(X,s(Y)) \ra leq(X,Y)\\
\end{array}
$
\end{center}
As happened with (1'), the derivation
$$
\begin{array}{ll}
let~U=f(Y)~in~leq(X,U)  ~\fnrl_{[X/s(Z)]} let~U=f(Y)~in~leqS(Z,U) ~\fnrl_{[Y/0]}\\
 let~U=0~in~leqS(Z,U) ~\fnrl_\epsilon leqS(Z,0) ~\fnrl_\epsilon false
\end{array}
$$
 gets rid of redundant solutions.

To which extent do our results guarantee the adequateness of the adaptation to let-narrowing of these strategies or others that could be defined? Certainly  any strategy is \emph{sound} for call-time choice semantics, because unrestricted $\fnrl$ is already sound (Theorem \ref{SoundLNarr}). This will be true also if the strategy uses derived rules in the sense of Section \ref{SemEqs}. With respect to completeness, we know that the space of let-narrowing derivations is complete wrt. let-rewriting (Theorem \ref{TLNarrComp}). But this does not imply  the completeness of the strategy, which in general will determine a smaller narrowing space. Therefore completeness of the strategy must be proved independently. Such a proof may use semantic methods (i.e., prove completeness wrt. \crwl-semantics) or operational methods (i.e., prove completeness wrt. $\f$-derivations).
We will not go deeper into the issue of strategies.

%
%
%
%
%


\section{Let-rewriting versus classical \nuevo{term} rewriting}\label{letR-classR}

In this section we examine the relationship between let-rewriting and ordinary term rewriting, with the focus put in the set of c-terms reachable by rewriting with each of these relations.
As term rewriting is not able to handle expressions with let-bindings, during this section we  assume that all considered programs do not have let-bindings in the right-hand side of its rules.

We will first prove in Section \ref{letsound} that let-rewriting is sound with respect to term rewriting, in the sense that any c-term that can be reached by a let-rewriting derivation from a given expression can also be reached by a term rewriting derivation starting from the same expression. As we know, completeness does not hold in general because
run-time choice computes more values than call-time choice for arbitrary programs. However, we will be able to prove completeness of let-rewriting \wrt\ term rewriting
for the class of {\em deterministic} programs, a notion close to confluence
that will be defined in Section \ref{letcompl}.
Finally, we will conclude in Section \ref{subsect:letnarrVsNarr} with a comparison between let-narrowing and narrowing, that will follow easily from the results in previous subsections and the adequacy of let-narrowing to let-rewriting.

Thanks to the strong equivalence between \crwl\ and let-rewriting we can
choose the most appropriate point of view for each of the two goals (soundness
and completeness): we will use  let-rewriting for proving  soundness, and
\crwl\ for defining the property of determinism and proving that, under
determinism, completeness of let-rewriting wrt. term rewriting also holds.

\subsection{Soundness of let-rewriting wrt. classical \nuevo{term} rewriting}
\label{letsound}


In order to relate let-rewriting to term rewriting, we first need to find a way for term rewriting to cope with let-bindings, which are not supported by its syntax, that is only able to handle expressions from $Exp$. Therefore,
we define the following syntactic transformation from $LExp$ into $Exp$ 
that takes care of
removing the let constructions, thus losing the
sharing information they provide.
\begin{definition}[Let-binding elimination transformation]\label{DefLetElimTrans}
Given $e \in LExp$ we define its transformation into a 
let-free expression $\tlr{e} \in Exp$ as:
$$
\begin{array}{l}
  \tlr{X}=_{\it def} X\\
  \tlr{h(e_1, \ldots, e_n)}=_{\it def} h(\tlr{e_1}, \ldots, \tlr{e_n})\\
  \tlr{let~X=e_1~in~e_2}=_{\it def} \tlr{e_2}[X/\tlr{e_1}]
\end{array}
$$
\end{definition}

Note that $\tlr{e} \equiv e$ for any $e \in Exp$.

We will need also the  following auxiliary lemma  showing 
the interaction between term rewriting derivations and substitution application.

\begin{lemma}[Copy lemma]\label{lemNoComp}
For all $e, e_1, e_2 \in Exp$, $X \in \var$:
\begin{enumerate}
    \item[i)]  $e_1 \rw e_2$ implies $e[X/e_1] \rw^* e[X/e_2]$.
    \item[ii)]  $e_1 \rw^* e_2$ implies $e[X/e_1] \rw^* e[X/e_2]$.
\end{enumerate}
\end{lemma}

Note how in {\it i)}, each of the different copies of $e_1$ introduced in $e$ by the substitution has to be reduced to $e_2$ in a different term rewriting step in order to reach the expression $e[X/e_2]$.


Using 
this lemma  we can get a first soundness result stating that the result of one let-rewriting step
can also be obtained in zero or more steps of ordinary
rewriting, after erasing the sharing information by means of the 
let-binding elimination transformation.

\begin{lemma}[One-Step Soundness of let-rewriting wrt. term rewriting]\label{LemLSoundRw}
For all $e, e' \in LExp$ we have that $e \f e'$ implies $\tlr{e} \rw^* \tlr{e}'$.
\end{lemma}

The remaining soundness results follow easily from this lemma. 
The first one shows how we can mimic let-rewriting with term rewriting through the let-binding elimination transformation. 
But then, 
as $\tlr{e} \equiv e$ for any $e \in Exp$, we conclude that for let-free expressions let-rewriting is a subrelation of term rewriting.
\begin{theorem}[Soundness of let-rewriting wrt. term rewriting]\label{TLSoundRw}
For any $e, e' \in LExp$ we have that $e \fe e'$ implies $\tlr{e} \rw^* \tlr{e'}$. As a consequence, if $e,e' \in Exp$ then $e \fe e'$ implies $e \rw^* e'$, i.e., $(\fe \cap \ (Exp \times Exp)) \subseteq \ \rw^*$.
\end{theorem}
\begin{proof}
The first part follows from an immediate induction on the length of the let-derivation, using Lemma \ref{LemLSoundRw} for the inductive step. The rest is obvious taking into account that $e \equiv \tlr{e}$ and $e' \equiv \tlr{e'}$ when $e,e' \in Exp$.
\end{proof}

To conclude this part, we can combine this last result with the equivalence of \crwl\ and let-rewriting, thus getting  the following soundness result for \crwl\ with respect to term rewriting.
\begin{theorem}[Soundness of \crwl\ wrt. term rewriting]\label{TCrwlSoundRw}
For any $e \in Exp$, $t \in CTerm_{\perp}$, if $e \clto t$ then there exists $e' \in Exp$ such that $e \rw^* e'$ and $t \sqsubseteq |e'|$.
\end{theorem}
\begin{proof}
Assume $e \clto t$.
By Theorem \ref{LComp}, there exists $e'' \in LExp$ such that $e \fe e''$ and $t \sqsubseteq |e''|$. Then, by Theorem \ref{TLSoundRw}, we have $\tlr{e}
\rw^* \tlr{e''}$. As $e\in Exp$, we have $e \equiv \tlr{e}$ and we can choose $e' \equiv \tlr{e''}\in Exp$ so we get
$e \rw^* e'$. It is easy to check that $|e''|=|\tlr{e''}|$
and then we have $t\sqsubseteq |e''| = | \tlr{e''} | = |e'|$.
\end{proof}

\subsection{Completeness of \crwl\ wrt. classical \nuevo{term} re\-writing}
\label{letcompl}

We prove here the completeness  of the \crwl\ framework \wrt\ term rewriting  for the class of \crwl-deterministic programs, which are defined as follows.

\begin{definition}[\crwl-deterministic program]\label{DefCrwlDetProg}
A program $\prog$ is \emph{\crwl-deterministic} iff for any expression $e \in Exp_{\perp}$ its denotation $\den{e}^{\cal P}$  is a directed set. In other words,  iff for all $e \in Exp_{\perp}$ and $t_1, t_2\in \den{e}^{\cal P}$, there exists $t_3 \in \den{e}^{\cal P}$ with $t_1 \sqsubseteq t_3$ and $t_2 \sqsubseteq t_3$.
\end{definition}

Thanks to the equivalence of \crwl\ and let-rewriting, it is  easy to characterize \crwl-determinism also in terms of let-rewriting derivations.

\begin{lemma}\label{detlet}
A program $\prog$ is \crwl-deterministic iff for any $e \in Exp$, $e', e'' \in LExp$ with $\prog \vdash e \fe e'$ and $\prog \vdash e \fe e''$ there exists $e''' \in LExp$ such that $\prog \vdash e \fe e'''$ and $|e'''| \sqsupseteq |e'|,|e'''| \sqsupseteq |e''|$.
%
\end{lemma}
\begin{proof}
For the left to right implication, assume a \crwl-deterministic program $\prog$ and $e\in Exp$, $e',e''\in LExp$ with $e\fe e'$ and $e\fe e''$. By part $iii)$ of Theorem \ref{LSound} we have $|e'|,|e''|\in \den{e}$ and then by Definition \ref{DefCrwlDetProg} there exists $t\in \den{e}$ such that $|e'|,|e''|\sqsubseteq t$. Now, by part $iii)$ of Theorem \ref{LComp} there exists $e'''\in LExp$ such that $e\fe e'''$ and $t\sqsubseteq |e'''|$, so we have $|e'|,|e''|\sqsubseteq t\sqsubseteq |e'''|$ as expected.

Regarding the converse implication, assume $e\in Exp$ with $t_1,t_2\in \den{e}$. By part $iii)$ of Theorem \ref{LComp} there exist $e',e''\in LExp$ such that $e\fe e'$, $e\fe e''$ and $t_1\sqsubseteq |e'|$, $t_2\sqsubseteq |e''|$. Then by hypothesis there exists $e'''\in LExp$ such that $e\fe e'''$ and $|e'|,|e''|\sqsubseteq |e'''|$. Now, by part $iii)$ of Theorem \ref{LSound} we have $|e'''|\in \den{e}$ and this $|e'''|$ is the $t_3$ of Definition  \ref{DefCrwlDetProg} we are looking for, i.e., $t_3\in \den{e}$ and $t_1,t_2\sqsubseteq t_3$.
%
%
%
\end{proof}

\crwl-determinism is intuitively close to confluence of term rewriting, 
but these two properties are not equivalent, as shown by the following example of 
a \crwl-deterministic but not confluent program.
\begin{example}
Consider the program $\cal P$ given by the 
rules \[f \ra a\hh5 f \ra \mathit{loop}\hh5 \mathit{loop}\ra \mathit{loop}\]
where $a$ is a constructor. It is clear that $\rw_{\prog}$ is not confluent
($f$ can be reduced to $a$ and \textit{loop}, which cannot be joined into a common reduct),
but it is \crwl-deterministic, since
$\den{f}^{\cal P} = \{\perp,a\}$, $\den{\mathit{loop}}^{\cal P} = \{\perp\}$ and $\den{a}^{\cal P} = \{\perp,a\}$,
which are all directed sets.
\end{example}

We conjecture  that the reverse implication is true, i.e., that confluence of term rewriting implies \crwl-determinism.
Nevertheless, a precise proof for this fact seems surprisingly complicated and we have not yet completed it.
A key ingredient in our completeness proof is the notion of \crwl-denotation of a substitution, which is the set of c-substitutions whose range can be obtained by \crwl-reduction over the range of the starting expression.

\begin{definition}[\crwl-denotation for a substitution]
Given a program $\prog$, the \crwl-denotation of a $\sigma \in Subst_{\perp}$ is defined as:
$$
\den{\sigma}^{\prog}_{\it CRWL} = \{\theta\! \in\! CSubst_{\perp}~|~\forall X \in \var,~\cl \sigma(X) \clto \theta(X) \}
$$
We will usually omit the subscript \crwl\ and/or the superscript $\prog$ when implied by the context.
\end{definition}

Any substitution $\theta$ in the denotation of some substitution $\sigma$ contains less information than $\sigma$, because it only holds in its range a finite part of the possibly infinite denotation of the expressions in the range of $\sigma$. We 
formalize this property in the following result.
%


\begin{proposition}\label{propDenSubstElemsDsord}
For all $\sigma \in Subst_\perp$, $\theta \in \den{\sigma}$, we have that $\theta \dsord \sigma$.
\end{proposition}

%
Besides, we will use the notion of deterministic substitution, which is a substitution with only deterministic expressions in its range.
\begin{definition}[Deterministic substitution]
The set $DSubst_{\perp}$ of \emph{deterministic substitutions} for a given program $\prog$ is defined as
$$
DSubst_\perp = \{\sigma \in Subst_\perp~|~\forall X \in dom(\sigma). \den{\sigma(X)} \mbox{ is a directed set}\}
$$
\end{definition}
Then $CSubst_\perp \subseteq DSubst_\perp$, and under any program $\forall \sigma \in Subst_\perp. \den{\sigma} \subseteq CSubst_\perp \subseteq DSubst_\perp$. 
Note that the determinism of substitutions depends on the program, which gives meaning to the functions in its range. Obviously if a program is deterministic then $Subst_\perp = DSubst_\perp$.\\

A good thing about deterministic substitutions is that their denotation is always a directed set.

\begin{proposition}\label{auxDenSubs2}
For all $\sigma \in DSusbt_{\perp}$, $\den{\sigma}$ is a directed set.
\end{proposition}

But the fundamental property of deterministic substitutions is that, for any \crwl-statement starting from an instance of an expression that has been constructed using a deterministic substitution, there is another \crwl-statement to the same value from another instance of the same expression that now has been built using a c-substitution taken from the denotation of the starting substitution. This property is a direct consequence of Proposition \ref{auxDenSubs2}.

\begin{lemma}\label{auxDenSubs3}
For all $\sigma \in DSusbt_{\perp}$, $e \in Exp_{\perp}, t \in CTerm_{\perp}$,
$$
\mbox{if } e\sigma \clto t \mbox{ then } \exists \theta \in \den{\sigma} \mbox{ such that } e\theta \clto t
$$
\end{lemma}
\begin{proof}[Proof (sketch)]
We proceed by a case distinction over $e$. If $e \equiv X \in dom(\sigma)$ then we have $e\sigma \equiv \sigma(X) \clto t$, and we can define
$$
\theta(Y) = \left \{\begin{array}{ll}

                     t & \mbox{ if } Y \equiv X \\
                     \perp & \mbox{ if } Y \in dom(\sigma) \setminus \{X\} \\

                     Y & \mbox{ otherwise }\\
                    \end{array}\right.
$$
Then it is easy to check that $\theta \in \den{\sigma}$ and besides $e\theta \equiv \theta(X) \equiv t \clto t$ by Lemma \ref{lemmashells}, so we are done. If $e \equiv X \in \var \setminus dom(\sigma)$ then we have $e\sigma \equiv \sigma(X) \equiv X \clto t$, and given $\overline{Y} = dom(\sigma)$ it is easy to check that $[\overline{Y/\perp}] \in \den{\sigma}$, and besides $e[\overline{Y/\perp}] \equiv X \clto t$ by hypothesis.

Finally if $e \not\in \var$ we proceed by induction on the structure of the
proof for $e\sigma \clto t$. The interesting cases are those for \crule{DC} and
\crule{OR} where we use that $\sigma \in DSusbt_{\perp}$, so by Proposition \ref{auxDenSubs2} its denotation
is directed. Then there must exist some $\theta
\in \den{\sigma}$ which is greater than each of the $\theta_i$ obtained by
induction hypothesis over the premises of the starting \crwl-proof for $e\sigma
\clto t$. Using the monotonicity of Proposition \ref{PropMonSubstCrwlLet} we can
prove $e\theta \clto t$, which also holds for
\crwl, by Theorem \ref{thEquivCrwlCrwllet} (see \ref{proofs}, page~\pageref{DEMO_auxDenSubs3} for details).
%
\end{proof}

Now we are finally ready to prove our first completeness result of \crwl\ wrt. term rewriting, for deterministic programs.
\begin{lemma}[Completeness lemma for \crwl\ wrt. term rewriting]\label{LCompCdRw}
Let $\prog$ be a \crwl-deterministic program, and  $e,e' \in Exp$. Then:
$$
e \rw^* e' \mbox{ implies } \den{e'} \subseteq \den{e}
$$
\end{lemma}
\begin{proof}
We can just prove this result for $e \rw e'$, then its extension for an arbitrary number of term rewriting steps holds by a simple induction on the length of the term rewriting derivation, using transitivity of $\subseteq$.

Assume $e \rw e'$, then the step must be of the shape $e \equiv \con[f(\overline{p})\sigma] \rw \con[r\sigma] \equiv e'$ for some program rule $(f(\overline{p}) \tor r) \in \prog$, $\sigma \in Subst$. First, let us focus on the case for $\con = [\ ]$, and then assume some $t \in CTerm_\perp$ such that $\cl e' \equiv r\sigma \clto t$. As $\prog$ is deterministic then $\sigma \in DSubst_\perp$, therefore by Lemma \ref{auxDenSubs3} there must exist some $\theta \in \den{\sigma}$ such that $\cl r\theta \clto t$. But then we can use $\theta$ to build the following \crwl-proof.
$$
\infer[OR]{f(\overline{p})\theta \clto t}
{
\ldots
\ p_i\theta \clto p_i\theta
\ldots
\ r\theta \clto t
}
$$
where for each $p_i \in \overline{p}$ we have $\cl p_i\theta \clto p_i\theta$ by Lemma \ref{lemmashells}, as $p_i \in CTerm$ because $\prog$ is a constructor system, and so $p_i\theta \in CTerm_\perp$, as $\theta \in \den{\sigma} \subseteq CSubst_\perp$. But we also have $\theta \dsord \sigma$ by Proposition \ref{propDenSubstElemsDsord}, therefore by applying the monotonicity for substitutions from 
Proposition \ref{PropMonSubstCrwlLet} ---which also holds for \crwl, by Theorem \ref{thEquivCrwlCrwllet}--- we get $\cl e \equiv f(\overline{p})\sigma \clto t$. Hence $\den{e'} = \den{r\sigma} \subseteq \den{f(\overline{p})\sigma} = \den{e}$.

Finally, we can generalize this result to arbitrary contexts by using the compositionality of \crwl\ from Theorem \ref{thCompoCrwl}. Given a term rewriting step $e \equiv \con[f(\overline{p})\sigma] \rw \con[r\sigma] \equiv e'$ then by the proof for $\con = [\ ]$ we get $\den{r\sigma} \subseteq \den{f(\overline{p})\sigma}$, but then
$$
\begin{array}{ll}
\den{e'} = \den{\con[r\sigma]} \\
= \bigcup\limits_{t \in \den{r\sigma}} \den{\con[t]} & \mbox{ by Theorem \ref{thCompoCrwl}} \\
\subseteq \bigcup\limits_{t \in \den{f(\overline{p})\sigma}} \den{\con[t]} & \mbox{ as $\den{r\sigma} \subseteq \den{f(\overline{p})\sigma}$} \\
= \den{\con[f(\overline{p})\sigma]} = \den{e} & \mbox{ by Theorem \ref{thCompoCrwl}} \\
\end{array}
$$
\end{proof}

%


The previous lemma, together with the equivalence of \crwl\ and let-rewriting given by Theorem \ref{TheorEquivLetRwDown} 
and Theorem \ref{thEquivCrwlCrwllet}, allows us to obtain a strong relationships between term rewriting, let-rewriting and \crwl, for the class of \crwl-deterministic programs.

\begin{theorem}\label{TAdCrwlRw}
Let $\prog$ be a \crwl-deterministic program, and  $e,e'\in Exp, t \in CTerm$. Then:
\begin{enumerate}
    \item[a)] $e \rw^* e'$ implies $e \fe e''$ for some $e'' \in LExp$ with $|e''| \sqsupseteq |e'|$.
    \item[b)] $e \rw^* t$\ iff\ $e \fe t$\ iff\ $\cl e \clto t$.
\end{enumerate}
\end{theorem}
%
%
%
%
%

Notice that in part {\it a)} 
we cannot ensure $e \rw^* e'$ implies $e \fe e'$, 
because term rewriting can reach some intermediate
expressions not reachable by let-rewriting. For instance, given the deterministic program with
the rules $g \rw a$ and $f(x) \rw c(x,x)$, we have $f(g) \rw^* c(g,a)$, but $f(g) \not\fe c(g,a)$. 
Still, parts {\it a)} 
is a strong completeness results for let-rewriting wrt. term rewriting for deterministic programs, since it says that the outer constructed part obtained in a rewriting derivation can be also obtained or even refined in a let-rewriting derivation. Combined with Theorem \ref{TLSoundRw}, part {\it a)} 
expresses a kind of equivalence between let-rewriting and term rewriting, valid for general derivations, even non-terminating ones. For derivations reaching a constructor term (not further reducible), part 
{\it b)} gives an even stronger equivalence result.

\subsection{Let-narrowing versus narrowing for deterministic systems}\label{subsect:letnarrVsNarr}

Joining the results of the previous section with the adequacy of let-narrowing to let-rewriting,
we can easily establish some relationships between let-narrowing and ordinary term rewriting/narrowing, 
summarized in the following result.

\begin{theorem}\label{lnr-wrt-rw}
For any program $\prog$, $e \in Exp, \theta \in CSubst$ and $ t \in CTerm$:
\begin{description}
    \item[a)] If $e \fnre_{\theta} t$ 
    then $e\theta \rw^* t$.
    \item[b)] If in addition ${\cal P}$ is \crwl-deterministic, then:
\begin{description}
    \item[b$_1$)]  If $e\theta \rw^* t$ then $\exists t'\in CTerm,\ \sigma, \theta' \in CSubst$ such that $e \fnre_{\sigma} t'$, 
    $t'\theta' \equiv t$ and $\sigma\theta' = \theta[var(e)]$.
    \item[b$_2$)] If $e \nr^*_{\theta} t$, the same conclusion of \emph{(b$_1$)} holds.
\end{description}
\end{description}
 \end{theorem}

Part $a)$ expresses soundness of $\fnr$ wrt. \nuevo{term} rewriting, and part $b)$ is a completeness
result for $\fnr$ wrt. \nuevo{term} rewriting/narrowing, for the class of deterministic programs.

\begin{proof}
 Part $a)$ 
 follows from soundness of let-narrowing wrt. let-rewriting (Theorem \ref{SoundLNarr}) and soundness of let-rewriting wrt. \nuevo{term} rewriting of Theorem \ref{TAdCrwlRw}. 

 For part $b_1)$, for let-narrowing, assume $e\theta \rw^* t$. By the  completeness of let-rewriting wrt. \nuevo{term} rewriting for deterministic programs
 (Theorem \ref{TAdCrwlRw}), we have $e\theta \f^* t$, and then by
 the completeness of let-narrowing wrt. let-rewriting
 (Theorem \ref{TLNarrComp}), there exists a narrowing derivation  $e \fnre_{\sigma} t'$ with
 $t'\theta' = t$ and $\sigma\theta' = \theta[FV(e)]$. But notice that for $e\in Exp$, the sets
 $FV(e)$ and $var(e)$ coincide, and the proof is finished. 

 Finally, $b_2)$ follows simply from soundness of (ordinary) narrowing
 wrt. term rewriting and $b_1)$.
 \end{proof}



\section{Conclusions}
\label{conclusions}

This paper contains a  thorough presentation of the theory of first order let-rewriting and let-narrowing for constructor-based term rewriting systems. These two relations are  simple notions of one-step reduction that express sharing as it is required by the call-time choice semantics of non-determinism adopted in the functional logic programming paradigm.
In a broad sense, let-rewriting and let-narrowing can be seen as particular syntactical presentations of
term graph rewriting and narrowing. However, keeping our formalisms very close
to the syntax and basic notions of term rewriting systems (terms,
substitutions, syntactic unification,\ldots) has been an essential aid in
establishing strong equivalence results with respect to the \crwl-framework ---a
well-established realization of call-time choice semantics---, which was one of the main aims of
the paper.

Along the way of proving such equivalence we have developed powerful semantic tools that are interesting in themselves. Most remarkably,
the $\crwllet$-logic, a conservative extension of \crwl\ that deals with let-bindings, and the notion of hypersemantics of expressions and contexts,  for which we prove deep compositionality results not easily achievable by thinking directly in terms of reduction sequences.

We have shown in several places the methodological power of having provably equivalent reduction-based and logic-based semantics.
In some occasions, we have used the properties of the \crwl-semantics to investigate interesting aspects of reductions, as replaceability conditions or derived operational rules, like bubbling. In others, we have followed the converse way. For instance, by transforming let-rewriting reductions into ordinary term rewriting reductions, we  easily concluded that let-rewriting (call-time choice) provides less computed values than term rewriting (run-time choice). By using again semantic methods, we proved the opposite inclusion for deterministic programs, obtaining for such programs an equivalence result of let-rewriting   and term rewriting.

In our opinion, the different pieces of this work can be used separately for different
purposes. The \crwllet-logic provides a denotational semantics reflecting call-time choice for
programs making use of local bindings. The let-rewriting and let-narrowing relations provide clear and abstract descriptions
of how computations respecting call-time choice can proceed. They can be useful to explain basic operational aspects of functional logic languages to students or novice programmers, for instance. They  have been  used also as underlying formalisms to investigate other aspects of functional logic programming that need a clear notion of reduction;  for instance, when proving essential properties of type systems, like subject reduction or progress.
In addition, all the pieces are interconnected by strong theoretical results, which may be useful depending on the pursued goal.

Just like classical term rewriting and narrowing, the let-rewriting and narrowing relations define too broad computation spaces  as to be adopted directly as concrete operational procedures of a programming language. To that purpose, they should be accompanied by a strategy that selects only certain computations. In this paper we have only given an example-driven discussion of strategies.
We are quite confident that some known on-demand evaluation
strategies, like lazy, needed or natural rewriting/narrowing, can be adapted to our formal setting.
In \cite{pepm10_adriJuan,jaimeprole2011} we work out in more detail two concrete on-demand strategies for slight variants of let-rewriting and narrowing formalisms.

A subject of future work that might be of interest to the functional logic community is that of completing the comparison  of different formalisms proposed in the field to capture call-time choice semantics: \crwl, admissible term graph rewriting/narrowing, natural semantics \emph{\`{a} la} Launchbury, and let-rewriting/narrowing. Proving their equivalence  would greatly enrich the set of tools available to the functional logic programming theoretician, since any known or future result obtained for one of the approaches could be applied to the rest on a sound technical basis.


\appendix
\section{Detailed proofs for the results}
\label{proofs}

In the proofs we will use the usual notation for positions, subexpressions and repacements from \cite{BaaderNipkow-98}. The \emph{set of positions} of an expression $e \in Exp$ is a set $O(e)$ of strings of positive integers defined as:
\begin{itemize}
    \item If $e \equiv X \in \var$, then $O(e) = {\epsilon}$, where $\epsilon$ is the empty string.
    \item If $e \equiv h(e_1, \ldots, e_n)$ with $h \in \Sigma$, then
    $$
    O(e) =  \{ \epsilon \} \cup \bigcup_{i =1}^{n} \{ ip ~|~ p \in O(e_i) \}
    $$
\end{itemize}

The \emph{subexpression of $e$ at position $p \in O(e)$}, denoted $e|_p$, is defined as:
$$\begin{array}{rcl}
e|_{\epsilon} & = & e \\
h(e_1, \ldots, e_n)|_{ip} & = & e_i|_p \\
\end{array}$$

For a position $p \in O(e)$, we define the \emph{replacement of the subexpression of $e$ at position $p$ by $e'$} ---denoted $e[e']_p$--- as follows:
$$\begin{array}{rcl}
e[e']_{\epsilon} & = & e' \\
h(e_1, \ldots, e_n)[e']_{ip} & = & h(e_1, \ldots, e_i[e']_p, \ldots, e_n) \\
\end{array}$$


When performing proofs by induction we will usually use IH to refer to the induction hypothesis of the current induction. We will use an asterisk to denote the use of a let-rewriting rule one or more times, as in \lrrule{Flat*}.  We will also use the following auxiliary results.

\newcommand{\teoremi}[1]{\noindent {\em #1}~\\}

\subsection{Lemmas}

The following lemmas are used in the proofs for the results in the article. Most of them are straightforwardly proved by induction, so we only detail the proof in the interesting cases.


\begin{lemma}\label{shellPatterns}
$\forall t \in CTerm_ {\perp}. ~|t| = t$.
\end{lemma}


\begin{lemma}\label{lemma:crwlPatterns}
$\forall t \in CTerm_ {\perp}. ~{\cal P}\vdash_{{\it CRWL}_{let}} t \crwlto t$.
\end{lemma}


\begin{lemma}\label{LLetOrd2}
 Given $\theta, \theta' \in LSubst_{\perp}$, $e \in LExp_{\perp}$, if $\theta \sqsubseteq \theta'$ then $e\theta \sqsubseteq e\theta'$.
\end{lemma}

\begin{lemma}\label{auxMon}
Given $\theta \in LSubst_{\perp}$, $e, e' \in LExp_{\perp}$, if $e \sqsubseteq e'$ then $e\theta \sqsubseteq e'\theta$.
\end{lemma}

\begin{lemma}\label{LLetOrd3}
For every $e, e' \in LExp_{\perp}$, ${\cal C} \in Cntxt$, if $|e| \sqsubseteq |e'|$ then $|{\cal C}[e]| \sqsubseteq |{\cal C}[e']|$.
\end{lemma}
\begin{proof}\label{DEMO_LLetOrd3}
We proceed by induction on the structure of ${\cal C}$. The base case is straightforward because of the hypothesis. For the Inductive Step we have:
        \begin{itemize}
         \item ${\cal C} \equiv h(\ldots, {\cal C'}, \ldots)$. Directly by IH.
         \item ${\cal C} \equiv let~X={\cal C'}~in~e_1$, so ${\cal C}[e] \equiv let~X={\cal C'}[e]~in~e_1$. Then:
$$
\begin{array}{l}
|{\cal C}[e]| = |let~X={\cal C'}[e]~in~e_1| = |e_1|[X/|{\cal C'}[e]|] \\
\sqsubseteq_{\mi{IH}^{(*)}} |e_1|[X/|{\cal C'}[e']|] = |let~X={\cal C'}[e']~in~e_1| = |{\cal C}[e']|  \\
\end{array}
$$
$(*)$ By IH we have $|{\cal C'}[e]| \sqsubseteq |{\cal C'}[e']|$, therefore $[X/|{\cal C'}[e]|] \sqsubseteq [X/|{\cal C'}[e']|]$. Finally, by Lemma \ref{LLetOrd2}, $|e_1|[X/|{\cal C'}[e]|] \sqsubseteq |e_1|[X/|{\cal C'}[e']|]$.
         \item ${\cal C} \equiv let~X=e_1~in~{\cal C'}$. Similar to the previous case but using Lemma \ref{auxMon} to obtain $|{\cal C'}[e]|$ $[X/|e_1|] \sqsubseteq |{\cal C'}[e']|[X/|e_1|]$ from the IH $|{\cal C'}[e]| \sqsubseteq |{\cal C'}[e']|$.
       \end{itemize}
\end{proof}

\begin{lemma}\label{lemma:context}
If $|e| = |e'|$ then $|\cntx[e]| = |\cntx[e']|$
\end{lemma}
\begin{proof}\label{DEMO_lemma:context}
Since $\sqsubseteq$ is a partial order, we know by reflexivity that $|e| \sqsubseteq |e'|$ and $|e'| \sqsubseteq |e|$. Then by Lemma \ref{LLetOrd3} we have $|\cntx[e]| \sqsubseteq |\cntx[e']|$ and $|\cntx[e']| \sqsubseteq |\cntx[e]|$. Finally, by antisymmetry of the partial order $\sqsubseteq$ we have that $|\cntx[e]| = |\cntx[e']|$.
\end{proof}

\begin{lemma}\label{LCasc1}
For all $e_1, e_2 \in LExp, X \in {\cal V}$, $|e_1[X/e_2]| \equiv |e_1|[X/|e_2|]$
\end{lemma}
\begin{proof}\label{DEMO_LCasc1}
By induction on the structure of $e_1$. The most interesting case is when $e_1 \equiv let~Y=s_1~in~s_2$.
    By the variable convention $Y \not\in dom([X/e_2])$ and $Y \not\in \vran([X/e_2])$, so:
$$
\begin{array}{ll}
|e_1[X/e_2]|\equiv |let~Y=s_1[X/e_2]~in~s_2[X/e_2]| \\
\equiv |s_2[X/e_2]|[Y/|s_1[X/e_2]|] \\
\equiv_{\mathit{IH}} |s_2|[X/|e_2|][Y/(|s_1|[X/|e_2|])]\\
\equiv |s_2|[Y/|s_1|][X/|e_2|] & \mbox{(*)} \\
\equiv |let~Y=s_1~in~s_2|[X/|e_2|] \equiv |e_1|[X/|e_2|]
\end{array}
$$
(*) Using Lemma \ref{auxBind} with the matching $[e/|s_2|, \theta/[X/|e_2|], X/Y, e'/|s_1|]$.
\end{proof}

\begin{lemma}\label{lemaAproxTheta}
Given $\theta \in LSubst_{\perp}$, $e,e' \in LExp_{\perp}$, if $e \ordap e'$ then $e\theta \ordap e'\theta$.
%
\end{lemma}

\begin{lemma}\label{T28}\label{lemSintaxtConSubst}
For every $\sigma \in LSubst_{\perp}$, ${\cal C} \in Cntxt$ and $e \in LExp_{\perp}$ such that $(dom(\sigma) \cup \vran(\sigma)) \cap BV({\cal C}) = \emptyset$ we have that $({\cal C}[e])\sigma \equiv {\cal C}\sigma[e\sigma]$. 
\end{lemma}
\begin{proof}
By induction on the structure of ${\cal C}$. The most interesting cases are those concerning let-expressions:
\begin{itemize}
        \item ${\cal C} \equiv let~X={\cal C'}~in~e_1$: therefore ${\cal C}[e] \equiv let~X={\cal C'}[e]~in~e_1$. Then
$$
\begin{array}{c}
({\cal C}[e])\sigma \equiv let~X=({\cal C'}[e])\sigma~in~e_1\sigma \equiv_{\mathit{IH}}^{(*)} let~X={\cal C'}\sigma[e\sigma]~in~e_1\sigma \\
\equiv (let~X=({\cal C'}[])\sigma~in~e_1\sigma)[e\sigma] \equiv^{(**)} ((let~X={\cal C'}[]~in~e_1)\sigma)[e\sigma] \equiv {\cal C}\sigma[e\sigma]  \\
\end{array}
$$
$(*)$: by definition $BV(let~X = {\cal C'}~in~e) = BV({\cal C'})$, so $(dom(\sigma) \cup \vran(\sigma)) \cap BV({\cal C}) = \emptyset = (dom(\sigma) \cup \vran(\sigma)) \cap BV({\cal C'})$.\\
$(**)$: we can apply the last step because by hypothesis we can assure that we do not need any renaming to apply $(let~X={\cal C'}[]~in~e_1)\sigma$.
        \item ${\cal C} \equiv let~X=e_1~in~{\cal C'}$: therefore ${\cal C}[e] \equiv let~X=e_1~in~{\cal C'}[e]$. Then
$$
\begin{array}{c}
({\cal C}[e])\sigma \equiv let~X=e_1\sigma~in~({\cal C'}[e])\sigma \equiv_{\mathit{IH}} let~X=e_1\sigma~in~{\cal C'}\sigma[e\sigma] \\
\equiv (let~X=e_1\sigma~in~ ({\cal C'}[])\sigma)[e\sigma] \equiv^{(*)} ((let~X=e_1~in~{\cal C'}[])\sigma)[e\sigma] \equiv {\cal C}\sigma[e\sigma]
\end{array}
$$
$(*)$: we can apply the last step because by hypothesis we can assure that we do not need any renaming to apply $(let~X=e_1~in~{\cal C'}[])\sigma$.
\end{itemize}
\end{proof}

\begin{lemma}\label{lDios}
For any $e \in Exp_{\perp}$, $t \in CTerm_{\perp}$ and program ${\cal P}$, if
${\cal P}\vdash e \clto t$ then there is a derivation for ${\cal P}\vdash e \clto t$ in which every
free variable used belongs to $FV(e \clto t)$.
\end{lemma}
\begin{proof}\label{DEMO_lDios}
A simple extension of the proof in \cite{DiosLopez07}.
\end{proof}

\begin{lemma}\label{LAlfCrwlP}
For every ${\it CRWL}_{\it let}$ derivation $e \clto t$ there exists $e' \in
LExp_{\perp}$ which is syntactically equivalent to $e$ module
$\alpha$-conversion, and a $CRWL_{let}$ derivation for $e' \clto t$ such
that if ${\cal B}$ is the set of bound variables used in $e' \clto t$ and
${\cal E}$ is the set of free variables used in the instantiation of extra
variables in $e' \clto t$ then ${\cal B} \cap ({\cal E} \cup var(t)) =
\emptyset$.
\end{lemma}
\begin{proof}\label{DEMO_LAlfCrwlP}
By Lemma \ref{lDios}, if ${\cal F}$ is the set of free variables used in $e'
\clto t$, then ${\cal F} \subseteq FV(e' \clto t)$, in fact ${\cal F} = FV(e'
\clto t)$, as $FV(e')$ and $FV(t)$ are used in the top derivation of the
derivation tree for $e' \clto t$. As by definition ${\cal E} \cup var(t)
\subseteq {\cal F}$, if we prove ${\cal B} \cap {\cal F} = \emptyset$ then
${\cal B} \cap ({\cal E} \cup var(t)) = \emptyset$ is a trivial consequence. To
prove that we will prove that for every $a \in LExp_{\perp}$ used in the
derivation for $e' \clto t$ we have $BV(a) \cap FV(a) = \emptyset$. We can build
$e'$ using $\alpha$-conversion to ensure that $BV(e') \cap FV(e') =
\emptyset$. This can be easily maintained as an invariant during the derivation,
as the new let-bindings that appear during the derivation are those
introduced in the instances of the rule used during the \textbf{OR} steps, and
be can ensure by $\alpha$-conversion that $BV(a) \cap FV(a) = \emptyset$ for
these instances too, as $\alpha$-conversion leaves the hypersemantics
untouched.
\end{proof}

\subsection{Proofs for Section \ref{sectPrelimCRWL}}

\teoremi{Theorem \ref{thCompoCrwl} (Compositionality of \crwl)}
For any $\con \in Cntxt$, $e, e' \in Exp_\perp$
$$
\den{\con[e]} = \bigcup\limits_{t \in \den{e}} \den{\con[t]}
$$
As a consequence: $\den{e} = \den{e'} \Leftrightarrow \forall \con \in Cntxt. \den{\con[e]} = \den{\con[e']}$

\begin{proof}\label{DEMO_thCompoCrwl}
We prove that $\con[e] \clto t \Leftrightarrow \exists s \in CTerm_{\perp}$ such that $e \clto s$ and $\con[s] \clto t$.
~\\~\\
$\Rightarrow$) Induction on the size of the proof for $\con[e] \clto t$.

{\bf Base case} The base case only allows the proofs $\con[e]\clto\bot$ using \clrule{B}, $\con[e]\equiv X\clto X$ using \clrule{RR} and $\con[e]\equiv c \clto c$ with $c \in CS$ using \clrule{DC}, that are clear. When $\con = [\ ]$ the proof is trivial with $s = t$ and using Lemma \ref{lemma:crwlPatterns}.

{\bf Inductive step} Direct application of the IH.

~\\~\\
$\Leftarrow$) By induction on the size of the proof for $\con[s] \clto t$

{\bf Base case} 
The base case only allows the proofs $\con[s]\clto\bot$, $\con[s]\equiv X\clto X$ and $\con[s]\equiv c \clto c$ with $c \in CS$, that are clear. When $\con = [\ ]$ we have that $\exists s \in CTerm_{\perp}$ such that $e \clto s$ and $s \clto t$. Since $s \clto t$ by Lemma \ref{lemmashells} we have $t \ordap s$, and using Proposition \ref{propCrwlletPolar} $e \clto t$ ---as $e \ordap e$ because $\ordap$ is a partial order.

{\bf Inductive step} Direct application of the IH.

\end{proof}

\subsection{Proofs for Section \ref{discussion}}

\teoremi{Theorem \ref{brcrwl}}
Let ${\cal P}$ be a \crwl-program, $e \in Exp_\perp$ and $t \in CTerm_\perp$. Then:
$${\cal P} \conscrwl e \crwlto t\ \mathit{iff}\ e \rightarrowtail^*_{\cal P} t$$
\begin{proof}
It is easy to see that $\rightarrowtail^*$ 
coincides with the  relation  defined by the   {\it BRC}-proof calculus of \cite{GHLR99}, that is,
${\cal P} \vdash_{\it BRC} e \crwlto e'\ \leftrightarrow e \rightarrowtail^* e'$.
But in that paper it is proved that {\it BRC}-derivability
and \crwl-derivability (called there {\it GORC}-derivability) are equivalent.
\end{proof}

\subsection{Proofs for Section \ref{let-rewriting}}

\teoremi{Lemma \ref{auxBind} (Substitution lemma for let-expressions)}
Let $e, e'\in LExp_{\perp}$, $\theta \in Subst_{\perp}$ and $X \in \var$ such that $X \not\in dom(\theta) \cup \vran(\theta)$.
Then:
$$(e[X/e'])\theta \equiv e\theta[X/e'\theta]$$

\begin{proof}\label{DEMO_auxBind}
By induction over the structure of $e$. The most interesting cases are the base cases:
\begin{itemize}
\item $e \equiv X$: Then
$
\begin{array}{l}
\\
(e[X/e'])\theta \equiv (X[X/e'])\theta \equiv e'\theta \equiv X[X/e'\theta] \\
\equiv_{X \not\in dom(\theta)} X\theta[X/e'\theta] \equiv e\theta[X/e'\theta] \\
\end{array}
$
\item $e \equiv Y \not\equiv X$: Then
$
\begin{array}{l}
\\
(e[X/e't])\theta \equiv (Y[X/e'])\theta \equiv Y\theta \\
\equiv_{X \not\in ran(\theta)} Y\theta[X/e'\theta] \equiv e\theta[X/e'\theta] \\
\end{array}
$
\end{itemize}
\end{proof}

\subsection{Proofs for Section \ref{sect:letRwRelation}}

\teoremi{Lemma \ref{LRwCerr} (Closedness under $CSubst$ of let-rewriting)}
For any $e, e' \in LExp$, $\theta \in CSubst$ we have that $e \f^n e'$ implies $e\theta \f^n e'\theta$.
\begin{proof}\label{DEMO_LRwCerr}
We prove that $e \f e'$ implies $e\theta \f e'\theta$ by a case distinction over the rule of the let-rewriting calculus applied:
 \begin{description}
        \item[\lrrule{Fapp}] Assume $f(t_1, \ldots, t_n) \f r$, using $(f(p_1, \ldots, p_n) \to e) \in \mathcal{P}$ and $\sigma \in CSubst$ such that $\forall i.p_i\sigma = t_i$ and $e\sigma = r$. But since $\sigma \theta \in CSubst$ and $\forall i.p_i\sigma\theta = t_i\theta$ then we can perform a \lrrule{Fapp} step $f(t_1, \ldots, t_n)\theta \equiv f(t_1\theta, \ldots, t_n\theta) \f e\sigma\theta \equiv r\theta$.


        \item[\lrrule{LetIn}] Easily since $X \not\in dom(\theta)$ because $X$ is fresh.

        \item[\lrrule{Bind}] Assume $let~X=t~in~e \f e[X/t]$ and some $\theta \in CSubst$. Then $t \in CTerm$ by the conditions of \lrrule{Bind}, hence $t\theta \in CTerm$ too and we can perform a \lrrule{Bind} step $(let~X=t~in~e)\theta \equiv let~X=t\theta~in~e\theta \f e\theta[X/t\theta]$. Besides $X \not\in (dom(\theta) \cup \vran(\theta))$ by the variable convention, and so $e\theta[X/t\theta] \equiv e[X/t]\theta$ by Lemma \ref{auxBind}, so are done.

        \item[\lrrule{Elim}] Easily as $X \not\in FV(e_2\theta)$ because $X \not\in \vran(\theta)$ by the variable convention.

        \item[\lrrule{Flat}] Similar to the previous case since $Y \not\in FV(e_3\theta)$.

        \item[\lrrule{Contx}] Assume ${\cal C}[e] \f {\cal C}[e']$ because $e \f' e'$ by one of the previous rules, and some $\theta \in CSubst$. Then we have already proved that $e\theta \f e'\theta$. Besides by the variable convention we have $BV({\cal C}) \cap (dom(\theta) \cup \vran(\theta)) = \emptyset$, hence by Lemma \ref{T28} $(\con[e])\theta \equiv \con\theta[e\theta]$. Furthermore, if $e \f e'$ was a \lrrule{Fapp} step using $\sigma \in CSubst$ to build the instance of the program rule $(f(\overline{p})\sigma \tor r\sigma)$, then $\vran(\sigma|_{\setminus var(\overline{p})}) \cap BV(\con) = \emptyset$ by the conditions of \lrrule{Contx}, and therefore $\vran((\sigma\theta)|_{\setminus var(\overline{p})}) \cap BV(\con) = \emptyset$. But as $\sigma\theta$ is the substitution used in the \lrrule{Fapp} step $e\theta \f e'\theta$, then $\con\theta[e\theta] \f \con\theta[e'\theta]$ by \lrrule{Contx}. 
        On the other hand, if $e \f e'$ was not a \lrrule{Fapp} step then $\con\theta[e\theta] \f \con\theta[e'\theta]$ too, and finally we can apply Lemma \ref{T28} again to get $\con\theta[e'\theta] \equiv (\con[e'])\theta$.
        \end{description}

The proof for $e \f^n e'$ proceeds straightforwardly by induction on the length $n$ of the derivation.
\end{proof}

\teoremi{Proposition \ref{termlr} (Termination of $\fnf$)}
Under any program we have that $\fnf$ is terminating.
\begin{proof}\label{DEMO_termlr}
We define for any $e\in LExp$ the size $(k_1,k_2,k_3)$, where
\\[0.5ex]
{\em
\begin{tabular}{l}
$k_1 \equiv$ number of subexpressions in $e$ to which \lrrule{LetIn} is applicable.
\\
$k_2 \equiv$ number of \emph{lets} in $e$.
\\
$k_3 \equiv$ sum of the levels of nesting of all let-subexpressions in $e$.
\end{tabular}
\\[0.5ex]
}

\noindent
Sizes are lexicographically ordered. We prove now that application of \emph{\lrrule{LetIn}, \lrrule{Bind}, \lrrule{Elim}, \lrrule{Flat}} in any context (hence, also the application of \lrrule{Contxt}) decreases the size, what proves termination of $\fnf$. The effect of each rule in the size is summarized as follows (in each case, we stop
at the decreasing component):
\begin{center}
\begin{tabular}{ll}
\emph{(LetIn):} & $(<,\_,\_)$
\\
\emph{(Bind):} & $(=,<,\_)$
\\
\emph{(Elim):} & $(\leq,<,\_)$
\\
\emph{(Flat):} & $(=,=,<)$
\\
\end{tabular}
\end{center}
\end{proof}

\teoremi{Lemma \ref{lemma:bigPeeling} (Peeling lemma)}
For any $e, e' \in LExp$ if $e \nf{\it lnf} e'$ ---i.e, $e'$ is a $\fnf$ normal form for $e$--- then $e'$ has the shape $e' \equiv let~\overline{X=f(\overline{t})}~in~e''$ such that $e'' \in \var$ or $e'' \equiv h(\overline{t'})$ with $h \in \Sigma$, $\overline{f} \subseteq FS$ and $\overline{t},\overline{t'} \subseteq CTerm$. \\
Moreover if $e \equiv h(e_1, \ldots, e_n)$ with $h \in \Sigma$, then
$$
e \equiv h(e_1, \ldots, e_n) \fnfe let~\overline{X=f(\overline{t})}~in~h(t_1, \ldots, t_n) \equiv e'
$$
under the conditions above, and verifying also that $t_i \equiv e_i$ whenever $e_i \in CTerm$.

\begin{proof}\label{DEMO_lemma:bigPeeling}
We prove it by contraposition: if an expression $e$ does not have that shape, $e$ is not a $\fnf$ normal form. We define the set of expressions which are not cterms as:

\begin{tabular}{ll@{}}
$nt$ ::=    & $c(\ldots,nt,\ldots)$ \\
            & $|~f(\overline{e})$ \\
            & $|~let~X=e_1~in~e_2$ \\
\end{tabular}

We also define the set of expressions which do not have the presented shape recursively as:

\begin{tabular}{ll}
$ne$ ::=    & $h(\ldots, nt, \ldots)$ \\
            & $|~let~X=f(\overline{t})~in~ne$ \\
            & $|~let~X=f(\ldots,nt,\ldots)~in~e$ \\
            & $|~let~X=c(\overline{e})~in~e$ \\
            & $|~let~X=(let~Y=e'~in~e'')~in~e$ \\
\end{tabular}

We prove by induction on the structure of an expression $ne$ that it is always possible to perform a $\fnf$ step:

{\bf Base case:}
\begin{itemize}
\item $ne \equiv h(\ldots, nt, \ldots)$: there are various cases depending on $nt$:
    \begin{itemize}
        \item at some depth the non-cterm will contain a subexpression $c'(\ldots, nt', \ldots)$ where $nt'$ is a function application $f(\overline{e})$ or a let-rooted expression $let~X=e_1~in~e_2$. Therefore we can apply the rule \lrrule{Contx} with \lrrule{LetIn} in that position.
        \item $f(\overline{e})$: we can apply the rule \lrrule{LetIn} and perform the step $$h(\ldots, f(\overline{e}), \ldots) \fnf let~X=f(\overline{e})~in~h(\ldots, X, \ldots)$$
        \item $let~X=e_1~in~e_2$: the same as the previous case.
    \end{itemize}
\item $let~X=f(\ldots,nt,\ldots)~in~e$: we can perform a \lrrule{Contx} with \lrrule{LetIn} step in $f(\ldots,nt,\ldots)$ as in the previous $h(\ldots, nt, \ldots)$ case.
\item $let~X=c(\overline{e})~in~e$: if $\overline{e}$ are cterms $\overline{t}$, then $c(\overline{t})$ is a cterm and we can perform a \lrrule{Bind} step $let~X=c(\overline{t})~in~e \fnf e[X/c(\overline{t})]$. If $\overline{e}$ contains any expression $ne$ then we can perform a \lrrule{Contx} with \lrrule{LetIn} step as in the previous $h(\ldots, nt, \ldots)$ case.
\item $let~X=(let~Y=e'~in~e'')~in~e$: by the variable convention we can assume that $Y \notin FV(e)$, so we can perform a \lrrule{Flat} step $let~X=(let~Y=e'~in~e'')~in~e \fnf let~Y=e'~in~let~X=e''~in~e$.
\end{itemize}

{\bf Inductive step:}
\begin{itemize}
    \item $let~X=f(\overline{t})~in~ne$: by IH we have that $ne \fnf ne'$, so by the rule \lrrule{Contx} we can perform a step $let~X=f(\overline{t})~in~ne \fnf let~X=f(\overline{t})~in~ne'$.
\end{itemize}

Notice that if the original expression has the shape $h(e_1, \ldots, e_n)$ the arguments $e_i$ which are cterms remain unchanged in the same position. The reason is that no rule can affect them: the only rule applicable at the top is \lrrule{LetIn}, and it can not place them in a let binding outside $h(\ldots)$; besides cterms do not match with the left-hand side of any rule, so they can not be rewritten by any rule.
\end{proof}

\teoremi{Lemma \ref{LCascCrec} (Growing of shells)}
Under any program $\prog$ and for any $e, e' \in LExp$
\begin{enumerate}
    \item[i)] $e \fe e'$ implies $|e| \sqsubseteq |e'|$
    \item[ii)] $e \fnfe e'$ implies $|e| \equiv |e'|$
\end{enumerate}

\begin{proof}[Proof for Lemma \ref{LCascCrec}]\label{DEMO_LCascCrec}
We prove the lemma for one step ($e \f e'$ and $e \fnf e'$) by a case distinction over the rule of the let-rewriting calculus applied:
\begin{description}
\item[(Fapp)] The step is $f(t_1, \ldots, t_n) \f r$, and $|f(t_1, \ldots, t_n)| = \perp \sqsubseteq |r|$.

\item[(LetIn)] The equality $|h(e_1,\ldots, e, \ldots, e_n)| = |let~X=e~in~h(e_1,\ldots, X, \ldots, e_n)|$ follows easily by a case distinction on $h$.

\item[(Bind)] The step is $let~X=t~in~e \f e[X/t]$, so $|let~X=t~in~e| = |e|[X/|t|] = |e[X/t]|$ by Lemma \ref{LCasc1}.

\item[(Elim)] The step is $let~X=e_1~in~e_2 \f e_2$ with $X \notin FV(e_2)$. Then $|let~X=e_1~in~e_2| = |e_2|[X/|e_1|] = |e_2|$. Since the variables in the shell of an expression is a subset of the variables in the original expression, we can conclude that if $X \notin FV(e_2)$ then $X \notin FV(|e_2|)$.

\item[(Flat)] The step is $let~X = (let~Y =e_1~in~e_2)~in~e_3 ~\f~ let~Y = e_1~in~(let~X=e_2~in~e_3)$ with $Y \notin FV(e_3)$. By the variable convention we can assume that $X \notin FV(let~Y =e_1~in~e_2)$ ---in particular $X \notin FV(e_1)$. Then:
        $$
        \begin{array}{ll}
        |let~Y = e_1~in~(let~X=e_2~in~e_3)| & \\
        = |let~X=e_2~in~e_3|[Y/|e_1|] & \\
        = (|e_3|[X/|e_2|])[Y/|e_1|] & \\
        \end{array}
        $$
        Notice that $X \notin dom([Y/|e_1|])$ and $X \notin \vran([Y/|e_1|]) = FV(|e_1|)$ because $X \notin FV(e_1)$ and $FV(|e_1|) \subseteq FV(e_1)$. Therefore we can use Lemma \ref{auxBind}:
        $$
        \begin{array}{ll}
        (|e_3|[X/|e_2|])[Y/|e_1|] & \\
        = (|e_3|[Y/|e_1|])[X/(|e_2|[Y/|e_1|])] & \mbox{By Lemma \ref{auxBind}} \\
        = |e_3|[X/(|e_2|[Y/|e_1|])] & Y \notin FV(e_3) \mbox{, so } Y \notin FV(|e_3|) \\
        = |e_3|[X/|let~Y =e_1~in~e_2|] & \\
        = |let~X = (let~Y =e_1~in~e_2)~in~e_3| & \\
        \end{array}
        $$

\item[(Contx)] The step is $\cntx[e] \f \cntx[e']$ with $e \f e'$ using any of the previous rules. Then we have $|e| \sqsubseteq |e'|$, and by Lemma \ref{LLetOrd3} $\cntx[e] \sqsubseteq \cntx[e']$. If the step is $\cntx[e] \fnf \cntx[e']$ then rule (Fapp) has not been used in the reduction $e \fnf e'$ and by the previous rules we have $|e| = |e'|$. In that case by Lemma \ref{lemma:context} we have $\cntx[e] = \cntx[e']$.
\end{description}
The extension of this result to $\fe$ and $\fnfe$ is a trivial induction over the number of steps of the derivation.
\end{proof}

\subsection{Proofs for Section \ref{crwllet}}

\teoremi{Theorem \ref{thEquivCrwlCrwllet} (\crwl\ vs. \crwll)}
For any program $\prog$ without lets, and any $e \in Exp_\perp$:
$$
\dc{e}^{\prog} = \dcl{e}^{\prog}
$$
\begin{proof}
As any calculus rule from \crwl\ is also a rule from \crwll, then any \crwl-proof is also a \crwll-proof, therefore $\dc{e} \subseteq \dcl{e}$.
For the other inclusion, assume no let-binding is present in the program and let $e \in Exp$. Then, for any $t \in CTerm_\perp$, as the rules of \crwll\ do not introduce any let-binding and the rule (Let) is only used for let-rooted expressions, the \crwll-proof $\prog \vdcrwll e \clto t$ will be also a \crwl-proof for $\prog \vdcrwll e \clto t$, hence $\dcl{e} \subseteq \dc{e}$ too.
\end{proof}

%

The following Lemma is used to prove point \emph{iii)} of Lemma \ref{lemmashells}. Notice that this Lemma uses the notions of hyperdenotation ($\denn{~}$) and hyperinclusion ($\ohs$) presented in the final part of Section \ref{crwllet}.

\begin{lemma}\label{LemEx2iii}
Under any program $\mathcal{P}$ and for any $e \in LExp_\perp$ we have that $\denn{e} \leqhyp \lambda \theta.(|e\theta|\!\!\uparrow)\!\!\downarrow$.
\end{lemma}
\begin{proof}
We will use the following equivalent characterization of $(e\!\!\uparrow)\!\!\downarrow$:
$$
(e\!\!\uparrow)\!\!\downarrow = \{e_1 \in LExp_\perp~|~\exists e_2 \in LExp_\perp.~e \ordap e_2 \wedge e_1 \ordap e_2\}
$$
note that $\{e_2 \in LExp_\perp~|~e \ordap e_2 \}$ is precisely the set $e\!\!\uparrow$. Besides note that:
$$
\begin{array}{l}
\denn{e} \leqhyp \lambda \theta.(|e\theta|\!\!\uparrow)\!\!\downarrow \\
\Leftrightarrow \forall \theta \in CSubst_\perp.~\den{e\theta} \subseteq (|e\theta|\!\!\uparrow)\!\!\downarrow \\
\Leftrightarrow \forall \theta \in CSubst_\perp, t \in CTerm_\perp.~e\theta \clto t\\
~~~~~~\Rightarrow t \in (|e\theta|\!\!\uparrow)\!\!\downarrow  \\
\Leftrightarrow \forall \theta \in CSubst_\perp, t \in CTerm_\perp.~e\theta \clto t\\
~~~~~~\Rightarrow \exists t' \in CTerm_\perp.~ |e\theta| \ordap t' \wedge t \ordap t'
\end{array}
$$
where $t' \in CTerm_\perp$ is implied by $|e\theta| \ordap t'$. To prove this last formulation first consider the case when $t \equiv \perp$. Then we are done with $t' \equiv |e\theta|$ because then $|e\theta| \ordap |e\theta| \equiv t'$ and $t \equiv \perp \ordap |e\theta| \equiv t'$.

For the other case we proceed by induction on the structure of $e$. Regarding the base cases:
\begin{itemize}
    \item If $e \equiv \perp$ then $t \equiv \perp$ and we are in the previous case.
    \item If $e \equiv X \in \var$ then $e\theta \equiv \theta(X) \clto t$, and as $\theta \in CSubst_\perp$ then $\theta(X) \in CTerm_\perp$ which implies $t \ordap \theta(X)$ by Lemma \ref{lemmashells}. But then we can take $t' \equiv \theta(X)$ for which $t \ordap \theta(X) \equiv t'$ and $|e\theta| \equiv |\theta(X)| \equiv \theta(X)$ ---by Lemma \ref{shellPatterns} since $\theta(X) \in CTerm_\perp$---, and $\theta(X) \ordap \theta(X) \equiv t'$.
    \item If $e \equiv c \in DC$ then either $t \equiv \perp$ and we are in the previous case, or $t \equiv c$. But then we can take $t' \equiv c$ for which $|e\theta| \equiv c \ordap c \equiv t'$, and $t \equiv c \ordap c \equiv t'$.
    \item If $e \equiv f \in FS$ then $|e\theta| \equiv |f| \equiv \perp$, and so $|e\theta|\!\!\uparrow = CTerm_\perp$ and $(|e\theta|\!\!\uparrow)\!\!\downarrow = CTerm_\perp \supseteq \den{e\theta}$, so we are done.
\end{itemize}
Concerning the inductive steps:
\begin{itemize}
    \item If $e \equiv f(e_1, \ldots, e_n)$ for $f \in FS$ then $|e\theta| \equiv \perp$ and we proceed like in the case for $e \equiv f$.
    \item If $e \equiv c(e_1, \ldots, e_n)$ for $c \in DC$ then either $t \equiv \perp$ and we are in the previous case, or $t \equiv c(t_1, \ldots, t_n)$ such that $\forall i.~e_i\theta \clto t_i$. But then by IH we get $\forall i.~\exists t'_i.~|e_i\theta| \ordap t'_i \wedge t_i \ordap t'_i$, so we can take $t' \equiv c(t'_1, \ldots, t'_n)$ for which $|e\theta| \equiv c(|e_1\theta|, \ldots, |e_n\theta|) \ordap c(t'_1, \ldots, t'_n) \equiv t'$ and $t \equiv c(t_1, \ldots, t_n) \ordap c(t'_1, \ldots, t'_n) \equiv t'$.
    \item If $e \equiv let~X=e_1~in~e_2$ then either $t \equiv \perp$ and we are in the previous case, or we have the following proof:
$$
\infer[Let]{e\theta \equiv let~X=e_1\theta~in~e_2\theta \clto t}
           {
          e_1\theta \clto t_1
        \ ~~e_2\theta[X/t_1] \clto t
            }
$$

Then by IH over $e_1$ we get that $\exists t'_1.~|e_1\theta| \ordap t'_1 \wedge t_1 \ordap t'_1$. Hence $[X/t_1] \ordap [X/t'_1]$ so by Proposition \ref{PropMonSubstCrwlLet} we have that $e_2\theta[X/t_1] \clto t$ implies $e_2\theta[X/t'_1] \clto t$. But then we can apply the IH over $e_2$ with $\theta[X/t'_1]$ to get some $t' \in CTerm_\perp$ such that $t \ordap t'$ and $|e_2\theta[X/t'_1]| \ordap t'$, which implies:

$$
\begin{array}{ll}
t' \sqsupseteq |e_2\theta[X/t'_1]| \\
\equiv |e_2\theta|[X/|t'_1|] & \mbox{by Lemma \ref{LCasc1}} \\
\equiv |e_2\theta|[X/t'_1] & \mbox{by Lemma \ref{shellPatterns} as $t'_1 \in CTerm_\perp$} \\
\sqsupseteq |e_2\theta|[X/|e_1\theta|] & \mbox{as $|e_1\theta| \ordap t'_1$} \\
\equiv |let~X=e_1\theta~in~e_2\theta| \equiv |e\theta|
\end{array}
$$


\end{itemize}
\end{proof}

\teoremi{Lemma \ref{lemmashells}}
For any program  $e \in LExp_\perp$, $t,t' \in CTerm_\perp$:
  \begin{enumerate}
  \item $ t \clto t'$ iff $t' \ordap t$.
  \item $|e|\in\den{e}$.
  \item $\den{e} \subseteq (|e|\!\!\uparrow)\!\!\downarrow$, where for a given $E \subseteq LExp_\bot$ its upward closure is $E\!\!\uparrow = \{e' \in LExp_\bot |~\exists e \in E.~e \ordap e'\}$, its downward closure is $E\!\!\downarrow = \{e' \in LExp_\bot |~ \exists e \in E.~e' \ordap e\}$, and those operators are overloaded for let-expressions as $e\!\!\uparrow = \{e\}\!\!\uparrow$ and $e\!\!\downarrow = \{e\}\!\!\downarrow$.
\end{enumerate}

\begin{proof}\label{DEMO_lemmashells}
\begin{enumerate}
  \item Easily by induction on the structure of $t$.
  \item Straightforward by induction on the structure of $e$. In the case of let expressions, the proof uses $|e| \in CTerm_{\perp}$ and Proposition \ref{closednessCSubst} in order to apply the  \crwll ~rule (Let).
  \item By Lemma \ref{LemEx2iii} we have that $\denn{e} \leqhyp \lambda \theta.(|e\theta|\!\!\uparrow)\!\!\downarrow$. By definition of hyperinclusion ---Definition \ref{def:setFunOp}--- we know that $\denn{e}\epsilon \subseteq (\lambda \theta.(|e\theta|\!\!\uparrow)\!\!\downarrow)\epsilon$, so $\denn{e}\epsilon = \den{e\epsilon} \equiv \den{e} \subseteq (|e|\!\!\uparrow)\!\!\downarrow \equiv (|e\epsilon|\!\!\uparrow)\!\!\downarrow = (\lambda \theta.(|e\theta|\!\!\uparrow)\!\!\downarrow)\epsilon$.
\end{enumerate}
\end{proof}

\teoremi{Proposition \ref{propCrwlletPolar} (Polarity of \crwll)}
For any program  $e, e' \in LExp_\perp$, $t, t' \in CTerm_\perp$, if $e \ordap e'$ and $t' \ordap t$ then $ e \clto t$ implies $ e' \clto t'$ with a proof of the same size or smaller---where the size of a \crwll-proof is measured as the number of rules of the calculus used in the proof.

\begin{proof}\label{DEMO_propCrwlletPolar}
By induction on the size of the \crwl-derivation. All the cases are straightforward except the \clrule{Let} rule:
\begin{description}
    \item[\clrule{Let}] We have the derivation:
     $$
     \infer[(Let)]{e \equiv let~X=e_1~in~e_2 \clto t}
         {e_1 \clto t_1
         &
         e_2[X/t_1] \clto t
         }
     $$
     Since $e \ordap e'$ then $e' \equiv let~X=e_1'~in~e_2'$ with $e_1 \ordap e'_1$ and $e_2 \ordap e'_2$. As $e_1 \ordap e'_1$ and $t_1 \ordap t_1$ ---because $\ordap$ is reflexive--- then by IH we have $ e'_1 \clto t_1$. We know that $e_2 \ordap e'_2$ so by Lemma \ref{lemaAproxTheta} we have $e_2[X/t_1] \ordap e'_2[X/t_1]$ and by IH $\prog \vdcrwll e'_2[X/t_1] \clto t'$ such that $t' \ordap t$. Therefore:
     $$
     \infer[(Let)]{e' \equiv let~X=e'_1~in~e'_2 \clto t'}
         {e'_1 \clto t_1
         &
         e'_2[X/t_1] \clto t'
         }
     $$
     \end{description}
\end{proof}

\teoremi{Proposition \ref{closednessCSubst} (Closedness under c-substitutions)}
For any   $e \in LExp_\perp$, $t \in CTerm_\perp$, $\theta \in CSubst_\perp$, $t \in \den{e}$ implies $t\theta \in \den{e\theta}$.

\begin{proof}\label{DEMO_closednessCSubst}
By induction on the size of the \crwll-proof. All the cases are straightforward except the \clrule{Let} rule:
\begin{description}
    \item[\clrule{Let}] In this case the expression is $e \equiv let~X=e_1~in~e_2$ so we have a derivation
    $$
    \infer[\clrule{Let}]
        {let~X=e_1~in~e_2 \clto t}
        {e_1 \clto t_1 & e_2[X/t_1] \clto t}
    $$
    By IH we have that $e_1\theta \clto t_1\theta$ and $(e_2[X/t_1])\theta \clto t\theta$. By the variable convention we assume that $X \notin dom(\theta)\cup \vran(\theta)$, so by Lemma \ref{auxBind} $e_2[X/t_1]\theta \equiv e_2\theta[X/t_1\theta]$ and $e_2\theta[X/t_1\theta] \clto t\theta$. Then we can construct the proof:
    $$
    \infer[\clrule{Let}]
        {let~X=e_1\theta~in~e_2\theta \clto t\theta}
        {e_1\theta \clto t_1\theta & e_2\theta[X/t_1\theta] \clto t\theta}
    $$
    \end{description}

\end{proof}

\teoremi{Theorem \ref{lem:weakcomp} (Weak Compositionality of \crwll)}
For any $\con \in Cntxt$, $e \in LExp_\perp$
$$
\den{\con[e]} = \bigcup\limits_{t \in \den{e}} \den{\con[t]}\qquad \mbox{if } BV(\con) \cap FV(e) = \emptyset
$$
As a consequence, $\den{let~X=e_1~in~e_2} = \bigcup_{t_1 \in \den{e_1}}\den{e_2[X/t_1]}$.

\begin{proof}\label{DEMO_lem:weakcomp}
We prove that $\con[e] \clto t \Leftrightarrow \exists s \in CTerm_{\perp}$ such that $e \clto s$ and $\con[s] \clto t$.
~\\~\\
$\Rightarrow$) By induction on the size of the proof for $\con[e] \clto t$. The proof proceeds in a similar way to the proof for Theorem \ref{thCompoCrwl}, page \pageref{DEMO_thCompoCrwl}, so we only have to prove the \clrule{Let} case:
%
%
%
\begin{description}
\item[\clrule{Let}] There are two cases depending on the context $\con$ (since $\con \neq [\ ]$):
    \begin{itemize}
    \item $\con \equiv let~X=C'~in~e_2$) Straightforward.
    \item $\con \equiv let~X=e_1~in~\con'$) The proof is
    $$
    \infer[\clrule{Let}]
        { \con[e] \equiv let~X=e_1~in~\con'[e] \clto t}
        { e_1 \clto t_1 & \con'[e][X/t_1] \clto t}
    $$
    We assume that $X \notin var(t_1)$ by the variable convention, since $X$ is bound in $\con$ and we can rename it freely. Moreover, we assume also that $X \notin BV(\con')$ because $X$ is bound in $\con$, so we could rename the bound occurrences in $\con'$. Therefore $(dom([X/t_1]\cup\vran([X/t_1]))\cap BV(\con') =  \emptyset$ and $\con'[e][X/t_1] \equiv (\con'[X/t_1])[e[X/t_1]]$ by Lemma \ref{T28}. Since $BV(\con) \cap FV(e) = \emptyset$ by the premise and $X \in BV(\con)$ then $X \notin FV(e)$, so $(\con'[X/t_1])[e[X/t_1]] \equiv \con'[X/t_1][e]$. Then by IH $\exists s \in CTerm_{\perp}$ such that $e \clto s$ and $\con'[X/t_1][s] \clto t$. Therefore we can build:
    $$
    \infer[\clrule{Let}]
        { \con[s] \equiv let~X=e_1~in~\con'[s]  \clto t}
        { e_1 \clto t_1 & \con'[s][X/t_1] \equiv^{(*)} \con'[X/t_1][s] \clto t}
    $$
    (*) Using Lemma \ref{T28} as above and the assumption that $X \notin var(s)$ by the variable convention, since $X$ is bound in $\con$ and we can rename it freely.
    \end{itemize}
\end{description}
~\\~\\
$\Leftarrow$) By induction on the size of the proof for $\con[s] \clto t$. As before, the proof proceeds in a similar way to the proof for Theorem \ref{thCompoCrwl}, page \pageref{DEMO_thCompoCrwl}, so we only have to prove the \clrule{Let} case:

%
\begin{description}
    \item[\clrule{Let}] If we use \clrule{Let} then there are two cases depending on the context $\con$ (since $\con \neq [\ ]$):
    \begin{itemize}
    \item $\con = let~X=\con'~in~e_2$) Straighforward.
        \item $\con = let~X=e_1~in~\con'$) then we have $e \clto s$ and
        $$
        \infer[\clrule{Let}]
            {\con[s] \equiv let~X=e_1~in~\con'[s] \clto t}
            { e_1 \clto t_1 & \con'[s][X/t_1] \clto t}
        $$
        By the same reasoning as in the second case of the \clrule{Let} rule of the $\Rightarrow$) part of this theorem, $\con'[s][X/t_1] \equiv \con'[X/t_1][s]$. Then by IH $\con'[X/t_1][e] \clto t$. Again by the same reasoning we have $\con'[e][X/t_1] \equiv \con'[X/t_1][e]$, so we can build the proof:
        $$
        \infer[\clrule{Let}]
            {\con[e] \equiv let~X=e_1~in~\con'[e] \clto t}
            { e_1 \clto t_1 & \con'[e][X/t_1] \equiv \con'[X/t_1][e] \clto t}
        $$
    \end{itemize}
\end{description}

~\\~\\This ends the proof of the main part of the theorem. With respect to the consequence $\dcl{let~X=e_1~in~e_2} = \bigcup_{t_1 \in \dcl{e_1}}\dcl{e_2[X/t_1]}$ we have:
$$
\begin{array}{lr}
\dcl{let~X=e_1~in~e_2} & \\
= \dcl{(let~X=[\ ]~in~e_2)[e_1]} & \\
= \bigcup\limits_{t_1 \in \dcl{e_1}} \dcl{let~X=t_1~in~e_2} & \mbox{by Theorem \ref{lem:weakcomp}} \\
= \bigcup\limits_{t_1 \in \dcl{e_1}} \dcl{e_2[X/t_1]} & \mbox{by Proposition \ref{propFnfPreservHipSem}} \\
\end{array}
$$

In the last step we replace $let~X=t_1~in~e_2$ by $e_2[X/t_1]$ which is a \lrrule{Bind} step of $\fnf$, so by Proposition \ref{propFnfPreservHipSem} it preserves the denotation.
\end{proof}

For Proposition \ref{PropMonSubstCrwlLet}, in this Appendix we prove a generalization of the statement appearing in Section \ref{crwllet} (page \pageref{PropMonSubstCrwlLet}). However, it is easy to check that Proposition \ref{PropMonSubstCrwlLet} in Section \ref{crwllet} follows easily from points {\em 2} and {\em 3} here.

\teoremi{Proposition \ref{PropMonSubstCrwlLet} (Monotonicity for substitutions of \crwll)}
For any program $e \in LExp_\perp$, $t \in CTerm_\perp$, $\sigma, \sigma' \in LSubst_\perp$
\begin{enumerate}
    \item If $\forall X \in \var, s \in CTerm_\perp$ given $ \sigma(X) \clto s$ with size $K$ we also have $\sigma'(X) \clto s$ with size $K' \leq K$, then $ e\sigma \clto t$ with size $L$ implies $ e\sigma' \clto t$ with size $L' \leq L$.
    \item If $\sigma \ordap \sigma'$ then $ e\sigma \clto t$ implies $ e\sigma' \clto t$ with a proof of the same size or smaller.
    \item If $\sigma \dsord \sigma'$ then $\den{e\sigma} \subseteq \den{e\sigma'}$.
\end{enumerate}

\begin{proof}\label{DEMO_PropMonSubstCrwlLet}
\begin{enumerate}
\item If $e \equiv X \in \var$, assume $X\sigma \clto t$, then $X\sigma' \clto t$ with a proof of the same size or smaller, by hypothesis. Otherwise we proceed by induction on the structure of the proof $e\sigma \clto t$.
\begin{description}
 \item[Base cases]~
 \begin{description}
   \item[\clrule{B}] Then $t \equiv \perp$ and $e\sigma' \clto \perp$ with a proof of size $1$ just applying rule \clrule{B}.
   \item[\clrule{RR}] Then $e \in \var$ and we are in the previous case.
   \item[\clrule{DC}] Then $e \equiv c \in CS^0$, as $e \not\in \var$, hence $e\sigma \equiv c \equiv e\sigma'$ and every proof for $e\sigma \clto t$ is a proof for $e\sigma' \clto t$.
  \end{description}
 \item[Inductive steps]~
   \begin{description}
     \item[\clrule{DC}] Then $e \equiv c(e_1, \ldots, e_n)$, as $e \not\in \var$, and we have:
$$
\infer[\clrule{DC}]{e\sigma \equiv c(e_1\sigma, \ldots, e_n\sigma) \clto c(t_1, \ldots, t_n) \equiv t}
          {e_1\sigma \clto t_1 & \ldots & e_n\sigma \clto t_n }
$$
By IH or the proof of the other cases $\forall i \in \{1, \ldots, n\}$ we have $e_i\sigma' \clto t_i$ with a proof of the same size or smaller, so we can built a proof for $e\sigma' \equiv c(e_1\sigma', \ldots, e_n\sigma') \clto c(t_1, \ldots, t_n) \equiv t$ using \clrule{DC}, with a size equal or smaller than the size of the starting proof.
   \item[\clrule{OR}]  Similar to the previous case.
   \item[\clrule{Let}] Then $e \equiv let~X=e_1~in~e_2$, as $e \not\in \var$, and we have:
$$
\infer[\clrule{Let}]
    {let~X=e_1\sigma~in~e_2\sigma \clto t}
    {e_1\sigma \clto t_1 & e_2\sigma[X/t_1] \clto t}
$$
    By IH we have $e_1\sigma \clto t_1$. By the variable convention we assume that $X \notin dom(\sigma)\cup\vran(\sigma)$ and $X \notin dom(\sigma')\cup\vran(\sigma')$. Then it is easy to check that $\forall Y \in \var, s,t \in CTerm_\perp$, given $Y(\sigma[X/t]) \clto s$ with size $K$ we also have $Y(\sigma'[X/t]) \clto s$ with size $K' \leq K$. Then by IH we have $e_2\sigma'[X/t_1] \clto t$. Therefore we can construct a proof with a size equal or smaller than the starting one:
    $$
\infer[\clrule{Let}]
    {let~X=e_1\sigma'~in~e_2\sigma' \clto t}
    {e_1\sigma' \clto t_1 & e_2\sigma'[X/t_1] \clto t}
$$
   \end{description}
 \end{description}

\item By induction on the size of the \crwll-proof. The cases for classical CRWL appear in \cite{vado02}, so we only have to prove the case for the \clrule{Let} rule:
    \begin{description}
    \item[\clrule{Let}] In this case the expression is $e \equiv let~X=e_1~in~e_2$ so we have a proof
    $$
    \infer[\clrule{Let}]
        {let~X=e_1\sigma~in~e_2\sigma \clto t}
        {e_1\sigma \clto t_1 & e_2\sigma[X/t_1] \clto t}
    $$
    By IH we have that $e_1\sigma \clto t_1$. By the variable convention we can assume that $BV(e)\cap(dom(\sigma)\cup \vran(\sigma)) = \emptyset$ and $BV(e)\cap(dom(\sigma')\cup \vran(\sigma')) = \emptyset$. With the previous properties it is easy to see that $\sigma[X/t_1] \ordap \sigma'[X/t_1]$, so by IH $e_2\sigma'[X/t_1] \clto t$. Therefore we can build the proof:
    $$
    \infer[\clrule{Let}]
        {let~X=e_1\sigma'~in~e_2\sigma' \clto t}
        {e_1\sigma' \clto t_1 & e_2\sigma'[X/t_1] \clto t}
    $$

    \end{description}

\item By induction on the structure of $e$:
    \begin{description}
    \item[$e \equiv X \in \var$ -] In this case $\dcl{X\sigma} \subseteq \dcl{X\sigma'}$ because by the hypothesis $\sigma \dsord \sigma'$.
    \item[$e \equiv h(e_1, \ldots, e_n)$ -] Applying Theorem \ref{lem:weakcomp} with $\con \equiv h([\ ], e_2\sigma, \ldots, e_n\sigma)$ we have 
    $\dcl{h(e_1\sigma, \ldots, e_n\sigma)} = \dcl{\con[e_1\sigma]} = \bigcup\limits_{t \in \dcl{e_1\sigma}} \dcl{\con[t]}$ because $BV(\con) = \emptyset$. On the other hand, by Theorem \ref{lem:weakcomp} we also know that 
    $$\begin{array}{lll}
      \dcl{h(e_1\sigma', e_2\sigma, \ldots, e_n\sigma)} & = & \dcl{\con[e_1\sigma']} \\
                                                        & = & \bigcup\limits_{t \in \dcl{e_1\sigma'}} \dcl{\con[t]}
      \end{array}$$
    Since by IH we have $\dcl{e_1\sigma} \subseteq \dcl{e_1\sigma'}$ it is easy to check that $$\bigcup\limits_{t \in \dcl{e_1\sigma}} \dcl{\con[t]} \subseteq \bigcup\limits_{t \in \dcl{e_1\sigma'}} \dcl{\con[t]}$$ so $\dcl{h(e_1\sigma, e_2\sigma, \ldots, e_n\sigma)} \subseteq \dcl{h(e_1\sigma', e_2\sigma, \ldots, e_n\sigma)}$. Using the same reasoning in the rest of subexpressions $e_i\sigma$ we can prove:\\ $\dcl{h(e_1\sigma', e_2\sigma, \ldots, e_n\sigma)} \subseteq \dcl{h(e_1\sigma', e_2\sigma', e_3\sigma \ldots, e_n\sigma)}$\\
$\dcl{h(e_1\sigma', e_2\sigma', e_3\sigma \ldots, e_n\sigma)} \subseteq \dcl{h(\ldots, e_3\sigma', e_4\sigma \ldots, e_n\sigma)}$\\
\ldots \\
$\dcl{\ldots, e_{n-1}\sigma', e_n\sigma)} \subseteq \dcl{h(e_1\sigma', \ldots, e_n\sigma')}$\\
Then by  transitivity of $\subseteq$ we have:\\
 $\dcl{h(e_1, \ldots, e_n)\sigma} \equiv \dcl{h(e_1\sigma, \ldots, e_n\sigma)} \subseteq$
 \\
$\dcl{h(e_1\sigma', \ldots, e_n\sigma')} \equiv \dcl{h(e_1, \ldots, e_n)\sigma'}$.
    \item[$e \equiv let~X=e_1~in~e_2$ -] As Theorem \ref{lem:weakcomp} states, $\dcl{let~X=e_1\sigma~in~e_2\sigma} = \bigcup\limits_{t_1 \in \dcl{e_1\sigma}}\dcl{e_2\sigma[X/t_1]}$. By the Induction Hypothesis we have that $\dcl{e_1\sigma} \subseteq \dcl{e_1\sigma'}$. Due to the variable convention we assume that $X \notin dom(\sigma)\cup\vran(\sigma)$ and $X \notin dom(\sigma')\cup\vran(\sigma')$, so it is easy to check that $\sigma[X/t] \dsord \sigma'[X/t]$ for any $t \in CTerm$. Then by the Induction Hypothesis we know that $\dcl{e_2\sigma[X/t]} \subseteq \dcl{e_2\sigma'[X/t]}$. Therefore 
    $$\begin{array}{rll}
    \dcl{(let~X=e_1~in~e_2)\sigma} & = & \bigcup \limits_{t_1 \in \dcl{e_1\sigma}}\dcl{e_2\sigma[X/t_1]} \\ 
                                   & \subseteq & \bigcup\limits_{t_1 \in \dcl{e_1\sigma'}}\dcl{e_2\sigma'[X/t_1]} \\ 
                                   & = & \dcl{let~X=e_1\sigma'~in~e_2\sigma'} \\ 
                                   & = & \dcl{(let~X=e_1~in~e_2)\sigma'} \\
                                   
    \end{array}$$
    \end{description}
\end{enumerate}
\end{proof}


\teoremi{Theorem \ref{CompHipSem} (Compositionality of hypersemantics)}
For all $\con \in Cntxt$, $e \in LExp_\perp$
\begin{displaymath}
\denn{\con[e]} = \denn{\con}\denn{e}
\end{displaymath}
As a consequence: $\denn{e} = \denn{e'} \Leftrightarrow \forall \con \in Cntxt. \denn{\con[e]} = \denn{\con[e']}$.

\begin{proof}\label{DEMO_CompHipSem}
By induction over the structure of contexts. 
The base case is $\con = []$, so $\denn{\con[e]} = \denn{e} = \denn{[]}\denn{e} = \denn{\con}\denn{e}$, as $\denn{[]}$ is the identity function by definition. Regarding the inductive step:
\begin{itemize}
 \item $\con = h(e_1, \ldots, \con', \ldots, e_n)$: Then
$$
\begin{array}{ll}
\denn{\con}\denn{e} 
= \lambda\theta.\bigcup\limits_{t \in \denn{\con'}\denn{e}\theta}\den{h(e_1\theta, \ldots, t, \ldots, e_n\theta)} \\
= \lambda\theta.\bigcup\limits_{t \in \denn{\con'[e]}\theta}\den{h(e_1\theta, \ldots, t, \ldots, e_n\theta)} & \mbox{ by IH} \\
= \lambda\theta.\bigcup\limits_{t \in \den{(\con'[e])\theta}}\den{h(e_1\theta, \ldots, t, \ldots, e_n\theta)} & \mbox{ by definition} \\
= \lambda\theta.\den{h(e_1\theta, \ldots, (\con'[e])\theta, \ldots, e_n\theta)} & \mbox{ by Lemma \ref{lem:weakcomp}} \\
= \lambda\theta.\den{(\con[e])\theta}
= \denn{\con[e]}
\end{array}
$$
 \item $\con = let~X=\con'~in~s$: Then
$$
\begin{array}{ll}
\denn{\con}\denn{e} 
= \lambda\theta.\bigcup\limits_{t \in \denn{\con'}\denn{e}\theta}\den{let~X = t~in~s\theta} & \mbox{ by definition}\\
= \lambda\theta.\bigcup\limits_{t \in \denn{\con'}\denn{e}\theta}\den{s\theta[X/t]} &
\mbox{ by rule (Bind)$^{(*)}$ }\\
= \lambda\theta.\bigcup\limits_{t \in \denn{\con'[e]}\theta}\den{s\theta[X/t]} & \mbox{ by IH}\\
= \lambda\theta.\bigcup\limits_{t \in \den{(\con'[e])\theta}}\den{s\theta[X/t]} & \mbox{ by definition}\\
= \lambda\theta.\den{let~X=(\con'[e])\theta~in~s\theta} & \mbox{ by Lemma \ref{lem:weakcomp}} \\
= \denn{\con[e]}
\end{array}
$$
(*): by Proposition \ref{propFnfPreservHipSem} $\den{let~X = t~in~s\theta} = \den{s\theta[X/t]}$ since $let~X = t~in~s\theta \fnf s\theta[X/t]$. 
 \item $\con = let~X=s~in~\con'$: Then
$$
\begin{array}{ll}
\denn{\con}\denn{e} 
= \lambda\theta.\bigcup\limits_{t \in \denn{s}\theta}\denn{\con'}\denn{e}(\theta[X/t]) \\
= \lambda\theta.\bigcup\limits_{t \in \denn{s}\theta}\denn{\con'[e]}(\theta[X/t]) & \mbox{ by IH} \\
= \lambda\theta.\bigcup\limits_{t \in \denn{s}\theta}\den{(\con'[e])(\theta[X/t])} & \mbox{ by definition} \\
= \lambda\theta.\bigcup\limits_{t \in \den{s\theta}}\den{(\con'[e])(\theta[X/t])} & \mbox{ by definition} \\
= \lambda\theta.\bigcup\limits_{t \in \den{s\theta}}\den{((\con'[e])\theta)[X/t]} \\
= \lambda\theta.\den{let~X=s\theta~in~(\con'[e])\theta} & \mbox{ by Lemma \ref{lem:weakcomp}} \\
= \denn{\con[e]}
\end{array}
$$
\end{itemize}
\end{proof}


\teoremi{Proposition \ref{HipSemDecUnion}}
Consider two sets $A,B$, and let ${\cal F}$ be the set of functions $A \rightarrow\partes{B}$. Then:
\begin{enumerate}
    \item[i)] $\ohs$ is indeed a partial order on ${\cal F}$, and $\sd{f}$ is indeed a decomposition of $f\in {\cal F}$, i.e., $\Uhs{(\sd{f})} = f$.
    \item[ii)] Monotonicity of hyperunion wrt. inclusion: for any ${\cal I}_1,{\cal I}_2 \subseteq {\cal F}$
$$
{\cal I}_1 \subseteq {\cal I}_2 \mbox{ implies } \Uhs{\cal I}_1 \ohs \Uhs{\cal I}_2
$$
    \item[iii)] Distribution of unions: for any ${\cal I}_1,{\cal I}_2 \subseteq {\cal F}$
$$
\Uhs{({\cal I}_1 \cup {\cal I}_2)} = (\Uhs{\cal I}_1) \uhs (\Uhs{\cal I}_2)
$$
    \item[iv)] Monotonicity of decomposition wrt. hyperinclusion: for any $f_1, f_2 \in {\cal F}$
$$
f_1 \ohs f_2 \mbox{ implies } \sd{f_1} \subseteq \sd{f_2}
$$
\end{enumerate}

\begin{proof}\label{DEMO_HipSemDecUnion}~
\begin{enumerate}
    \item[i)] The binary relation $\ohs$ is a partial order on ${\cal F}$ because:
    \begin{itemize}
        \item It is reflexive, as for any function $f$ and any $x \in A$ we have that $f(x) = f(x)$, and thus $f(x) \subseteq f(x)$, therefore $f \ohs f$.
        \item It is transitive because given some functions $f_1, f_2, f_3$ such that $f_1 \ohs f_2$ and $f_2 \ohs f_3$, then for any $x \in A$ we have $f_1(x) \subseteq f_2(x) \subseteq f_3(x)$ by definition of $\ohs$, hence $f_1 \ohs f_3$.
        \item It is antisymmetric \wrt\ extensional function equality, because for any pair of hypersemantics $f_1, f_2$ such that $f_1 \ohs f_2$ and $f_2 \ohs f_1$ and any $x \in A$ we have that $f_1(x) \subseteq f_2(x)$ and $f_2(x) \subseteq f_1(x)$ by definition of $\ohs$, hence $f_1(x) = f_2(x)$ by antisymmetry of $\subseteq$ and $f_1 = f_2$.
    \end{itemize}
In order to prove that $\sd{f}$ is indeed a decomposition of $f\in {\cal F}$ we first perform a little massaging by using the definitions of $\Uhs{}$ and $\sd{}$.
$$
\begin{array}{l}
\Uhs{(\sd{f})} = \Uhs{\{\hl a.\{b\}~|~a \in A, b \in f(a)\}} 
= \lambda x \in A.\bigcup\limits_{a \in A}\bigcup\limits_{b \in f(a)} (\hl a.\{b\}) x
\end{array}
$$
Now we will use the fact that $\ohs$ is a partial order, and therefore it is antisymmetric, so mutual inclusion by $\ohs$ implies equality.

\begin{itemize}
%
    \item \underline{$f \ohs \Uhs{(\sd{f})}$}: Given arbitraries $a \in A$, $b \in f(a)$ then
$$
\begin{array}{ll}
(\Uhs{(\sd{f})})a = \bigcup\limits_{x \in A}\bigcup\limits_{y \in f(x)} (\hl x.\{y\})a \\
\supseteq \bigcup\limits_{y \in f(a)} (\hl a.\{y\})a &  \mbox{ as $a \in A$} \\

= \bigcup\limits_{y \in f(a)} \{y\} \ni b & \mbox{ as $b \in f(a)$}
\end{array}
$$
 \item \underline{$\Uhs{(\sd{f})} \ohs f$}: Given arbitraries $a \in A$, $b \in (\Uhs{(\sd{f})})a$ then we have that $b \in \bigcup\limits_{x \in A}\bigcup\limits_{y \in f(x)} (\hl x.\{y\})a$, therefore $\exists x \in A, y \in f(x)$ such that $b \in (\hl x.\{y\})a$. But then $a \equiv x$ ---otherwise $(\hl x.\{y\})a = \emptyset$--- and $y \equiv b$ ---because $b \in (\hl x.\{y\})a = \{y\}$---, and so $y \in f(x)$ implies $b \in f(a)$.
\end{itemize}
        \item[ii)] Given an arbitrary $a \in A$ then
$$
\begin{array}{ll}
(\Uhs{\cal I}_1)a = \bigcup\limits_{f \in {\cal I}_1}f(a) & \mbox{ by definition of $\Uhs{}$}\\
\subseteq \bigcup\limits_{f \in {\cal I}_2}f(a) & \mbox{ as } {\cal I}_1 \subseteq {\cal I}_2\\
= (\Uhs{{\cal I}_2})a & \mbox{ by definition of $\Uhs{}$}\\
\end{array}
$$
\item[iii)]
$$
\begin{array}{ll}
\Uhs{({\cal I}_1 \cup {\cal I}_2)} = \lambda a.\bigcup\limits_{f \in ({\cal I}_1 \cup {\cal I}_2)} f(a) & \mbox{ by definition of $\Uhs{}$} \\
= \lambda a.\bigcup\limits_{f \in {\cal I}_1} f(a) ~\cup~ \bigcup\limits_{f \in {\cal I}_2} f(a) \\
= \lambda a.(\Uhs{{\cal I}_1}) a \cup (\Uhs{{\cal I}_2}) a & \mbox{ by definition of $\Uhs{}$} \\
= (\Uhs{{\cal I}_1}) \uhs (\Uhs{{\cal I}_2}) & \mbox{ by definition of $\uhs$}
\end{array}
$$

\item[iv)] Suppose an arbitrary $\hat{\lambda}a.\{b\} \in \Delta f_1$ with $a \in A$ and $b \in f_1(a)$ by definition. Since $f_1 \ohs f_2$ then $f_1(a) \subseteq f_2(a)$. Therefore $b \in f_2(a)$ and $\hat{\lambda}a.\{b\} \in \Delta f_2$.
\end{enumerate}
\end{proof}

\teoremi{Proposition \ref{HipSemDistCntx} (Distributivity under context of hypersemantics union)}
\begin{displaymath}
\denn{\con}(\Uhs{\hdes}) = \Uhss{\hde\!\in\!\hdes}\denn{\con}\hde
\end{displaymath}
\begin{proof}\label{DEMO_HipSemDistCntx}
We proceed by induction on the structure of $\con$. Regarding the base case, then $\con = []$ and so:
$$
\begin{array}{ll}
\denn{\con}(\Uhs{\hdes}) = \Uhs{\hdes} & \mbox{ by definition of $\denn{\con}$} \\
= \Uhss{\hde \in H} \hde \\
= \Uhss{\hde \in H} \denn{\con} \hde & \mbox{ by definition of $\denn{\con}$} \\
\end{array}
$$
For the inductive step we have several possibilities.
\begin{itemize}
    \item $\con \equiv h(e_1, \ldots, \con', \ldots, e_n)$: then
$$
\begin{array}{ll@{}}
\denn{\con}(\Uhs{\hdes}) = \lambda\theta.\bigcup\limits_{t \in \denn{\con'}(\Uhs{\hdes})\theta} \den{h(e_1\theta, \ldots, t, \ldots, e_n\theta)} & \mbox{ by definition of $\denn{\con}$} \\
= \lambda\theta.\bigcup\limits_{t \in ((\Uhs{\{\denn{\con'}\hde~|~ \hde \in \hdes\}})\theta)} \den{h(e_1\theta, \ldots, t, \ldots, e_n\theta)} & \mbox{ by IH } \\
= \lambda\theta.\bigcup\limits_{t \in (\bigcup\limits_{\hde \in \hdes}\denn{\con'}\hde\theta)} \den{h(e_1\theta, \ldots, t, \ldots, e_n\theta)} & \mbox{ by definition of $\Uhs{}$ } \\
= \lambda\theta.\bigcup\limits_{\hde \in \hdes} \bigcup\limits_{t \in \denn{\con'}\hde\theta} \den{h(e_1\theta, \ldots, t, \ldots, e_n\theta)} \\
= \lambda\theta.\bigcup\limits_{\hde \in \hdes}\denn{\con}\hde\theta & \mbox{ by definition of $\denn{\con}$} \\
= \Uhss{\hde \in \hdes} \denn{\con}\hde & \mbox{ by definition of $\Uhs{}$}
\end{array}
$$

    \item $\con \equiv let~X=\con'~in~e$: then
$$
\begin{array}{ll}
\denn{\con}(\Uhs{\hdes}) = \lambda\theta.\bigcup\limits_{t \in \denn{\con'}(\Uhs{\hdes})\theta} \den{let~X=t~in~e\theta} & \mbox{ by definition of $\denn{\con}$} \\
= \lambda\theta.\bigcup\limits_{t \in ((\Uhs{\{\denn{\con'}\hde~|~ \hde \in \hdes\}})\theta)} \den{let~X=t~in~e\theta} & \mbox{ by IH } \\
= \lambda\theta.\bigcup\limits_{t \in (\bigcup\limits_{\hde \in \hdes}\denn{\con'}\hde\theta)} \den{let~X=t~in~e\theta} & \mbox{ by definition of $\Uhs{}$ } \\
= \lambda\theta.\bigcup\limits_{\hde \in \hdes} \bigcup\limits_{t \in \denn{\con'}\hde\theta} \den{let~X=t~in~e\theta} \\
= \lambda\theta.\bigcup\limits_{\hde \in \hdes}\denn{\con}\hde\theta & \mbox{ by definition of $\denn{\con}$} \\
= \Uhss{\hde\in\hdes} \denn{\con}\hde & \mbox{ by definition of $\Uhs{}$}
\end{array}
$$

    \item $\con \equiv let~X=e~in~\con'$: then
$$
\begin{array}{ll}
\denn{\con}(\Uhs{\hdes}) = \lambda\theta.\bigcup\limits_{t \in \denn{e}\theta}\denn{\con'}(\Uhs{\hdes})(\theta[X/t]) & \mbox{ by definition of $\denn{\con}$} \\
= \lambda\theta.\bigcup\limits_{t \in \denn{e}\theta}(\Uhs{\{\denn{\con'}\hde~|~ \hde \in \hdes\}})(\theta[X/t]) & \mbox{ by IH } \\
= \lambda\theta.\bigcup\limits_{t \in \denn{e}\theta}\bigcup\limits_{\hde \in \hdes}\denn{\con'}\hde(\theta[X/t]) & \mbox{ by definition of $\Uhs{}$ } \\
= \lambda\theta.\bigcup\limits_{\hde \in \hdes}\bigcup\limits_{t \in \denn{e}\theta}\denn{\con'}\hde(\theta[X/t]) & \mbox{ as $\hdes$ is independent from $t$} \\
= \lambda\theta.\bigcup\limits_{\hde \in \hdes}\denn{\con}\hde\theta & \mbox{ by definition of $\denn{\con}$} \\
= \Uhss{\hde \in \hdes} \denn{\con}\hde & \mbox{ by definition of $\Uhs{}$}
\end{array}
$$

\end{itemize}

\end{proof}

\subsection{Proofs for Section \ref{sectEqLetrwCrwlLet}}

\teoremi{Theorem \ref{T27} (Hyper-Soundness of let-rewriting)}
For all $e, e' \in LExp$, if $e \fe e'$ then $\denn{e'} \ohs \denn{e}$.

\begin{proof}\label{DEMO_T27}

We first prove the theorem for a single step of $\f$. We proceed assumming some $\theta \in CSubst_{\perp}$ such that $e'\theta \crwlto t$
and then proving $e\theta \crwlto t$. The case where $t \equiv \perp$ holds
trivially using the rule \textbf{B}, so we will prove the rest by a case distinction
on the rule of the let-rewriting calculus applied:
\begin{description}
\item[\lrrule{Fapp}] 
Assume $f(t_1, \ldots, t_n) \f r$ with $(f(p_1, \ldots, p_n) \tor e) \in \prog$, $\sigma \in CSubst$, such that $\forall i.p_i\sigma \equiv t_i$ and $e\sigma \equiv r$, and $\theta \in CSubts_{\perp}$ such that $r\theta \crwlto t$. Then as $\sigma\theta \in CSubts_{\perp}, \forall i. p_i\sigma\theta \equiv t_i\theta$ and $e\sigma\theta \equiv r\theta$ we can use the \clrule{OR} rule to build the following proof:
$$
\infer[\clrule{OR}]
    {f(t_1\theta, \ldots, t_n\theta) \crwlto t}
    { \infer[]{t_1\theta \crwlto t_1\theta}{\mbox{Lemma \ref{lemma:crwlPatterns}}}
      & ~~\ldots~~
      & \infer[]{t_n\theta \crwlto t_n\theta}{\mbox{Lemma \ref{lemma:crwlPatterns}}}
      & r\theta \crwlto t
    }
$$

\item[\lrrule{LetIn}] Assume $h(\ldots , e, \ldots) \f let~X=e~in~h(\ldots, X, \ldots)$ by \lrrule{LetIn} and $\theta \in CSubts_{\perp}$ such that $(let~X=e~in~h(\ldots, X, \ldots))\theta \crwlto t$. This proof must be of the shape of:
$$
\infer[\clrule{Let}]
    {let~X=e\theta~in~h(d_1\theta, \ldots, X\theta, \ldots, d_n\theta) \clto t}
    { e\theta \clto t_1~
      & ~h(d_1\theta, \ldots, X\theta, \ldots, d_n\theta)[X/t_1] \clto t
    }
$$
for some $d_1, \ldots, d_n \in LExp, t_1 \in CTerm_\perp$. Besides $X \not\in (dom(\theta) \cup \vran(\theta))$ by the variable convention\footnote{Actually, to prove this theorem properly, we cannot restrict the substitution to fulfill these restrictions, so in fact we rename the bound variables in an $\alpha$-conversion fashion and use the equivalence $e[X/e'] \equiv e[X/Y][Y/e']$ (with $Y$ the new bound variable), to use the hypothesis. This will be done implicitly when needed during the remaining of the proof.}, hence $X\theta \equiv X$ and so $h(d_1\theta, \ldots, X\theta, \ldots, d_n\theta)[X/t_1] \equiv h(d_1\theta, \ldots, t_1, \ldots, d_n\theta)$, as $X$ is fresh by the conditions in \lrrule{LetIn} and so it does not appear in any $d_i$. Now we have two possibilities:
\begin{enumerate}
\item[a)] $h \equiv c \in DC$ : Then $h(d_1\theta, \ldots, t_1, \ldots, d_n\theta) \crwlto t$ must proved by \clrule{DC}:
$$
\infer[\clrule{DC}]
    { c(d_1\theta, \ldots, t_1, \ldots, d_n\theta) \crwlto c(s_1, \ldots, t'_1, \ldots, s_n) \equiv t~~~}
    {
      d_1\theta \crwlto s_1
     \ \ldots
      \ t_1 \crwlto t'_1
      \ \ldots
      \ d_n\theta \crwlto s_n
      }
$$
for some $s_1, \ldots, s_n, t'_1 \in CTerm_{\perp}$. 
Then $t_1 \crwlto t'_1$ implies $t'_1 \sqsubseteq t_1$ by Lemma \ref{lemmashells}, hence $e\theta \crwlto t_1$ implies $e\theta \crwlto t'_1$ by Proposition \ref{propCrwlletPolar}, and we can build the following proof: 
$$
\infer{h(\ldots , e, \ldots)\theta \equiv c(d_1\theta, \ldots, e\theta, \ldots, d_n\theta) \crwlto c(s_1, \ldots, t'_1, \ldots, s_n) \equiv t}
     {
      d_1\theta \crwlto s_1
     \ \ldots
      \ e\theta \crwlto t'_1
      \ \ldots
     \ d_n\theta \crwlto s_n
      }
$$

        \item[b)] $h \equiv f \in FS$ : Then $h(d_1\theta, \ldots, t_1, \ldots, d_n\theta) \crwlto t$ must be proved by \clrule{OR}:
$$
\infer[\clrule{OR}]
    {f(d_1\theta, \ldots, t_1, \ldots, d_n\theta) \crwlto t}
     {
      d_1\theta \crwlto s_1\sigma
      & \ldots
      & t_1 \crwlto t'_1\sigma
      & \ldots
      & d_n\theta \crwlto s_n\sigma
      & r\sigma \crwlto t
     }
$$
for some $s_1\sigma, \ldots, s_n\sigma, t'_1\sigma \in CTerm_{\perp}$, $(f(s_1,\ldots, t'_1, \ldots s_n) \tor r) \in \prog$, $\sigma \in CSubst_\bot$.
Then we can prove $e\theta \crwlto t'_1\sigma$ like in the previous case, to build the following proof:
$$
\infer[\clrule{OR}]
    {h(\ldots , e, \ldots)\theta \equiv f(d_1\theta, \ldots, e\theta, \ldots, d_n\theta) \crwlto t}
     { d_1\theta \crwlto s_1\sigma
     & \ldots
     & e\theta \crwlto t'_1\sigma
     & \ldots
     & d_n\theta \crwlto s_n\sigma
     & r\sigma \crwlto t
     }
$$
\end{enumerate}

\item[\lrrule{Bind}] Assume $let~X=t_1~in~e \f e[X/t_1]$ by \lrrule{Bind} and $\theta \in CSubst_{\perp}$ such that $(e[X/t_1])\theta \crwlto t$. Then $X \not\in (dom(\theta) \cup\vran(\theta))$ by the variable convention, so we can apply Lemma \ref{auxBind} (Substitution lemma) to get $e\theta[X/t_1\theta] \equiv (e[X/t_1])\theta$. Besides $t_1 \in CTerm$ and $\theta \in CSubst_{\perp}$ by hypothesis, hence $t_1\theta \in CTerm_{\perp}$ and we can build the following proof:
$$
\infer[\clrule{Let}]{let~X=t_1\theta~in~e\theta \crwlto t}
   { \infer[]{t_1\theta \crwlto t_1\theta}{\mbox{Lemma \ref{lemma:crwlPatterns}}}
    & e\theta[X/t_1\theta] \equiv (e[X/t_1])\theta \crwlto t
   }
$$

\item[\lrrule{Elim}] Assume $let~X=e_1~in~e_2 \f e_2$ by \lrrule{Elim} and $\theta \in CSubts_{\perp}$ such that $e_2\theta \crwlto t$. Then $X \not\in \vran(\theta)$ by the variable convention 
and $X \not\in FV(e_2)$ by the condition of \lrrule{Elim}, hence $e_2\theta[X/\perp] \equiv e_2\theta$ and we can build the following proof:
$$
\infer[\clrule{Let}]{let~X=e_1\theta~in~e_2\theta \crwlto t}
   { \infer[\clrule{B}]{e_1\theta \crwlto \perp}{}
    & e_2\theta[X/\perp] \equiv e_2\theta \crwlto t
   }
$$

\item[\lrrule{Flat}] Assume $let~X=(let~Y=e_1~in~e_2)~in~e_3 \f let~Y=e_1~in~(let~X=e_2~in~e_3)$ by \lrrule{Flat} and $\theta \in CSubts_{\perp}$ such that
  $(let~Y=e_1~in~(let~X=e_2~in~e_3))\theta
  \crwlto t$. This proof must be must be of the shape of:
$$
\infer[\clrule{Let}]
    {let~Y=e_1\theta~in~(let~X=e_2\theta~in~e_3\theta) \crwlto t}
    {e_1\theta \crwlto t_1 &
     \infer[\clrule{Let}]{(let~X=e_2\theta~in~e_3\theta)[Y/t_1] \crwlto t}
    {e_2\theta[Y/t_1] \crwlto t_2
    & e_3\theta[Y/t_1][X/t_2] \crwlto t
        }
   }
$$
for some $t_1, t_2 \in CTerm_{\perp}$. Besides $Y \not\in \vran(\theta)$ by the variable convention and $Y \not\in FV(e_3)$ by the condition of \lrrule{Flat}, hence $e_3\theta[Y/t_1] \equiv e_3\theta$  and we can build the following proof: 
$$
\infer[\clrule{Let}]
    {let~X=(let~Y=e_1\theta~in~e_2\theta)~in~e_3\theta \crwlto t}
    {   \infer[\clrule{Let}]
            {let~Y=e_1\theta~in~e_2\theta \crwlto t_2}
            {\infer[]{e_1\theta \crwlto t_1}{\mi{Hypothesis}}
         &
         \infer[]
            {e_2\theta[Y/t_1] \crwlto t_2}
            {\mi{Hypothesis}}
          }
     & e_3\theta[X/t_2] \equiv e_3\theta[Y/t_1][X/t_2] \crwlto t
     }\\
$$

 \item[\lrrule{Contx}] 
By the proof of the other cases, $\denn{e'} \leqhyp \denn{e}$, but then $\denn{{\cal C}[e']} \leqhyp \denn{{\cal C}[e]}$ by Lemma \ref{T25}, and we are done.
\end{description}

The proof for several steps is a trivial induction on the length of the derivation $e \fe e'$.
\end{proof}

\teoremi{Proposition \ref{propFnfPreservHipSem} (The $\fnf$ relation preserves hyperdenotation)}
For all $e, e' \in LExp$, if $e ~\fnfe~ e'$ then $\denn{e} = \denn{e'}$---and therefore $\den{e} = \den{e'}$.

\begin{proof}\label{DEMO_propFnfPreservHipSem}
We first prove the lemma for one step of $\fnf$ by case distinction over the rule applied to reduce $e$ to $e'$. By Theorem \ref{T27} we already have that $\forall e, e' \in LExp$ if $e \fnf~ e'$ then $\denn{e'} \leqhyp \denn{e}$, so all that is left is proving that $\denn{e} \leqhyp \denn{e'}$ also, and finally applying the transitivity of $\leqhyp$, as it is a partial order by Lemma \ref{HipSemDecUnion}-i. We proceed assumming some $\theta \in CSubst_{\perp}$ such that $e\theta \crwlto t$ and then proving $e'\theta \crwlto t$. The case where $t \equiv \perp$ holds trivially using the rule \clrule{B}, so we will prove the other by a case distinction
on the rule of the $let$ calculus applied:
\begin{description}

\item[\lrrule{LetIn}] Assume $h(d_1, \ldots , e, \ldots, d_n) \f let~X=e~in~h(d_1, \ldots, X, \ldots, d_n)$ by the \lrrule{LetIn} rule and $\theta \in CSubts_{\perp}$ such that 
$$h(d_1, \ldots , e, \ldots, d_n)\theta \equiv h(d_1\theta, \ldots , e\theta, \ldots, d_n\theta) \clto t$$
Then by the compositionality of Theorem \ref{thCompoCrwllet} we have that $\exists t_1 \in \den{e\theta}$ such that $h(d_1\theta, \ldots , t_1, \ldots, d_n\theta) \clto t$. Besides $X$ is fresh and $X \not\in (dom(\theta) \cup \vran(\theta))$ by the variable convention, hence $$(let~X=e~in~h(d_1, \ldots, X, \ldots, d_n))\theta \equiv let~X=e\theta~in~h(d_1\theta, \ldots, X, \ldots, d_n\theta)$$
and
$$h(d_1\theta, \ldots, X, \ldots, d_n\theta)[X/t_1] \equiv h(d_1\theta, \ldots, t_1, \ldots, d_n\theta)$$
and so we can do:
\small{
$$
\infer[\clrule{Let}]{(let~X=e~in~h(d_1, \ldots, X, \ldots, d_n))\theta \equiv let~X=e\theta~in~h(d_1\theta, \ldots, X, \ldots, d_n\theta) \clto t}
           {
            \infer[]{e\theta \clto t_1}{\mathit{hypothesis}}~
          \ \infer[]{~h(d_1\theta, \ldots, X, \ldots, d_n\theta)[X/t_1] \equiv h(d_1\theta, \ldots, t_1, \ldots, d_n\theta) \clto t}{\mathit{hypothesis}}
            }
$$
}

\item[\lrrule{Bind}] Assume $let~X=t_1~in~e \f e[X/t_1]$ by \lrrule{Bind} and $\theta \in CSubst_{\perp}$ such that $(let~X=t_1~in~e)\theta \equiv let~X=t_1\theta~in~e\theta \crwlto t$. Then it must be with a proof of the following shape:
$$
\infer[\clrule{Let}]{let~X=t_1\theta~in~e\theta \crwlto t}
   { t_1\theta \crwlto t'_1~
   \ ~e\theta[X/t'_1] \crwlto t
   }
$$
But $\theta \in CSubst_{\perp}$ and $t_1 \in CTerm$ implies $t_1\theta \in CTerm_\perp$, and so $t_1\theta \crwlto t'_1$ implies $t'_1 \ordap t_1\theta$ by Lemma \ref{lemmashells}-1. Hence $[X/t'_1] \ordap [X/t_1\theta]$ and so $e\theta[X/t'_1] \crwlto t$ implies $e\theta[X/t_1\theta] \crwlto t$ by the monoticity of Proposition \ref{PropMonSubstCrwlLet}. Besides $X \not\in (dom(\theta) \cup\vran(\theta))$ by the variable convention, and so we can apply Lemma \ref{auxBind} (substitution lemma) to get $(e[X/t_1])\theta \equiv e\theta[X/t_1\theta]$, so we are done.

\item[\lrrule{Elim}] Assume $let~X=e_1~in~e_2 \f e_2$ by \lrrule{Elim} and $\theta \in CSubts_{\perp}$ such that $(let~X=e_1~in~e_2)\theta \equiv let~X=e_1\theta~in~e_2\theta \crwlto t$. Then it must be with a proof of the following shape:
$$
\infer[\clrule{Let}]{let~X=e_1\theta~in~e_2\theta \crwlto t}
   { e_1\theta \crwlto t_1~
   \ ~e_2\theta[X/t_1] \crwlto t
   }
$$
Then $X \not\in \vran(\theta)$ by the variable convention and $X \not\in FV(e_2)$ by the condition of \lrrule{Elim}, hence $e_2\theta \equiv e_2\theta[X/t_1] \crwlto t$, so we are done.

\item[\lrrule{Flat}] Straightforward since $e_3\theta[Y/t_1] \equiv e_3\theta$ because $Y \not\in \vran(\theta)$ by the variable convention and $Y \not\in FV(e_3)$ by the condition of \lrrule{Flat}.

 \item[\lrrule{Contx}] By the proof of the other cases, $\denn{e} \leqhyp \denn{e'}$, but then $\denn{{\cal C}[e]} \leqhyp \denn{{\cal C}[e']}$ by Lemma \ref{T25}, and we are done.
\end{description}
\end{proof}

The following lemmas ---Lemmas \ref{T35}, \ref{FVShells}, \ref{lemmaShellSubst1} and \ref{lemmaShellSubst2}--- will be used to prove Lemma \ref{T32}.

\begin{lemma}\label{T35}
Let linear $e,e_1 \in Exp$ such that $e\theta \ordap e_1$ for $\theta \in Subst_{\perp}$. Then $\exists \theta' \in Subst$ such that $e\theta' \equiv e_1$ and $\theta \ordap \theta'$.
\end{lemma}
\begin{proof}
By induction on the structure of $e$. For the base case ($e \equiv X \in \var$) we define a function $rep_{\perp} : Exp_{\perp} \to Exp \to Exp$ $rep_{\perp}(e,e')$ that replaces the occurrences of $\perp$ in $e$ by the expression $e'$. We define this function recursively on the structure of $e$:
\begin{itemize}
  \item $rep_{\perp}(\perp,e') = e'$
  \item $rep_{\perp}(Z,e') = Z$
  \item $rep_{\perp}(h(e_1, \ldots, e_n),e') = h(rep_{\perp}(e_1,e'), \ldots, rep_{\perp}(e_n,e'))$
\end{itemize}
  It is easy to check that $rep_{\perp}(e,e') = e''$ implies $e \ordap e''$. Then we define $\theta' \in Subst$ as:
  $$
    \theta'(Y) = \left\{    \begin{array}{ll}
                                e_1 & \mathit{if}~X \equiv Y \\
                                rep_{\perp}( \theta(Y), Y) & \mathit{if}~ Y \in dom(\theta) \smallsetminus \{X\} \\
                            \end{array} \right.
  $$
  Trivially $e\theta' \equiv X\theta' \equiv e_1$ and $\theta \ordap \theta'$ because $e\theta \ordap e_1$ by the premise and $\theta(Y) \ordap rep_{\perp}( \theta(Y), Y)$.

Regarding the inductive step ---$e \equiv h(e_1, \ldots, e_n)$--- we know that 
$$e\theta \equiv h(e_1\theta, \ldots, e_n\theta) \ordap e_1 \equiv h(e'_1, \ldots, e'_n)$$
so $e_i\theta \ordap e'_i$. Then by IH $\exists \theta'_i \in Subst$ such that $e_i\theta'_i \equiv e'_i$ and $\theta \ordap \theta'_i$. Then we define $\theta'$ as:
  $$
    \theta'(Y) = \left\{    \begin{array}{ll}
                                \theta'_1(Y) & \mathit{if}~Y \in var(e_1) \\
                                \theta'_2(Y) & \mathit{if}~Y \in var(e_2) \\
                                \ldots & \\
                                \theta'_n(Y) & \mathit{if}~Y \in var(e_n) \\
                                rep_{\perp}( \theta(Y), Y) & \mathit{if}~ Y \in dom(\theta) \smallsetminus var(e) \\
                            \end{array} \right.
  $$
  The substitution $\theta'$ is well defined because $e$ is linear. Then $e\theta' \equiv h(e_1\theta', \ldots, e_n\theta') \equiv h(e_1\theta'_1, \ldots, e_n\theta'_n) = h(e'_1, \ldots, e'_n) \equiv e_1$ and $\theta \ordap \theta'$ by IH and the fact that $\theta(Y) \ordap rep_{\perp}( \theta(Y), Y)$.
\end{proof}



\begin{lemma}\label{FVShells}
For any $e \in LExp_{\perp}$, $FV(|e|) \subseteq FV(e)$.
\end{lemma}
\begin{proof}
Straightforward by induction on the structure of $e$.
\end{proof}

\begin{lemma}\label{lemmaShellSubst1}
Given $e \in LExp$, $\theta \in LSubst_\perp$, $|e\theta| = |e|\hat{\theta}$ where $\hat{\theta}$ is defined as $ X\hat{\theta} = |X\theta|$
\end{lemma}
\begin{proof}
By induction on the structure of $e$. We have two base cases:
\begin{itemize}
  \item $e \equiv X \in \var$. Then $|e\theta| \equiv |X\theta| = X\hat{\theta} = |X|\theta \equiv |e|\hat{\theta}$.
  \item $e \equiv f(e_1, \ldots, e_n)$. Then $|e\theta| \equiv |f(e_1, \ldots, e_n)\theta| = |f(e_1\theta, \ldots, e_n\theta)| = \perp = \perp\hat{\theta} = |f(e_1, \ldots, e_n)|\hat{\theta} \equiv |e|\hat{\theta}$.
\end{itemize}

Regarding the inductive step we have:
\begin{itemize}
  \item $e \equiv c(e_1, \ldots, e_n)$. Straightforward.
  \item $e \equiv let~X=e_1~in~e_2$. Then $|e\theta| = |(let~X=e_1~in~e_2)\theta| = |let~X=e_1\theta~in~e_2\theta| = |e_2\theta|[X/|e_1\theta|]$. By IH we have that $|e_1\theta| = |e_1|\hat{\theta}$ and $|e_2\theta| = |e_2|\hat{\theta}$, so $|e_2\theta|[X/|e_1\theta|] = |e_2\theta| = (|e_2|\hat{\theta})[X/|e_1|\hat{\theta}]$. By the variable convention we can assume that $X \notin dom(\theta)\cup\vran(\theta)$, and since $dom(\hat{\theta}) = dom(\theta)$ and $\vran(\hat{\theta}) \subseteq \vran(\theta)$ ---using Lemma \ref{FVShells}--- we can use Lemma \ref{auxBind} and obtain $(|e_2|\hat{\theta})[X/|e_1|\hat{\theta}] = (|e_2|[X/|e_1|])\hat{\theta}$. Finally, $(|e_2|[X/|e_1|])\hat{\theta} = |let~X=e_1~in~e_2|\hat{\theta} = |e|\hat{\theta}$.
\end{itemize}
\end{proof}

\begin{lemma}\label{lemmaShellSubst2}
Given $e \in LExp$, $\theta \in LSubst_\perp$, if $|e| = \perp$ then $|e\theta| = \perp$.
\end{lemma}
\begin{proof}
By induction on the structure of $e$. Notice that $e$ cannot be a variable $X$ or an applied constructor symbol $c(e_1, \ldots, e_n)$ because in those cases $|e| \neq \perp$. The base case $e \equiv f(e_1, \ldots, e_n)$ is straightforward. Regarding the inductive step we have $e \equiv let~X=e_1~in~e_2$ such that $|let~X=e_1~in~e_2| = |e_2|[X/|e_1|] = \perp$. Then $|e\theta| = |(let X=e_1~in~e_2)\theta| = |let~X=e_1\theta~in~e_2\theta| = |e_2\theta|[X/|e_1\theta|]$. By Lemma \ref{LCasc1} $|e_2\theta|[X/|e_1\theta|] = |(e_2\theta)[X/e_1\theta]|$, and since $X \notin dom(\theta)\cup\vran(\theta)$ by the variable convention then we can apply Lemma \ref{auxBind} and $|(e_2\theta)[X/e_1\theta]| = |(e_2[X/e_1])\theta|$. Finally by Lemma \ref{lemmaShellSubst1} $|(e_2[X/e_1])\theta| = |e_2[X/e_1]|\hat{\theta}$, and by Lemma \ref{LCasc1} $|e_2[X/e_1]|\hat{\theta} = (|e_2|[X/|e_1|]) \hat{\theta} = \perp \hat{\theta} = \perp$.
\end{proof}

\teoremi{Lemma \ref{T32} (Completeness lemma for let-rewriting)}
For all $e \in LExp$ and $t \in CTerm_{\perp}$ such that $t \not\equiv \perp$,
$$
e \clto t \mbox{ implies } e \fe let~\overline{X = a}~in~t'
$$
for some $t' \in CTerm$ and  $\overline{a} \subseteq LExp$ in such a way that $t\sqsubseteq |let~\overline{X=a}~in~t'|$ and $|a_i| = \perp$ for every $a_i \in \overline{a}$. As a consequence, $t \sqsubseteq t'[\overline{X/\perp}]$.

\begin{proof}\label{DEMO_T32}
By induction on the size $s$ of the {\it CRWL}$_{let}$-proof, that we measure as the number of ${\it CRWL}_{let}$ rules applied. Concerning the base cases:
\begin{description}
  \item[\clrule{B}] This contradicts the hypothesis because then $t \equiv
  \perp$, so we are done. In the rest of the proof we will assume that $t
  \not\equiv \perp$ because otherwise we would be in this case.

  \item[\clrule{RR}] Then we have $X \crwlto X$. But then $X \f^0 X$ and
  $X \sqsubseteq X \equiv |X|$, so we are done with $\overline{X} =
  \emptyset$.

  \item[\clrule{DC}] Then we have $c \crwlto c$. But then $c \f^0 c$ and
      $c \sqsubseteq c \equiv |c|$, so we are done with $\overline{X} =
      \emptyset$.
\end{description}

Now we treat the inductive step:
\begin{description}
 \item[\clrule{DC}] Then we have $e \equiv c(e_1,\ldots,e_n)$ and the
    {\it CRWL}$_{let}$-proof has the shape:
    $$
    \infer[\clrule{DC}]{c(e_1,\ldots,e_n)\crwlto c(t_1,\ldots,t_n)}{e_1\crwlto
      t_1,\ldots,e_n\crwlto t_n}
    $$
    In the general case some $t_i$ will be equal to $\perp$ and some others will be different. 
For the sake of simplicity we consider the case when $n=2$ with $t_1=\bot$ and $t_2\not\equiv \bot$, the proof can be easily extended to the general case.
Then  we have $c(e_1, e_2) \crwlto c(\perp, t_2)$, so by IH over the second
    argument we get $e_2 \fe let~\overline{X_2=a_2}~in~t'_2$
    with $t'_2 \in CTerm$, $|a_{2_i}| = \perp$ for every $a_{2_i} \in \overline{a_2}$ and
    $|let~\overline{X_2=a_2}~in~t'_2| =
    t'_2[\overline{X_2/\perp}] \sqsupseteq t_2$. So:
    $$
    \begin{array}{ll}
      c(e_1, e_2) \fe c(e_1, let~\overline{X_2=a_2}~in~t'_2) & \mbox{ by IH} \\
      \f let~Y=(let~\overline{X_2=a_2}~in~t'_2)~in~c(e_1, Y) & \mbox{ by \lrrule{LetIn}} \\
      \fe let~\overline{X_2=a_2}~in~let~Y=t'_2~in~c(e_1, Y) & \mbox{ by \lrrule{Flat*}} \\
      \f let~\overline{X_2=a_2}~in~c(e_1, t'_2) & \mbox{ by \lrrule{Bind}} \\
    \end{array}
    $$
    Then there are several possible cases:
    \begin{enumerate}
    \item[a)] $e_1 \equiv f_1(\overline{e_1})$: Then
    $let~\overline{X_2=a_2}~in~c(f_1(\overline{e_1}), t'_2) \f
    let~\overline{X_2=a_2}~in~let~Z=f_1(\overline{e_1})~in~c(Z, t'_2)$, by
    \lrrule{LetIn}. So we are done as  $|a_{2_i}| = \perp$ for every
    $a_{2_i}$ by the IH, $|f_1(\overline{e_1})| = \perp$ and
$|let~\overline{X_2=a_2}~in~let$ $Z=f_1(\overline{e_1})~in~c(Z, t'_2)| = c(Z,t'_2)[\overline{X_2/\perp}, Z/\perp] \sqsupseteq c(\perp, t_2)$
    because
    $t'_2[\overline{X_2/\perp}] \sqsupseteq t_2$ by the IH, and $Z$ is fresh and
    so it does not appear in $t'_2$
    \item[b)] $e_1 \equiv t'_1 \in CTerm$: Then we are done as $|a_{2_i}| = \perp$
    for every $a_{2_i} \in \overline{a_2}$ by the IH, and $| let~\overline{X_2=a_2}~in~c(t'_1,
    t'_2)| = c(t'_1, t'_2)[\overline{X_2/\perp}] \sqsupseteq c(\perp, t_2)$,
    because $t'_2[\overline{X_2/\perp}] \sqsupseteq t_2$ by the IH
    \item[c)] $e_1 \equiv c_1(\overline{e_1}) \not\in CTerm$ with $c_1 \in CS$: 
    Then by Lemma \ref{T36} we have the derivation $c_1(\overline{e_1}) \fe let~\overline{X_1=f_1(\overline{t'_1})}~in~c_1(\overline{t_1})$.
    But then:
    $$
    \begin{array}{@{}l@{}l@{}}
      let~\overline{X_2=a_2}~in~c(c_1(\overline{e_1}), t'_2)\\
      \fe let~\overline{X_2=a_2}~in~c(let~\overline{X_1=f_1(\overline{t'_1})}~in~c_1(\overline{t_1}),t'_2) & \mbox{ Lemma \ref{T36}} \\
      \f let~\overline{X_2=a_2}~in~let~Y=(let~\overline{X_1=f_1(\overline{t'_1})}~in~c_1(\overline{t_1}))~in~c(Y, t'_2) & \mbox{ by \lrrule{LetIn}} \\
      \fe let~\overline{X_2=a_2}~in~let~\overline{X_1=f_1(\overline{t'_1})}~in~let~Y=c_1(\overline{t_1})~in~c(Y, t'_2) & \mbox{ by \lrrule{Flat*}} \\
      \f let~\overline{X_2=a_2}~in~let~\overline{X_1=f_1(\overline{t'_1})}~in~c(c_1(\overline{t_1}),t'_2) & \mbox{ by \lrrule{Bind}} \\
    \end{array}
    $$
    In the last step notice that $Y$ is fresh and it cannot appear in $t'_2$. Then we are done as $|f_i(\overline{t'_i})| = \perp$, $|a_{2_i}| = \perp$ for every $a_{2_i} \in \overline{a_2}$ by the IH, and $|let~\overline{X_2=a_2}~in~let~\overline{X_1=f_1(\overline{t'_1})}~in~c(c_1(\overline{t_1}),t'_2)| = c(c_1(\overline{t_1}),t'_2)[\overline{X_1/\perp}][\overline{X_2/\perp}]$ $\sqsupseteq c(\perp,t_2)$
%
    because $t'_2[\overline{X_2/\perp}] \sqsupseteq t_2$ by the IH, and no variable in $\overline{X_1}$ appears in $t'_2$ by $\alpha$-conversion, as those are bound variables which were present in $c_1(\overline{e_1})$ or that appeared after applying Lemma \ref{T36} to it, and this expression was placed in a position parallel to the position of $t'_2$.
  \item[d)] $e_1 \equiv let~X=e_{11}~in~e_{12}$: 
Then by Lemma \ref{T36} $let~X=e_{11}~in~e_{12} \fe let~\overline{X_1=f_1(\overline{t'_1})}~in~e''$ where $e'' \in \var$ or $e'' \equiv h_1(\overline{t_1})$. Then:

$$
\begin{array}{@{}ll@{}}
let~\overline{X_2=a_2}~in~c(let~X=e_{11}~in~e_{12}, t'_2) &  \\
\fe let~\overline{X_2=a_2}~in~c(let~\overline{X_1=f_1(\overline{t'_1})}~in~e'', t'_2) & \mbox{by Lemma \ref{T36}} \\
\f let~\overline{X_2=a_2}~in~let~Y=(let~\overline{X_1=f_1(\overline{t'_1})}~in~e'')~in~c(Y, t'_2) & \mbox{by \lrrule{LetIn} } \\
\fe let~\overline{X_2=a_2}~in~let~\overline{X_1=f_1(\overline{t'_1})}~in~let~Y=e''~in~c(Y, t'_2) & \mbox{by \lrrule{Flat$^*$}} \\
\end{array}
$$
Then we have two possibilities depending on $e''$:
    \begin{enumerate}
    \item[i)] $e'' \equiv Z \in \var$: Then we can do:
$$
\begin{array}{ll}
let~\overline{X_2=a_2}~in~let~\overline{X_1=f_1(\overline{t'_1})}~in~let~Y=Z~in~c(Y, t'_2) \\
\f let~\overline{X_2=a_2}~in~let~\overline{X_1=f_1(\overline{t'_1})}~in~c(Z, t'_2) & \mbox{ by \lrrule{Bind}}\\
\end{array}
$$
Then we are done as $|f_1(\overline{t'_1})| = \perp$, $|a_{2_i}| = \perp$ for every $a_{2_i} \in \overline{a_2}$ by IH, and $|let~\overline{X_2=a_2}~in~let~\overline{X_1=f_1(\overline{t'_1})}~in~c(Z, t'_2)| = c(Z, t'_2)[\overline{X_1/\perp}][\overline{X_2/\perp}] \sqsupseteq c(\perp, t_2)$, as $t'_2[\overline{X_2/\perp}] \sqsupseteq t_2$ by IH, and no variable in $\overline{X_1}$ appears in $t'_2$ by $\alpha$-conversion, like in the case {\it c)}.
    \item[ii)] $e'' \equiv h_1(\overline{t_1})$: there are two possible cases:
        \begin{enumerate}
    \item[A)] $h_1 = f_1 \in FS$: We are done as $|f_1(\overline{t'_1})| = \perp$, $|a_{2_i}| = \perp$ for every $a_{2_i} \in \overline{a_2}$ by IH, $|f_1(\overline{t_1})| = \perp$, and $|let~\overline{X_2=a_2}~in~let~\overline{X_1=f_1(\overline{t'_1})}~in~let~Y=f_1(\overline{t_1})~in~c(Y, t'_2)| = c(Y, t'_2)[Y/\perp][\overline{X_1}/\overline{\perp}][\overline{X_2}/\overline{\perp}] \sqsupseteq c(\perp, t_2)$, as by IH $t'_2[\overline{X_2}/\overline{\perp}] \sqsupseteq t_2$, $Y$ is fresh and so it does not appear in $t'_2$, and no variable in $\overline{X_1}$ appears in $t'_2$ as in the case {\it i)}.
        \item[B)] $h_1 = c_1 \in DC$: Then we can do a \lrrule{Bind} step:
$$
\begin{array}{ll}
let~\overline{X_2=a_2}~in~let~\overline{X_1=f_1(\overline{t'_1})}~in~let~Y=c_1(\overline{t_1})~in~c(Y, t'_2) \\
\f let~\overline{X_2=a_2}~in~let~\overline{X_1=a_1}~in~c(c_1(\overline{t_1}), t'_2) &\\
\end{array}
$$
Then we are done as $|f_1(\overline{t'_1})| = \perp$, $|a_{2_i}| = \perp$ for every $a_{2_i} \in \overline{a_2}$ by IH, and 
$$\begin{array}{ll}
  & |let~\overline{X_2=a_2}~in~let~\overline{X_1=f_1(\overline{t'_1})}~in~c(c_1(\overline{t_1}), t'_2)| \\
  = & c(c_1(\overline{t_1}), t'_2)[\overline{X_1/\perp}][\overline{X_2/\perp}] \\
  \sqsupseteq & c(\perp, t_2)\\
  \end{array}$$
as $t'_2[\overline{X_2/\perp}] \sqsupseteq t_2$ by IH, and no variable in $\overline{X_1}$ appears in $t'_2$, as we saw in {\it i)}.
    \end{enumerate}
    \end{enumerate}

\end{enumerate}

\item[\clrule{OR}] If $f$ has no arguments ($n=0$) then we have:
  $$
  \infer[\clrule{OR}]{f \crwlto t}{r\theta \crwlto t}
  $$
  with $(f \crwlto r) \in \prog$ and $\theta \in CSubst_\bot$. Let us define $\theta' \in
  CSubst$ as the substitution which is equal to $\theta$ except that every
  $\perp$ introduced by $\theta$ is replaced with some constructor symbol or
  variable. Then $\theta \sqsubseteq \theta'$, so by Proposition \ref{PropMonSubstCrwlLet} we have
  $r\theta' \crwlto t$ with a proof of the same size. But then applying the
  IH to this proof we get $r\theta' \fe let~\overline{X=a}~in~t'$ under the
  conditions of the lemma. Hence $f \f r\theta' \fe
  let~\overline{X=a}~in~t'$ applying (Fapp) in the first step, and we are
  done. \\

  If $n>0$, we will proceed as in the case for \clrule{DC}, doing a preliminary version for $f(e_1, e_2) \crwlto t$ which can be easily extended for the general case. Then we have:
  $$
  \infer[\clrule{OR}]{f(e_1, e_2) \crwlto t}
  {e_1 \crwlto \perp
    \ e_2 \crwlto t_2~
    \ ~r\theta \crwlto t
  }
  $$
  such that $t_2 \not\equiv \perp$, and with $(f(p_1, p_2) \to r) \in \prog$, $\theta \in CSubst_\perp$, such that $p_1\theta = \perp$ and $p_2\theta = t_2$. 
  Then applying the IH to $e_2 \crwlto t_2$ we get that $e_2
  \fe let~\overline{X_2=a_2}~in~t'_2$ such that $|a_{2_i}|=\perp$ for every
  $a_{2_i}$ and $|let~\overline{X_2=a_2}~in~t'_2|=t'_2[\overline{X_2/\perp}]
  \sqsupseteq t_2$. Then we can do:
  $$
  \begin{array}{ll}
    f(e_1, e_2) \fe f(e_1, let~\overline{X_2=a_2}~in~t'_2) & \mbox{ by IH} \\
    \f let~Y=(let~\overline{X_2=a_2}~in~t'_2)~in~f(e_1, Y) & \mbox{ by \lrrule{LetIn}} \\
    \fe let~\overline{X_2=a_2}~in~let~Y=t'_2~in~f(e_1, Y) & \mbox{ by \lrrule{Flat*}} \\
    \f let~\overline{X_2=a_2}~in~f(e_1, t'_2) & \mbox{ by \lrrule{Bind}} \\
  \end{array}
  $$
  Then applying Lemma \ref{T36} we get
  $$f(e_1, t'_2) \fe
  let~\overline{X_1=f_1(\overline{t'})}~in~f(t'_1, t'_2)$$
  Now as $t'_2[\overline{X_2/\perp}] \sqsupseteq t_2$ then $(t'_1,
  t'_2) \sqsupseteq (\perp, t_2)$, so by Lemma \ref{T35} there must exist
  $\theta' \in CSubst$ such that $\theta \sqsubseteq \theta'$ and $(p_1,
  p_2)\theta' \equiv (t'_1, t'_2)$. Then by
  Proposition \ref{PropMonSubstCrwlLet}, as $r\theta \crwlto t$ then $r\theta'
  \crwlto t$ with a proof of the same size. As $\theta' \in CSubst$ and $e \in
  LExp$ (because it is part of the program) then $r\theta' \in LExp$ and we can
  apply the IH to that proof getting that $r\theta' \fe
  let~\overline{X=a}~in~t'$ such that $|a_{i}|=\perp$ for every $a_{i}$ and
  $|let~\overline{X=a}~in~t'|=t'[\overline{X/\perp}] \sqsupseteq t$. Then we can do:
  $$
  \begin{array}{ll}
    let~\overline{X_2=a_2}~in~f(e_1, t'_2)\\
    \fe let~\overline{X_2=a_2}~in~let~\overline{X_1=f_1(\overline{t'})}~in~f(t'_1, t'_2) & \mbox{ by Lemma \ref{T36}} \\
    \equiv let~\overline{X_2=a_2}~in~let~\overline{X_1=f_1(\overline{t'})}~in~f(p_1, p_2)\theta' \\
    \f let~\overline{X_2=a_2}~in~let~\overline{X_1=f_1(\overline{t'})}~in~r\theta' & \mbox{ by \lrrule{Fapp}} \\
    \fe let~\overline{X_2=a_2}~in~let~\overline{X_1=f_1(\overline{t'})}~in~let~\overline{X=a}~in~t' & \mbox{ by $2^{nd}$ IH} \\
  \end{array}
  $$
  Then $|a_{2_i}|=\perp$ for every $a_{2_i} \in \overline{a_2}$ by IH, $|f_1(\overline{t'})|=\perp$ and $|a_{i}|=\perp$ for every $a_{i}$ by IH.
Besides the variables in $\overline{X_1} \cup \overline{X_2}$ either belong to $BV(e_1) \cup BV(e_2)$ or are fresh, hence none of them may appear in $t$ (by Lemma \ref{LAlfCrwlP} over $f(e_1, e_2) \crwlto t$ or by freshness).
%
So
  $t'[\overline{X/\perp}] \sqsupseteq t$ implies that $\forall p \in O(t')$ such
  that $t'|_{p} = Y$ for some $Y \in \overline{X_1}\cup \overline{X_2}$ then
  $t|_{p} = \perp$. But then
  $|let~\overline{X_2=a_2}~in~let~\overline{X_1=a_1}~in~let~\overline{X=a}~in~t'|
  \equiv t'[\overline{X/\perp}][\overline{X_1/\perp}][\overline{X_2/\perp}] \sqsupseteq
  t$.

\item[\clrule{Let}] Then $e \equiv let~X=e_1~in~e_2$ and we have a proof of the following shape:
$$
\infer[\clrule{Let}]{let~X=e_1~in~e_2 \clto t}
            { e_1 \clto t_1~
            \ ~e_2[X/t_1] \clto t
            }
$$
Then we have two possibilities:
\begin{enumerate}
 \item[a)] $t_1 \equiv \perp$: Then $e_2[X/t_1] \equiv e_2[X/\perp] \sqsubseteq e_2$. Hence, as $e_2[X/t_1] \clto t$ and $[X/t_1] \ordap \epsilon$, by Proposition \ref{PropMonSubstCrwlLet} we get $e_2\epsilon \equiv e_2 \clto t$ with a proof of the same size or smaller, and so by IH we get $e_2 \fe let~\overline{X=a}~in~t'$, with $t' \in CTerm$, $|a_i| \equiv \perp$ for every $a_i$ and $|let~\overline{X=a}~in~t'| \equiv t'[\overline{X/\perp}] \sqsupseteq t$, and we can do:
$$
let~X=e_1~in~e_2 \fe let~X=e_1~in~let~\overline{X=a}~in~t'
$$
%
Besides $X \not\in var(t)$ by Lemma \ref{LAlfCrwlP} over $let~X=e_1~in~e_2 \clto t$, and 
then $t'[\overline{X/\perp}] \sqsupseteq t$ implies $\forall p \in O(t')$ such that $t'|_{p} \equiv X$ then $t|_{p} \equiv \perp$, and we have several possible cases:
\begin{enumerate}
  \item[i)] $e_1 = f_1(\overline{e_1})$: Then we are donde because $|\overline{a}| \equiv \overline{\perp}$ by IH, $|f_1(\overline{e_1})| \equiv \perp$ and $|let~X=f_1(\overline{e_1})~in~let~\overline{X=a}~in~t'| \equiv t'[\overline{X/\perp}][X/\perp] \sqsupseteq t$, as $t'[\overline{X/\perp}] \sqsupseteq t$ and $\forall p \in O(t')$ such that $t'|_{p} \equiv X$ then $t|_{p} \equiv \perp$, as we saw above.

  \item[ii)] $e_1 = t'_1 \in CTerm$: But then
$$
\begin{array}{ll@{}}
let~X=t'_1~in~let~\overline{X=a}~in~t' \f let~\overline{X=a[X/t'_1]}~in~t'[X/t'_1] & \mbox{by \lrrule{Bind}}
\end{array}
$$
and we are done because $|\overline{a}| \equiv \overline{\perp}$ by IH, and so $|\overline{a}[X/t'_1]| \equiv \overline{\perp}$ by Lemma \ref{lemmaShellSubst2}. Besides, as in {\it i)}, $t'[\overline{X/\perp}] \sqsupseteq t$ combined with the fact that $\forall p \in O(t')$ such that $t'|_{p} \equiv X$ we have $t|_{p} \equiv \perp$, implies that $|let~\overline{X=a[X/t'_1]}~in~t'[X/t'_1]| \equiv t'[X/t'_1][\overline{X/\perp}] \sqsupseteq t$.

  \item[iii)] $e_1 = c_1(\overline{e_1}) \not\in CTerm$ with $c_1 \in CS$: Then by Lemma \ref{T36} we have $c_1(\overline{e_1}) \fe let~\overline{X_1=f_1({\overline{t_1}})}~in$ $c_1(\overline{t_1})$, hence
$$
\begin{array}{@{}ll@{}}
let~X=c_1(\overline{e_1})~in~let~\overline{X=a}~in~t' \\
\fe let~X=(let~\overline{X_1=f_1({\overline{t_1}})}~in~c_1(\overline{t_1}))~in~let~\overline{X=a}~in~t' & \mbox{by Lemma \ref{T36}} \\
\fe let~\overline{X_1=f_1({\overline{t_1}})}~in~let~X=c_1(\overline{t_1})~in~let~\overline{X=a}~in~t' & \mbox{by \lrrule{Flat$^*$}} \\
\f let~\overline{X_1=f_1({\overline{t_1}})}~in~let~\overline{X=a[X/c_1(\overline{t_1})]}~in~t'[X/c_1(\overline{t_1})] & \mbox{by \lrrule{Bind}} \\
\end{array}
$$
As by IH $|\overline{a}| \equiv \overline{\perp}$ then $|\overline{a[X/c_1(\overline{t_1})]}| \equiv \overline{\perp}$ by Lemma \ref{lemmaShellSubst2}. At this point we have to check that $|let~\overline{X_1=a_1}~in~let~\overline{X=a[X/c_1(\overline{t_1})]}$ $in$ $t'[X/c_1(\overline{t_1})]|$ $\equiv t'[X/c_1(\overline{t_1})][\overline{X/\perp}][\overline{X_1/\perp}] \sqsupseteq t$.
The variables in $\overline{X_1}$ either belong to $BV(c_1(\overline{e_1}))$ or are fresh, hence by $\alpha$-conversion none of them may appear in $t'$, because in $let~X=c_1(\overline{e_1})~in~let~\overline{X=a}~in~t'$ the expression $t'$ has no access to the variables bound in  $c_1(\overline{e_1})$. Hence $t'[X/c_1(\overline{t_1})][\overline{X/\perp}][\overline{X_1/\perp}] \equiv t'[X/t''][\overline{X/\perp}]$, for some $t'' \in CTerm_{\perp}$.
But then, as in {\it ii)}, $t'[\overline{X/\perp}] \sqsupseteq t$ combined with the fact that $\forall p \in O(t')$ such that $t'|_{p} \equiv X$ we have $t|_{p} \equiv \perp$, implies that $t'[X/t''][\overline{X/\perp}]\sqsupseteq t$.
%
  \item[iv)] $e_1 \equiv let~Y=e_{11}~in~e_{12}$: Then by Lemma \ref{T36} we have $let~Y=e_{11}~in~e_{12} \fe let~\overline{X_1=f_1({\overline{t_1}})}~in~h_1(\overline{t_1})$, and so
$$
\begin{array}{@{}l@{}l@{}}
let~X=(let~Y=e_{11}~in~e_{12})~in~let~\overline{X=a}~in~t' \\
\fe let~X=(let~\overline{X_1=f_1({\overline{t_1}})}~in~h_1(\overline{t_1}))~in~let~\overline{X=a}~in~t' & \mbox{ by Lemma \ref{T36}} \\
\fe let~\overline{X_1=f_1({\overline{t_1}})}~in~let~X=h_1(\overline{t_1})~in~let~\overline{X=a}~in~t' & \mbox{ by \lrrule{Flat$^*$}} \\
\end{array}
$$
Then either $h \in CS$ and we are like in {\it iii)} before the final \lrrule{Bind} step, or $h \in FS$ and $|h_1(\overline{t_1})| = \perp$ and $|\overline{a}| = \overline{\perp}$ (by IH), and $|let~\overline{X_1=a_1}~in~let~X=h_1(\overline{t_1})~in~let~\overline{X=a}~in~t'| \equiv t'[\overline{X/\perp}][X/\perp][\overline{X_1/\perp}] \equiv t'[\overline{X/\perp}][X/\perp]$ because $\overline{X_1} \cap var(t') = \emptyset$, as we saw in {\it iii)}. But then, as in {\it ii)}, $t'[\overline{X/\perp}] \sqsupseteq t$ combined with the fact that $\forall p \in O(t')$ such that $t'|_{p} \equiv X$ we have $t|_{p} \equiv \perp$, implies that $t'[\overline{X/\perp}][X/\perp] \sqsupseteq t$.
%
\end{enumerate}
 \item[b)] $t_1 \not\equiv \perp$: Then by IH we get $e_1 \fe let~\overline{X_1=a_1}~in~t'_1$, with $t'_1 \in CTerm$, $|a_{1_i}| \equiv \perp$ for every $a_{1_i}$ and $|let~\overline{X_1=a_1}~in~t'_1| \equiv t'_1[\overline{X_1/\perp}] \sqsupseteq t_1$. Hence $t_1 \sqsubseteq t'_1$ and so $e_2[X/t_1] \sqsubseteq e_2[X/t'_1]$, but then $e_2[X/t_1] \clto t$ implies $e_2[X/t'_1] \clto t$ with a proof of the same size or smaller, by Proposition \ref{propCrwlletPolar}. Therefore we may apply the IH to that proof to get $e_2[X/t'_1] \fe let~\overline{X=a}~in~t'$, with $t' \in CTerm$, $|a_i| \equiv \perp$ for every $a_i$ and $|let~\overline{X=a}~in~t'| \equiv t'[\overline{X/\perp}] \sqsupseteq t$. But then we can do:
$$
\begin{array}{ll}
let~X=e_1~in~e_2 \fe let~X=(let~\overline{X_1=a_1}~in~t'_1)~in~e_2 & \mbox{ by IH} \\
\fe let~\overline{X_1=a_1}~in~let~X=t'_1~in~e_2 & \mbox{ by \lrrule{Flat$^*$} } \\
\f let~\overline{X_1=a_1}~in~e_2[X/t'_1] & \mbox{ by \lrrule{Bind} } \\
\fe let~\overline{X_1=a_1}~in~let~\overline{X=a}~in~t' & \mbox{ by IH} \\
\end{array}
$$
Then by the IH's we have $|\overline{a}| = \overline{\perp}$ and $|\overline{a_1}| = \overline{\perp}$. Besides the variables in $\overline{X_1}$ either belong to $BV(e_1)$ or are fresh, hence none of them may appear in $t$ (by Lemma \ref{LAlfCrwlP} over $let~X=e_1~in~e_2 \clto t$ or by freshness). So  $t'[\overline{X/\perp}] \sqsupseteq t$ implies that $\forall p \in O(t')$ such that $t'|_{p} = Y$ for some $Y \in \overline{X_1}$ then $t|_{p} = \perp$. But then $|let~\overline{X_1=a_1}~in~let~\overline{X=a}~in~t'| \equiv t'[\overline{X/\perp}][\overline{X_1/\perp}] \sqsupseteq t$.
        \end{enumerate}
\end{description}
\end{proof}


\subsection{Proofs for Section \ref{SemEqs}}
\label{DEMO_lDistHD}

\teoremi{Lemma \ref{lDistHD}}
If $BV(\con) \cap FV(e_1) = \emptyset$ and $X \not\in FV(\con)$ then
$\denn{\con[let~X = e_1~in~e_2]} = \denn{let~X = e_1~in~\con[e_2]}$


\begin{proof}
One step of the rule (Dist) can be replaced by two steps (CLetIn) + (Bind):
\begin{center}
$\con[let~X = e_1~in~e_2] \f let~U = e_1~in~\con[let~X = U~in~e_2] \f let~U = e_1~in~\con[e_2[X/U]]$
\end{center}
followed by a renaming of $U$ by $X$ in the last expression.
Then the lemma follows from preservation of hypersemantics by (CLetIn) and (Bind) (Lemma \ref{lemCLetInPreserv} and Proposition \ref{propFnfPreservHipSem}).
\end{proof}

\teoremi{Proposition \ref{propHypSemPropAlter} ((Hyper)semantic properties of $?$)}
For any $e_1, e_2 \in LExp_\perp$
\begin{enumerate}
    \item[i)]
$
\den{e_1~?~e_2} = \den{e_1} \cup \den{e_2}
$
    \item[ii)]
$
\denn{e_1~?~e_2} = \denn{e_1} \uhs \denn{e_2}
$
\end{enumerate}
\begin{proof}
\begin{enumerate}
    \item[i)] Direct from definition of ? and the \crwl-proof calculus.
        \item[ii)]
$$
\begin{array}{ll}
\denn{e_1~?~e_2} = \lambda\theta.\den{(e_1~?~e_2)\theta} & \mbox{ by definition of $\denn{~}$}\\
= \lambda\theta.\den{e_1\theta~?~e_2\theta} \\
= \lambda\theta.(\den{e_1\theta} ~\cup~ \den{e_2\theta}) & \mbox{ by {\it i)}} \\
= \lambda\theta.(\denn{e_1}\theta ~\cup~ \denn{e_2}\theta) & \mbox{ by definition of $\denn{~}$}\\
= \denn{e_1} \uhs \denn{e_2} & \mbox{ by definition of $\uhs$}\\
\end{array}
$$
\end{enumerate}
\end{proof}

\subsection{Proofs for Section \ref{let-narrowing}}

\teoremi{Theorem \ref{SoundLNarr} (Soundness of the let-narrowing relation $\fnrl$)}
For any $e, e' \in LExp$, $e \fnre_{\theta} e'$ implies $e\theta \f^* e'$.

\begin{proof}\label{DEMO_SoundLNarr}
First we prove the soundness of narrowing for one step, proceeding by a case distinction over the rule used in $e \fnr_{\theta} ~e'$. The cases of (Elim), (Bind), (Flat) and (LetIn) are trivial, since narrowing and rewriting coincide for these rules.
\begin{description}
 \item[(Narr)] Then we have $f(\overline{t}) \fnr_{\theta} ~r\theta$ for $(f(\overline{p}) \tor r) \in {\cal P}$ fresh, $\theta \in CSubst$ such that $f(\overline{t})\theta \equiv f(\overline{p})\theta$.
But then $(f(\overline{p}) \tor r)\theta \equiv f(\overline{p})\theta \tor r\theta \equiv f(\overline{t})\theta \tor r\theta$, so we can do $e\theta \equiv f(\overline{t})\theta \f r\theta \equiv e'$ by (Fapp).\\

 \item[(Contxt)] Then we have ${\cal C}[e] \fnr_\theta  ~{\cal C}\theta[e']$ because $e \fnr_{\theta} e'$. Let us do a case distinction over the rule applied in $e \fnr_{\theta} ~e'$:\\
\begin{enumerate}
 \item[a)] $e \fnr_{\theta} ~e' \equiv f(\overline{t}) \fnr_{\theta} ~r \theta$ by
   (Narr), for $(f(\overline{p}) \tor r) \in {\cal P}$ fresh, so $f(\overline{t})\theta \f r\theta$ by \clrule{Fapp}. Then $({\cal
     C}[e])\theta \equiv ({\cal C}[e])\theta|_{\setminus var(\overline{p})}$,
   because the variables in $var(\overline{p})$ are fresh as $(f(\overline{p})
   \tor r)$ is. But then, as $dom(\theta) \cap BV({\cal C}) = \emptyset$ and
   $vRan(\theta|_{\setminus var(\overline{p})}) \cap BV({\cal C}) = \emptyset$
   by the conditions in (Contx), and $dom(\theta) \cap BV({\cal C}) = \emptyset$
   implies $dom(\theta|_{\setminus var(\overline{p})}) \cap BV({\cal C}) =
   \emptyset$, we can apply Lemma \ref{T28} getting $({\cal
     C}[e])\theta|_{\setminus var(\overline{p})} \equiv$\\ ${\cal
     C}\theta|_{\setminus var(\overline{p})}[e\theta|_{\setminus
     var(\overline{p})}] \equiv {\cal C}\theta|_{\setminus
     var(\overline{p})}[f(\overline{t})\theta|_{\setminus var(\overline{p})}]
   \equiv {\cal C}\theta[f(\overline{t})\theta]$, because the variables in
   $var(\overline{p})$ are fresh. Besides $\vran(\theta|_{\setminus var(\overline{p})}) \cap BV(\con) = \emptyset$, so we can apply \clrule{Contx} combined with an inner \clrule{Fapp} to do $(\con[e])\theta \equiv \con\theta[f(\overline{t})\theta] \f \con\theta[r\theta] \equiv \con\theta[e']$. \\

 \item[b)] In case a different rule was applied in $e \fnr_{\theta} ~e'$ then $\theta = \epsilon$.
 By the proof of the other cases we have $e\theta \equiv e \f e'$, so $({\cal C}[e])\theta \equiv {\cal C}[e] \f {\cal C}[e'] \equiv {\cal C}\theta[e']$ (remember $\theta = \epsilon$).
\end{enumerate}
 \end{description}

Now we prove the lemma for any number of steps $\f$, proceeding by induction over the length $n$ of $e \fnrc{n}_{\theta} e'$. The case $e \fnrc{0}_{\epsilon} e \equiv e'$ is straightforward because $e \fc{0} e \equiv e'$. For $n > 0$ we have the derivation $e \fnrc{}_\sigma e'' \fnrc{n-1}_{\gamma} e'$ with $\theta = \gamma \circ \sigma$. By the proof for one step $e\sigma \f e''$, and by the closeness under $CSubst$ of let-rewriting (Lemma \ref {LRwCerr}) $e\sigma\gamma \f e''\gamma$. By IH $e''\gamma \fe e'$, so we can link $e\theta \equiv e\sigma\gamma \f e''\gamma \fe e'$.
\end{proof}

\teoremi{Lemma \ref{lem:lifting} (Lifting lemma for the let-rewriting relation $\f$)}
Let $e,e' \in LExp$ such that $e\theta \fe e'$ for some $\theta \in CSubst$, and let
${\cal W}, {\cal B} \subseteq {\cal V}$ with $dom(\theta) \cup FV(e) \subseteq {\cal W}$, $BV(e) \subseteq {\cal B}$
and $(dom(\theta) \cup \vran(\theta)) \cap {\cal B} = \emptyset$, and for each \crule{Fapp} step of $e\theta \fe e'$  using a rule $R \in \prog$ and a substitution $\gamma \in CSubst$ then $\vran(\gamma|_{vExtra(R)}) \cap {\cal B} = \emptyset$.
Then there  exist a derivation $e ~\fnrl^*_{\sigma}~ e''$ and $\theta' \in CSubst$ such that:
$$ \mbox{(i)~} e''\theta' = e'
\qquad \mbox{(ii)~} \sigma\theta' = \theta[{\cal W}]
\qquad \mbox{(iii)~} (dom(\theta') \cup \vran(\theta')) \cap {\cal B} = \emptyset
$$
Besides, the let-narrowing derivation can be chosen to use mgu's at each \crule{Narr} step.

\begin{proof}\label{DEMO_lem:lifting}
Let us do a case distinction over the rule applied in $e\theta \f e'$:
\begin{description}

\item[(Fapp)] $e \equiv f(\overline{t})$, so:

\begin{center}
\begin{tikzpicture}[scale=0.9, auto] 
\tikzstyle{nthis}=[npath, minimum size=11mm]
    \node[nthis] (e) at (-2,0) {$f(\overline{t})$};
        \node at (-3, 0) {$e \equiv$};
    \node[nthis] (epp) at (2,0) {$r\sigma$};
        \node at (3.05, 0) {$\equiv e''$};
    \node[nthis] (etheta) at (-2,-2) {$f(\overline{t})\theta$};
    \node[nthis] (ep) at (2,-2) {$r\gamma$};
        \node at (3.05, -2) {$\equiv e'$};
    \draw [->, apath, densely dotted, snake=snake, line after snake=2mm] (e) to node {$l$} node[swap] {$\sigma$} (epp);
        \draw [->, apath] (etheta) to node[swap] {$l$} (ep);
    \draw [|->, apath, shorten <=1pt] (e) to node [swap] {$\theta$} (etheta);
    \draw [|->, apath, shorten <=1pt, dashed] (epp) to node {$\theta'$} (ep);

\begin{pgfonlayer}{background}
   \filldraw [fondoTerm]
      (e.north -| e.west) rectangle (ep.south -| ep.east);
\end{pgfonlayer}

\end{tikzpicture}
\end{center}

With an (Fapp) step $e\theta \equiv f(\overline{t})\theta \f r\gamma$ with $(f(\overline{p}) \to r) \in \prog$, $\gamma \in CSubst$, such that $f(\overline{t})\theta \equiv f(\overline{p})\gamma$ and $f(\overline{p}) \to r$ is a fresh variant.
We can assume that $dom(\gamma) \subseteq FV(f(\overline{p}) \tor r)$ without loss of generality. But then $dom(\theta) \cap dom(\gamma)= \emptyset$, and so $\theta \uplus \gamma$ is correctly defined, and it is a unifier of $f(\overline{t})$ and $f(\overline{p})$. So, there must exist $\sigma = mgu(f(\overline{t}), f(\overline{p}))$, which we can use to perform a (Narr) step, because $\sigma \in CSubst$ and $f(\overline{t})\sigma \equiv f(\overline{p})\sigma$.
$$e \equiv f(\overline{t}) \fnrl_\sigma~ r\sigma \equiv e''$$
As this unifier is an mgu then $dom(\sigma) \subseteq FV(f(\overline{t})) \cup FV(f(\overline{p}))$, $\vran(\sigma) \subseteq FV(f(\overline{t})) \cup FV(f(\overline{p}))$ and $\sigma \ordSus (\theta \uplus \gamma)$, so there must exist $\theta'_1 \in CSubst$ such that $\sigma\theta'_1 = \theta \uplus \gamma$. Besides we can define $\theta'_0 = \theta|_{\setminus ( dom(\theta'_1) \cup FV(f(\overline{t})) )}$ and then we can take $\theta' = \theta'_0 \uplus \theta'_1$ which is correctly defined as obviously $dom(\theta'_0) \cap dom(\theta'_1) = \emptyset$. Besides $dom(\theta'_0) \cap (FV(f(\overline{t})) \cup FV(f(\overline{p})) = \emptyset$, as if $Y \in FV(f(\overline{t}))$ then $Y \not\in dom(\theta'_0)$ by definition; and if $Y \in FV(f(\overline{p}))$ then $Y \not\in dom(\theta)$ as $\overline{p}$ belong to the fresh variant, and so $Y \notin dom(\theta'_0)$. Then the conditions in Lemma \ref{lem:lifting} hold:

\begin{itemize}
 \item \underline{Condition i)} $e''\theta' \equiv e'$: As $e''\theta' \equiv r\sigma\theta' \equiv r\sigma\theta'_1$ because given $Y \in FV(r\sigma)$, if $Y \in FV(r)$ then it belongs to the fresh variant and so $Y \notin dom(\theta) \supseteq dom(\theta'_0)$; and if $Y \in \vran(\sigma)\subseteq FV(f(\overline{t})) \cup FV(f(\overline{p}))$ then $Y \notin dom(\theta'_0)$ because $dom(\theta'_0) \cap (FV(f(\overline{t})) \cup FV(f(\overline{p}))) = \emptyset$. 
But $r\sigma\theta'_1 \equiv r(\theta \uplus \gamma) \equiv r\gamma \equiv e'$, because $\sigma\theta'_1 = \theta \uplus \gamma$ and $r$ is part of the fresh variant.

\item \underline{Condition ii)}  $\sigma\theta' = \theta[{\cal W}]$: Given $Y \in {\cal W}$, if $Y \in FV(f(\overline{t}))$ then $Y \notin dom(\gamma)$ and so $Y\theta \equiv Y(\theta \uplus \gamma) \equiv Y\sigma\theta'_1$, as $\sigma\theta'_1 = \theta \uplus \gamma$. But $Y\sigma\theta'_1 \equiv Y\sigma\theta'$ because given $Z \in var(Y\sigma)$, if $Z \equiv Y$ then as $Y \in FV(f(\overline{t}))$ then $Z \equiv Y \notin dom(\theta'_0)$ by definition of $\theta'_0$; if $Z \in \vran(\sigma)$ then 
$Z \not\in dom(\theta'_0)$, as we saw before.\\
On the other hand, $({\cal W} \setminus FV(f(\overline{t}))) \cap (FV(f(\overline{t})) \cup FV(f(\overline{p}))) = ({\cal W} \setminus FV(f(\overline{t})) \cap FV(f(\overline{t}))) \cup ({\cal W} \setminus FV(f(\overline{t})) \cap FV(f(\overline{p}))) = \emptyset \cup \emptyset = \emptyset$, because $FV(f(\overline{p}))$ are part of the fresh variant. So, if $Y \in {\cal W} \setminus FV(f(\overline{t}))$, then $Y \not\in dom(\sigma) \subseteq FV(f(\overline{t})) \cup FV(f(\overline{p}))$. Now if $Y \in dom(\theta'_0)$ then $Y\theta \equiv Y\theta'_0$ (by definition of $\theta'_0$), $Y\theta'_0 \equiv Y\theta'$ (as $Y \in dom(\theta'_0)$), $Y\theta' \equiv Y\sigma\theta'$ (as $Y \not\in dom(\sigma)$). If $Y \in dom(\theta'_1)$, $Y\theta \equiv Y(\theta \uplus \gamma)$ (as $Y \in {\cal W} \setminus FV(f(\overline{t}))$ implies it does not appear in the fresh instance), $Y(\theta \uplus \gamma) \equiv Y\sigma\theta'_1$ (as $\sigma\theta'_1 = \theta \uplus \gamma$), $Y\sigma\theta'_1 \equiv Y\theta_1'$ (as $Y \not\in dom(\sigma)$), $Y\theta'_1 \equiv Y\theta'$ (as  $Y \in dom(\theta'_1)$) and $Y\theta' \equiv Y\sigma\theta'$ (as $Y \not\in dom(\sigma)$). And if $Y \not\in (dom(\theta'_0) \cup dom(\theta'_1))$ then $Y \not\in dom(\theta')$, and as $Y \not\in dom(\sigma)$ and $Y\theta \equiv Y(\theta \uplus \gamma)$, then $Y\theta \equiv Y(\theta \uplus \gamma) \equiv Y\sigma\theta'_1 \equiv Y \equiv Y\sigma\theta'$.

\item \underline{Condition iii.1)} $dom(\theta') \cap {\cal B} = \emptyset$. Remember $\theta' = \theta'_0 \uplus \theta'_1$:
\begin{itemize}
 \item $dom(\theta'_0) \cap {\cal B} = \emptyset$: Given $Y \in dom(\theta'_0)$ then $Y \in dom(\theta)$ by definition of $\theta'_0$, and so $Y \not\in {\cal B}$, because $dom(\theta) \cap {\cal B} = \emptyset$ by hypothesis.
 \item $dom(\theta'_1) \cap {\cal B} = \emptyset$: As $\sigma$ is an mgu and $\sigma \ordSus \theta \uplus \gamma$, then $dom(\sigma) \subseteq dom(\theta \uplus \gamma)$. Given $Z \in {\cal B}$ then $Z \not\in dom(\theta)$, as $dom(\theta) \cap {\cal B} = \emptyset$ by hypothesis, and $Z \notin dom(\gamma) \subseteq FV(f(\overline{p}) \tor r)$ which are fresh, so $Z \not\in dom(\sigma)$. But then, as $\sigma\theta'_1 = \theta \uplus \gamma$, $Z \equiv Z(\theta \uplus \gamma) \equiv Z\sigma\theta'_1 \equiv Z\theta'_1$, so $Z \not\in dom(\theta'_1)$.
\end{itemize}

\item \underline{Condition iii.2)} $\vran(\theta') \cap {\cal B} = \emptyset$. Remember $\theta' = \theta'_0 \uplus \theta'_1$:
\begin{itemize}
 \item $\vran(\theta'_0) \cap {\cal B} = \emptyset$: Given $Y \in dom(\theta'_0)$ then $Y\theta'_0 \equiv Y\theta$ by definition of $\theta'_0$. As $\vran(\theta) \cap {\cal B} = \emptyset$ by hypothesis then it must happen $var(Y\theta) \cap {\cal B} = \emptyset$, so $var(Y\theta'_0) \cap {\cal B} = \emptyset$.
 \item $\vran(\theta'_1) \cap {\cal B} = \emptyset$: As $\sigma\theta'_1 = \theta \uplus \gamma$ then we can assume $dom(\theta'_1) \subseteq \vran(\sigma) \cup (dom(\theta \uplus \gamma) \setminus dom(\sigma))$.

\begin{itemize}
 \item Let $X \in dom(\theta'_1) \cap \vran(\sigma)$ be such that $X\theta'_1 \equiv r[Z]$ with $Z \in {\cal B}$. We will see that this $Z \in {\cal B}$ can appear in $X\theta'_1$ without leading to contradiction. The intuition is, as $\vran(\theta) \cap {\cal B} = \emptyset$ and $\vran(\gamma|_{vExtra(R)}) \cap {\cal B} = \emptyset$, then every $Z \in {\cal B}$ must come from an appearance in $e$ of the same variable, transmitted to $e'$ by the matching substitution $\gamma$, and so transmitted to $e''$ by $\sigma$.\\

As $X \in \vran(\sigma)$ then there must exist $Y \in dom(\sigma)$ such that $Y \longmapsto^\sigma r_1[X]_p \longmapsto^{\theta'_1} r_2[s[Z]]_p$. But as $\sigma\theta'_1 = \theta \uplus \gamma$ then $Y \longmapsto^{\theta \uplus \gamma} r_2[s[Z]]_p$. Then, $Z \in \vran(\theta \uplus \gamma)$, but $Z \in {\cal B}, \vran(\theta) \cap {\cal B} = \emptyset, \vran(\gamma|_{vExtra(R)}) \cap {\cal B} = \emptyset,dom(\gamma) \subseteq FV(f(\overline{p}) \tor s)$, so it must happen $Z \in \vran(\gamma|_{FV(\overline{p})})$, and as a consequence $Y \in FV(\overline{p})$. Let $o \in O(f(\overline{p}))$ (set of positions in $f(\overline{p})$) be such that $f(\overline{p})|_o \equiv Y$, then:
\begin{itemize}
 \item $((f(\overline{t}))\sigma)|_o \equiv ((f(\overline{p}))\sigma)|_o \equiv ((f(\overline{p}))|_o)\sigma \equiv Y\sigma \equiv r_1[X]_p$.
\item As $f(\overline{t}) \not\in dom(\gamma)$, which are the fresh variables of the variant of the program rule, $((f(\overline{t}))\theta)|_o \equiv ((f(\overline{t}))(\theta \uplus \gamma))|_o \equiv ((f(\overline{p}))(\theta \uplus \gamma))|_o \equiv ((f(\overline{p}))|_o)(\theta \uplus \gamma) \equiv Y(\theta \uplus \gamma) \equiv r_2[s[Z]]_p$
\end{itemize}
So, as $X \in dom(\theta'_1)$ then $X \not\in {\cal B}$ and $Z \in {\cal B}$ has been introduced by $\theta$, but this is impossible as $\vran(\theta) \cap {\cal B} = \emptyset$.
\item Let $Y \in dom(\theta) \setminus dom(\sigma)$ be. Then $Y\theta \equiv Y(\theta \uplus \gamma)$ (as $Y \in dom(\theta$), $Y(\theta \uplus \gamma) \equiv Y\sigma\theta'_1$ (as $\sigma\theta'_1 = \theta \uplus \gamma$), $Y\sigma\theta'_1 \equiv Y\theta'_1$ (as $Y \not\in dom(\sigma)$. But then no variable in ${\cal B}$ can appear in $Y\theta'_1 \equiv Y\theta$ as $(dom(\theta) \cup \vran(\theta)) \cap {\cal B} = \emptyset$.

\item Let $Y \in dom(\gamma) \setminus dom(\sigma)$ be. Then $Y\gamma \equiv Y(\theta \uplus \gamma) \equiv Y\sigma\theta'_1 \equiv Y\theta'_1$, reasoning like in the previous case. As $dom(\gamma) \subseteq FV(f(\overline{p}) \tor s)$ it can happen:
\begin{itemize}
 \item $Y \not\in FV(f(\overline{p}))$: Then no variable in ${\cal B}$ can appear in $Y\gamma$ because $\vran(\gamma|_{vExtra(R)}) \cap {\cal B} = \emptyset$ by the hypothesis.
 \item $Y \in FV(f(\overline{p}))$: Let $Z \in {\cal B}$ appearing in $Y\gamma$, then $Z$ appears in $f(\overline{t})$, so it must happen $Y \in dom(\sigma)$ because otherwise $\sigma$ could not be a unifier of $f(\overline{t})$ and $f(\overline{p})$. But this is a contradiction so this case is impossible.
\end{itemize}

\end{itemize}

\end{itemize}

\end{itemize}

\item[(LetIn)] In this case $e\theta \equiv h(e_1\theta, \ldots, e\theta, \ldots, e_n\theta)$ and $e \equiv h(e_1, \ldots, e, \ldots, e_n)$. Then the let-rewriting step is $$e\theta \equiv h(e_1\theta, \ldots, e\theta, \ldots, e_n\theta) \f let~X=e\theta ~in~h(e_1\theta, \ldots, X, \ldots, e_n\theta) \equiv e'$$ with $h \in \Sigma$, $e\theta \equiv f(\overline{e'})$ ---$f \in FS$--- or $e\theta \equiv let~Y=e'_1~in~e'_2$, and $X$ is a fresh variable. Notice that $e\theta$ is a let-rooted expression or a $f(\overline{e'})$ iff $e$ is a let-rooted expression or a function application, as $\theta \in CTerm$. Then we can apply a let-narrowing step:
$$ e \equiv h(e_1, \ldots, e, \ldots, e_n) \fnrl_\sigma ~let~X=e~in~h(e_1, \ldots, X, \ldots, e_n) \equiv e''$$
with $\sigma \equiv \epsilon$ and $\theta' \equiv \theta$. Then the conditions in Lemma \ref{lem:lifting} hold:

\begin{itemize}
\item[i)] $e''\theta' \equiv (let~X=e~in~h(e_1, \ldots, X, \ldots, e_n))\theta \equiv$\\
          $let~X=e\theta~in~h(e_1\theta, \ldots, X\theta, \ldots, e_n\theta) \equiv$\\
          $let~X=e\theta~in~h(e_1\theta, \ldots, X, \ldots, e_n\theta) \equiv e'$, since $X$ is fresh an it cannot appear in $dom(\theta')$.

\item[ii)] $\sigma \theta' \equiv \epsilon \theta \equiv \theta = \theta[{\cal W}]$.
\item[iii)] $(dom(\theta') \cup \vran(\theta')) \cap {\cal B} = (dom(\theta) \cup \vran(\theta)) \cap {\cal B} = \emptyset$ by hypothesis.
\end{itemize}

\item[(Bind)] In this case $e\theta \equiv let~X = t\theta~in~e_2\theta$ and $e \equiv let~X = t~in~e_2$. Then the let-rewriting step is $let~X = t\theta~in~e_2\theta \f e_2\theta[X/t\theta]$ with $t\theta \in CTerm$. As $\theta \in CTerm$, if $t\theta \in CTerm$ then $t \in CTerm$, so we can apply a let-narrowing step:
$$ e \equiv let~X = t~in~e_2 \fnrl_\sigma ~e_2[X/t] \equiv e''$$
with $\sigma \equiv \epsilon$ and $\theta' \equiv \theta$. Then the conditions in Lemma \ref{lem:lifting} hold:
\begin{itemize}
\item[i)] $e''\theta' \equiv e_2[X/t]\theta$. By the variable convention we can assume that $X \notin dom(\theta) \cup \vran(\theta)$, so by Lemma \ref{auxBind} $e_2[X/t]\theta \equiv e_2\theta[X/t\theta] \equiv e'$.

\item[ii) and iii)] As before.
\end{itemize}

\item[(Elim)] We have $e\theta \equiv let~X=e_1\theta ~in~ e_2\theta$, so $e \equiv let ~X=e_1 ~in~ e_2$. Then the let-rewriting step is $e\theta \equiv let~X=e_1\theta ~in~ e_2\theta \f e_2\theta$ with $X \notin FV(e_2\theta)$. By the variable convention $(dom(\theta) \cup \vran(\theta)) \cap BV(e) = \emptyset$, so as $X \in BV(e)$ then $X \notin dom(\theta) \cup \vran(\theta)$. Then $X \notin FV(e_2\theta)$ implies $X \notin FV(e_2)$ and we can apply a let-narrowing step:
$$ e \equiv let ~X=e_1 ~in~ e_2 \fnrl_\sigma ~e_2 \equiv e''$$
with $\sigma \equiv \epsilon$ and $\theta' \equiv \theta$. Then the conditions in Lemma \ref{lem:lifting} hold trivially.

\item[(Flat)] In this case $e\theta \equiv let ~X=(let~Y=e_1\theta ~in~e_2\theta) ~in~e_3\theta$ and $e \equiv let ~X=(let~Y=e_1 ~in~e_2) ~in~e_3$. The let-rewriting step is $e\theta \equiv let ~X=(let~Y=e_1\theta ~in~e_2\theta)~in~e_3\theta \f let~Y=e_1\theta ~in~let~X=~e_2\theta ~in~e_3\theta \equiv e'$ with $Y \notin FV(e_3\theta)$. By a similar reasoning as in the (Elim) case we conclude that $Y \notin dom(\theta) \cup \vran(\theta)$, so $Y \notin FV(e_3)$. Then we can apply a let-narrowing step:
$$ e \equiv let ~X=(let~Y=e_1 ~in~e_2) ~in~e_3 \fnrl_\sigma ~let~Y=e_1 ~in~let~X=~e_2 ~in~e_3 \equiv e''$$
with $\sigma \equiv \epsilon$ and $\theta' \equiv \theta$. Then the conditions in Lemma \ref{lem:lifting} hold trivially.

 \item[(Contx)] Then we have $e \equiv {\cal C}[s]$. By the variable convention $(dom(\theta) \cup \vran(\theta)) \cap BV(e) = \emptyset$, so by lemma \ref{T28} $e\theta \equiv ({\cal C}[s])\theta \equiv {\cal C}\theta[s\theta]$, and the step was
$$
e\theta \equiv {\cal C}\theta[s\theta] \f {\cal C}\theta[s'] \equiv e' \mbox{, because } s\theta \f s'
$$
Then we know that the lemma holds for $s\theta \f s'$, by the proof of the other cases, so taking ${\cal W'} = {\cal W} \cup FV(s)$ and ${\cal B'} = {\cal B}$ (as $BV(s) \subseteq BV({\cal C}[s])$) we can do $s \fnrl_{\sigma_2} s''$ for some $\theta'_2$ under the conditions stipulated. Now we can put this step into (Contx) to do:
$$
e \equiv {\cal C}[s] \fnrl_{\sigma_2} {\cal C}\sigma_2[s'']  \equiv e'' \mbox{ taking } \sigma=\sigma_2 \mbox{ and } \theta' = \theta'_2
$$
because if $s \fnrl_{\sigma_2} s''$ was a (Narr) step which lifts a (Fapp) step that uses the fresh variant $(f(\overline{p}) \tor r) \in {\cal P}$ and adjusts with $\gamma \in CSubst$, then:
\begin{itemize}
 \item $dom(\sigma_2) \cap BV({\cal C}) = \emptyset$: As $\sigma_2 = mgu(s, f(\overline{p}))$ then $dom(\sigma_2) \subseteq FV(s) \cup FV(f(\overline{p}))$. As $\sigma_2 \ordSus \theta \uplus \gamma$ and it is an mgu then $dom(\sigma_2) \subseteq dom(\theta \uplus \gamma)$. If $X \in FV(s) \cap dom(\sigma_2)$ then $X \not\in dom(\gamma) \subseteq FV(f(\overline{p}) \tor r)$, so it must happen $X \in dom(\theta)$; but then $X \not\in BV({\cal C})$ because $dom(\theta) \cap BV({\cal C}) = \emptyset$ by the variable convention.\\
Otherwise it could happen $X \in FV(f(\overline{p})) \cap dom(\sigma_2)$, then $X$ appears in the fresh variant and so it cannot appear in ${\cal C}$.
\item $\vran(\sigma_2|_{\setminus var(\overline{p})}) \cap BV({\cal C}) = \emptyset$: As $dom(\sigma_2) \subseteq FV(s) \cup FV(f(\overline{p}))$ then we have $\vran(\sigma_2|_{\setminus var(\overline{p})}) = \vran(\sigma_2|_{FV(s)})$. But as $\sigma_2 = mgu(s, f(\overline{p}))$ then $\vran(\sigma|_{FV(s)}) \subseteq FV(f(\overline{p}))$, which are part of the fresh variant, so every variable in $\vran(\sigma_2|_{\setminus var(\overline{p})})$ is fresh and so cannot appear in ${\cal C}$.
\end{itemize}

Then the conditions in Lemma \ref{lem:lifting} hold:
\begin{itemize}
 \item[ii)] $\sigma\theta' = \theta[{\cal W}]$: Because ${\cal W} \subseteq {\cal W'}$, and $\sigma_2\theta'_2 = \theta[{\cal W'}]$, by the proof of the other cases.

 \item[i)] $e''\theta' \equiv e'$: As 
$BV({\cal C}\sigma_2) = BV({\cal C})$, by the variable convention, $BV({\cal C}) \subseteq BV(e) \subseteq BV({\cal B})$, by the hypothesis, and $(dom(\theta'_2) \cup \vran(\theta'_2)) \cap {\cal B} = \emptyset$, by the proof of the other cases, then $(dom(\theta'_2) \cup \vran(\theta'_2)) \cap BV({\cal C}\sigma_2) = \emptyset$. But then:
$$
e''\theta' \equiv ({\cal C}\sigma_2[s''])\theta'_2 \equiv \underbrace{{\cal C}\sigma_2\theta'_2}_{{\cal C}\theta}[\underbrace{s''\theta'_2}_{s'}] \equiv e'
$$
Because we have $s''\theta'_2 \equiv s'$, by the proof of the other cases, and because $FV({\cal C}) \subseteq FV(e) \subseteq {\cal W}$ and $\sigma_2\theta'_2 = \theta[{\cal W}]$, as we saw in the previous case (remember $\sigma=\sigma_2$ and $\theta' = \theta'_2$).

\item[iii)] $(dom(\theta') \cup \vran(\theta')) \cap {\cal B} = \emptyset$: Because $\theta' = \theta'_2$ and the proof of the other cases.
\end{itemize}
\end{description}

The proof for any number of steps proceeds by induction over the number $n$ of steps of the derivation $e\theta \f^n e'$. The base case where $n=0$ is straightforward, as then we have $e\theta \f^0 e\theta \equiv e'$ so we can do $e \fnrl^0_\epsilon~e \equiv e''$, so $\sigma = \epsilon$ and taking $\theta' = \theta$ the lemma holds. In the inductive step we have $e\theta \f e_1 \fe e'$, and we will try to build the following diagram:

\begin{center}
\begin{tikzpicture}[scale=0.8, auto] 
\tikzstyle{nthis}=[npath, minimum size=9mm]
    \node[nthis] (e) at (-2,0) {$e$};
    \node[nthis] (epp1) at (2,0) {$e''_1$};
    \node[nthis] (epp2) at (6,0) {$e''_2$};
    \node (epp) at (7.1,0) {$\equiv e''$};
    \node[nthis] (etheta) at (-2,-2) {$e\theta$};
    \node[nthis] (e1) at (2,-2) {$e_1$};
    \node[nthis] (ep) at (6,-2) {$e'$};
    \draw [->, apath, densely dotted, snake=snake, line after snake=2mm] (e) to node {$l$} node[swap] {$\sigma_1$} (epp1);
    \draw [->, apath, densely dotted, snake=snake, line after snake=2mm] (epp1) to node {$l^*$} node[swap] {$\sigma_2$} (epp2);
        \draw [->, apath] (etheta) to node [swap] {$l$} (e1);
        \draw [->, apath] (e1) to node {$*$} node [swap] {$l$} (ep);
    \draw [|->, apath, shorten <=1pt] (e) to node [swap] {$\theta$} (etheta);
    \draw [|->, apath, shorten <=1pt, dashed] (epp1) to node [swap] {$\theta_1$} (e1);
    \draw [|->, apath, shorten <=1pt, dashed] (epp2) to node {$\theta'_2 = \theta'$} (ep);

\begin{pgfonlayer}{background}
   \filldraw [fondoTerm]
      (e.north -| e.west) rectangle (ep.south -| ep.east);
\end{pgfonlayer}

\end{tikzpicture}
\end{center}

By the previous proof for one step we have $e \fnrl_{\sigma_1} e''_1$ and $\theta'_1 \in CSubst$ under the conditions stipulated. In order to this with the IH we define the sets ${\cal B}_1 = {\cal B} \cup BV(e_1)$ and ${\cal W}_1 = ({\cal W} \setminus dom(\sigma_1)) \cup \vran(\sigma_1) \cup vE$, where $vE$ is the set of extra variables in the fresh variant $f(\overline{p}) \tor s$ used in $e \fnrl_{\sigma_1} e''_1$, if it was a (Narr) step; or empty otherwise. We also define $\theta_1 = \theta'_1|_{{\cal W}_1}$. Then:
\begin{itemize}
 \item $FV(e''_1) \cup dom(\theta_1) \subseteq {\cal W}_1$: We have $dom(\theta_1) \subseteq {\cal W}_1$ by definition of $\theta_1$. On the other hand we have $FV(e''_1) \subseteq {\cal W}_1$ because given $X \in FV(e''_1)$ we have two possibilities:
 \begin{itemize}
 \item[a)] $X \in FV(e)$): then $X \notin dom(\sigma_1)$ since otherwise it disappears in the step $e \fnrl_{\sigma_1} e''$. As $dom(\theta) \cup FV(e) \subseteq {\cal W}$ then $X \in {\cal W} \setminus dom(\sigma_1)$, so $X \in {\cal W}_1$.
 \item[b)] $X \notin FV(e)$) : then there are two possibilities:
    \begin{itemize}
        \item[i)] $X$ has been inserted by $\sigma_1$, so $X \in \vran(\sigma_1)$ and $X \in {\cal W}_1$.
        \item[ii)] $X$ has been inserted as an extra variable in a (Narr) step. Since the narrowing substitution is a mgu then $\sigma_1$ cannot affect $X$, so $X \in {\cal W}_1$ because $X \in vE$.
    \end{itemize}
\end{itemize}

 \item $e''_1\theta_1 \equiv e_1$: Because as we have seen, $FV(e''_1) \subseteq {\cal W}_1$, and so $e''_1\theta_1 \equiv e''_1\theta'_1|_{{\cal W}_1} \equiv e''_1\theta'_1 \equiv e_1$, by the proof for one step.

 \item $BV(e''_1) \subseteq {\cal B}_1$: As $\theta'_1 \in CSubst$, $e''_1\theta'_1 \equiv e_1$ and no $CSubst$ can introduce any binding then $BV(e_1) = BV(e''_1)$. But ${\cal B}_1 = {\cal B} \cup BV(e_1)$, so $BV(e''_1) = BV(e_1) \subseteq {\cal B}_1$.

 \item $(dom(\theta_1) \cup \vran(\theta_1)) \cap {\cal B}_1 = \emptyset$: As $\theta'_1 \in CSubst$, $e''_1\theta'_1 \equiv e_1$ and no $CSubst$ can introduce any binding then $BV(e_1) = BV(e''_1)$. Then it can happen:
\begin{enumerate}
 \item[a)] $BV(e''_1) \subseteq BV(e)$: Then ${\cal B} =  {\cal B}_1$, as $BV(e_1) = BV(e''_1) \subseteq BV(e) \subseteq {\cal B}$ by hypothesis. Then, as $(dom(\theta'_1) \cup \vran(\theta'_1)) \cap {\cal B} = \emptyset$ by the proof for one step, then $(dom(\theta'_1) \cup \vran(\theta'_1)) \cap {\cal B}_1 = \emptyset$, and so $(dom(\theta_1) \cup \vran(\theta_1)) \cap {\cal B}_1 = \emptyset$, because $\theta_1 = \theta'_1|_{{\cal W}_1}$ and so its domain and variable range is smaller than the domain of $\theta'_1$.
 \item[b)] $BV(e''_1) \supset BV(e)$: Then $e \fnrl_{\sigma_1} e''_1$ must have been a (LetIn) step and so $\sigma = \epsilon$ and $\theta'_1 = \theta$. As the new bounded variable $Z$ is fresh wrt. $\theta$ then it is also fresh for $\theta'_1 = \theta$, and so ${\cal B}_1 = {\cal B} \cup \{Z\}$ has no intersection with $dom(\theta'_1) \cup \vran(\theta'_1)$ nor with $dom(\theta_1) \cup \vran(\theta_1)$, which is smaller.
\end{enumerate}

\item $\sigma_1\theta_1 = \theta[{\cal W}]$: It is enough to see that $\sigma_1\theta_1 = \sigma_1\theta'_1[{\cal W}]$, because we have $\sigma_1\theta'_1 = \theta[{\cal W}]$ by the proof for one step, and this is true because given $X \in {\cal W}$:
\begin{enumerate}
 \item[a)] If $X \in dom(\sigma_1)$ then $FV(X\sigma_1) \subseteq \vran(\sigma_1) \subseteq {\cal W}_1$, so as $\theta_1 = \theta'_1|_{{\cal W}_1}$ then $X\sigma_1\theta_1 \equiv X\sigma_1\theta'_1|_{{\cal W}_1} \equiv X\sigma_1\theta'_1$.
 \item[b)] If $X \in {\cal W} \setminus dom(\sigma_1)$ then $X \in {\cal W}_1$ by definition, and so $X\sigma_1\theta_1 \equiv X\theta_1$ (as $X \not\in dom(\sigma_1)$), $X\theta_1 \equiv X\theta'_1|_{{\cal W}_1} \equiv X\theta'_1$ (as $X \in {\cal W}_1$), and $X\theta'_1 \equiv X\sigma\theta'_1$ (as $X \not\in dom(\sigma_1)$).
\end{enumerate}
\end{itemize}

So we have $e''_1\theta_1 \equiv e_1$ and $e_1 \fe e'$, but then we can apply the induction hypothesis to $e''_1\theta_1 \fe e'$ using ${\cal W}_1$ and ${\cal B}_1$, which fulfill the hypothesis of the lemma, as we have seen. Then we get $e''_1 \fnrl^*_{\sigma_2}~ e''_2$ and $\theta'_2 \in CSubst$ under the conditions stipulated. But then we have:
$$
e \fnrl_{\sigma_1} e''_1 \fnrl^*_{\sigma_2}~ e''_2 \mbox{ taking } e'' \equiv e''_2 \mbox{, } \sigma = \sigma_1\sigma_2 \mbox{ and } \theta' = \theta'_2
$$
for which we can prove the conditions in Lemma \ref{lem:lifting}:
\begin{itemize}
 \item[i)] $e''\theta' \equiv e'$: As $e''\theta' \equiv e''_2\theta'_2 \equiv e'$ by IH.
 \item[ii)] $\sigma\theta' = \theta[{\cal W}]$: That is, $\sigma_1\sigma_2\theta'_2 = \theta[{\cal W}]$. As we have  $\sigma_1\theta_1 = \theta[{\cal W}]$, as we saw before, all that is left is proving $\sigma_1\sigma_2\theta'_2 =\sigma_1\theta_1[{\cal W}]$, which happens because given $X \in {\cal W}$:
\begin{enumerate}
 \item[a)] If $X \in dom(\sigma_1)$ then $FV(X\sigma_1) \subseteq \vran(\sigma_1) \subseteq {\cal W}_1$, so as $\sigma_2\theta'_2 = \theta_1[{\cal W}_1]$ by IH, then $(X\sigma_1)\sigma_2\theta'_2 \equiv (X\sigma_1)\theta_1$.
 \item[b)] If $X \in {\cal W} \setminus dom(\sigma_1)$ then $X \in {\cal W}_1$ by definition, and so, as $\sigma_2\theta'_2 = \theta_1[{\cal W}_1]$ by IH, then $X\sigma_1\sigma_2\theta'_2 \equiv X\sigma_2\theta'_2$ (as $X \not\in dom(\sigma_1)$), $X\sigma_2\theta'_2 \equiv X\theta_1$ (as $X \in {\cal W}_1$), $X\theta_1 \equiv X\sigma_1\theta_1$ (as $X \not\in dom(\sigma_1)$).
\end{enumerate}
 \item[iii)] $(dom(\theta') \cup \vran(\theta')) \cap {\cal B} = \emptyset$: That is $(dom(\theta'_2) \cup \vran(\theta'_2)) \cap {\cal B} = \emptyset$, which happens as $(dom(\theta'_2) \cup \vran(\theta'_2)) \cap {\cal B}_1 = \emptyset$ by IH and ${\cal B} \subseteq {\cal B}_1$.
\end{itemize}

\end{proof}

\subsection{Proofs for Section \ref{letR-classR}}

\nc{\btlr}[1]{\widehat{#1}} 


The let-binding elimination transformation $~\tlr{\_}~$\ satisfies the following interesting properties, which illustrate that its definition is sound.
\begin{lemma}\label{TLemTranRw1}\label{TLemTranRw2}
For all $e,e' \in LExp$, $\con \in Cntxt$, $X \in \var$ we have:
\begin{enumerate}
    \item[i)] $|\tlr{e}| \equiv |e|$.
    \item[ii)] If $e \in Exp$ then $\tlr{e} \equiv e$.
    \item[iii)] $FV(\tlr{e}) \subseteq FV(e)$
    \item[iv)] $\tlr{e[X/e']} = \tlr{e}[X/\tlr{e}']$.
\end{enumerate}
\end{lemma}
\begin{proof}
\begin{enumerate}
    \item[i--iii)] Easily by induction on the structure of $e$.
    \item[iv)] A trivial induction on the structure of $e$, using Lemma \ref{auxBind} for the case when $e$ has the shape $e \equiv let~X=e_1~in~e_2$.
\end{enumerate}
\end{proof}


\teoremi{Lemma \ref{lemNoComp} (Copy lemma)}
For all $e, e_1, e_2 \in Exp$, $X \in \var$:
\begin{enumerate}
    \item[i)]  $e_1 \rw e_2$ implies $e[X/e_1] \rw^* e[X/e_2]$.
    \item[ii)]  $e_1 \rw^* e_2$ implies $e[X/e_1] \rw^* e[X/e_2]$.
\end{enumerate}

\begin{proof}\label{DEMO_lemNoComp}
To prove {\it i)} we proceed by induction on the structure of $e$. Concerning the base cases:
    \begin{itemize}
        \item If $e \equiv X$ then $e[X/e_1] \equiv e_1 \rw e_2 \equiv e[X/e_2]$, by hypothesis.
        \item If $e \equiv Y \in \var \setminus \{X\}$ then $e[X/e_1] \equiv  Y \rw^0 Y \equiv e[X/e_2]$.
        \item Otherwise $e \equiv h$ for some $h \in \Sigma$, so $e[X/e_1] \equiv h \rw^0 h \equiv e[X/e_2]$
    \end{itemize}
Regarding the inductive step, then $e \equiv h(e'_1, \dots, e'_n)$ and so
$$
\begin{array}{ll}
e[X/e_1] \equiv h(e'_1[X/e_1], \dots, e'_n[X/e_1]) \\
\rw^* h(e'_1[X/e_2], \dots, e'_n[X/e_2]) & \mbox{ by IH, $n$ times } \\
\equiv e[X/e_2] \\
\end{array}
$$
The proof for {\it ii)} follows the same structure.
\end{proof}

\teoremi{Lemma \ref{LemLSoundRw} (One-Step Soundness of let-rewriting wrt. term rewriting)}
For all $e, e' \in LExp$ we have that $e \f e'$ implies $\tlr{e} \rw^* \tlr{e}'$.

\begin{proof}\label{DEMO_LemLSoundRw}
We proceed by a case distinction over the rule of let-rewriting used in the step $e \f e'$.
\begin{description}
    \item[\crule{Fapp}] Then we have:
$$
e \equiv f(\overline{p})\theta \f r\theta \equiv e' \mbox{ for some } (f(\overline{p}) \tor r) \in \prog, \theta \in CSubst
$$
But then $f(\overline{p})\theta, r\theta \in Exp$, therefore $\tlr{f(\overline{p})\theta} \equiv f(\overline{p})\theta$ and $\tlr{r\theta} \equiv r\theta$, by Lemma \ref{TLemTranRw1} {\em ii)}, and so we can link $\tlr{e} \equiv \tlr{f(\overline{p})\theta} \equiv f(\overline{p})\theta \rw r\theta \equiv \tlr{r\theta} \equiv \tlr{e'}$, by a term rewriting step.
    \item[\crule{LetIn}] Then we have:
$$
e \equiv h(e_1, \ldots, e_k, \ldots, e_n) \f let~X=e_k~in~h(e_1, \ldots, X, \ldots, e_n) \equiv e'$$
where $X$ is a fresh variable (among other conditions). But then
$$
\begin{array}{ll}
\tlr{e'} \equiv \tlr{h(e_1, \ldots, X, \ldots, e_n)}[X/\tlr{e_k}] \equiv h(\tlr{e_1}, \ldots, X, \ldots, \tlr{e_n})[X/\tlr{e_k}] & \\
\equiv h(\tlr{e_1}, \ldots, \tlr{e_k}, \ldots, \tlr{e_n}) & \mbox{ as $X$ is fresh} \\
\equiv \tlr{h(e_1, \ldots, e_k, \ldots, e_n)} \equiv \tlr{e} & \\
\end{array}
$$
Therefore $\tlr{e} \rw^0 \tlr{e} \equiv \tlr{e'}$.

    \item[\crule{Bind}] Then we have:
$$
e \equiv let~X=t~in~e_1 \f e_1[X/t] \equiv e' \mbox{ with $t \in CTerm$}
$$
But then $\tlr{e} \equiv \tlr{e_1}[X/\tlr{t}] \equiv \tlr{e_1[X/t]} \equiv \tlr{e'}$, by Lemma \ref{TLemTranRw2} {\em iv)}, hence $\tlr{e} \rw^0 \tlr{e} \equiv \tlr{e'}$.

    \item[\crule{Elim}] Then we have:
$$
e \equiv let~X=e_1~in~e_2 \f e_2 \equiv e' \mbox{ with $X \not\in FV(e_2)$}
$$
But then
$$
\begin{array}{ll}
\tlr{e} \equiv \tlr{e_2}[X/\tlr{e_1}] \\
\equiv \tlr{e_2[X/e_1]} & \mbox{ by Lemma \ref{TLemTranRw2}}~iv) \\
\equiv \tlr{e_2} \equiv \tlr{e'} & \mbox{ as $X \not\in FV(e_2)$}
\end{array}
$$
Therefore $\tlr{e} \rw^0 \tlr{e} \equiv \tlr{e'}$.

\item[\crule{Flat}] Then we have:
$$
e \equiv let~X = (let~Y = e_1~in~e_2)~in~e_3 \f let~Y=e_1~in~(let~X = e_2~in~e_3) \equiv e'
$$
where $Y \notin FV(e_3)$. But then
$$
\begin{array}{@{}ll@{}}
\tlr{e} \equiv \tlr{e_3}[X/\tlr{let~Y=e_1~in~e_2}] \equiv \tlr{e_3}[X/(\tlr{e_2}[Y/\tlr{e_1}])] \\
\equiv \tlr{e_3}[X/\tlr{e_2}][Y/\tlr{e_1}] & Y \notin FV(\tlr{e_3}) \mbox{ by Lemma \ref{TLemTranRw1} {\em iii)} }\\
\equiv (\tlr{let~X=e_2~in~e_3})[Y/\tlr{e_1}] \equiv \tlr{e'} & \\
\end{array}
$$
Therefore $\tlr{e} \rw^0 \tlr{e} \equiv \tlr{e'}$.

    \item[\crule{Contx}] Then we have:
$$
e \equiv \con[e_1] \f \con[e_2] \equiv e'
$$
with $e_1 \f e_2$ by some of the previous rules, therefore $\tlr{e_1} \rw^* \tlr{e_2}$ by the proof of the previous cases. We will prove that $\tlr{e_1} \rw^* \tlr{e_2}$ implies $\tlr{\con[e_1]} \rw^* \tlr{\con[e_2]}$, thus getting $\tlr{e} \rw^* \tlr{e'}$ as a trivial consequence.

We proceed by induction on the structure of $\con$. Regarding the base case then $\con \equiv []$ and so $\tlr{\con[e_1]} \equiv \tlr{e_1} \rw^* \tlr{e_2} \equiv \tlr{\con[e_2]}$ by hypothesis. For the inductive step:\begin{itemize}
    \item If $\con \equiv let~X=\con'~in~a$ then by IH we get $\tlr{\con'[e_1]} \rw^* \tlr{\con'[e_2]}$, and so
$$
\begin{array}{ll}
\tlr{\con[e_1]} \equiv \tlr{a}[X/\tlr{\con'[e_1]}] \\
\rw^* \tlr{a}[X/\tlr{\con'[e_2]}] \mbox{ by IH and Lemma \ref{lemNoComp}} \\
\equiv \tlr{\con[e_2]}
\end{array}
$$
Notice that it is precisely because of this case that we cannot say that $e \f e'$ implies $\tlr{e} \rw^* \tlr{e'}$ in zero or one steps, because the copies of $\tlr{\con'[e_1]}$ made by the substitution $[X/\tlr{\con'[e_1]}]$ may force the zero or one steps derivation from $\tlr{\con'[e_1]}$ to be repeated several times in derivation $\tlr{a}[X/\tlr{\con'[e_1]}] \rw^* \tlr{a}[X/\tlr{\con'[e_2]}]$. This is typical situation when mimicking term graph rewriting derivations by term rewriting.

    \item If $\con \equiv let~X=a~in~\con'$ then $\tlr{\con[e_1]} \equiv \tlr{\con'[e_1]}[X/\tlr{a}] \rw^* \tlr{\con'[e_2]}[X/\tlr{a}] \equiv \tlr{\con[e_2]}$, by IH combined with closedness under substitutions of term rewriting.

    \item Otherwise $\con \equiv h(a_1, \ldots, \con', \ldots, a_n)$ and then $\tlr{\con[e_1]} \equiv h(\tlr{a_1}, \ldots, \tlr{\con'[e_1]}, \ldots, \tlr{a_n})$ $ \rw^* h(\tlr{a_1}, \ldots, \tlr{\con'[e_2]}, \ldots, \tlr{a_n}) \equiv \tlr{\con[e_2]}$ by IH.

\end{itemize}
\end{description}
\end{proof}

\teoremi{Proposition \ref{propDenSubstElemsDsord}}
For all $\sigma \in Subst_\perp$, $\theta \in \den{\sigma}$, we have that $\theta \dsord \sigma$.

\begin{proof}\label{DEMO_propDenSubstElemsDsord}
Given some $X \in \var$, we have two possibilities. If $X \in dom(\theta)$ then taking any $t \in CTerm_\perp$ such that $\cl \theta(X) \clto t$, by Lemma \ref{lemmashells} we have $t \ordap \theta(X)$, because $\theta \in \den{\sigma} \subseteq CSubst_\perp$. But $\theta \in \den{\sigma}$ implies $\cl \sigma(X) \clto \theta(X)$, therefore $\cl \sigma(X) \clto t$ by the polarity from Proposition \ref{propCrwlletPolar}, which holds for \crwl\ too. Hence $\den{\theta(X)} \subseteq \den{\sigma(X)}$.

On the other hand, if $X \not\in dom(\theta)$ then for any $t \in CTerm_\perp$ such that $\cl \theta(X) \equiv X \clto t$ we have that $t \equiv \perp$ or $t \equiv X$. If $t \equiv \perp$ then $\cl \sigma(X) \clto t$ by rule \crule{B}. Otherwise $\theta \in \den{\sigma}$ implies $\cl \sigma(X) \clto \theta(X) \equiv X \equiv t$. Hence $\den{\theta(X)} \subseteq \den{\sigma(X)}$.
\end{proof}

\teoremi{Proposition \ref{auxDenSubs2}}
For all $\sigma \in DSusbt_{\perp}$, $\den{\sigma}$ is a directed set.

\begin{proof}\label{DEMO_auxDenSubs2}
For any preorder $\leq$, any directed set $D$ wrt. it and any elements $e_1, e_2 \in D$ by $e_1 \sqcup_D e_2$ we denote the element $e_3 \in D$ such that $e_1 \leq e_3$ and $e_2 \leq e_3$ that must exist because $D$ is directed.

Now, given any $\sigma \in DSubst_\perp$ we have that $\forall X \in \var, \den{\sigma(X)}$ is a directed set, because if $X \in dom(\sigma)$ then we can apply the definition of $DSubst_\perp$ and otherwise $\den{X} = \{X, \perp\}$, which is directed. Now given $\theta_1, \theta_2 \in \den{\sigma}$ we can define $\theta_3 \in CSubst_\perp$ as $\theta_3(X) = \theta_1(X) \sqcup_{\sigma(X)} \theta_2(X)$, which fulfills:
\begin{enumerate}
 \item $\theta_i \ordap \theta_3$ for $i \in \{1,2\}$, because for any $X \in \var$ we have that $\den{\sigma(X)}$ is directed (as we saw above) and $\theta_i(X) \in \den{\sigma(X)}$ (because $\theta_1, \theta_2 \in \den{\sigma}$), therefore $\theta_i(X) \ordap \theta_1(X) \sqcup_{\sigma(X)} \theta_2(X) = \theta_3(X)$ by definition.
 \item $\theta_3 \in \den{\sigma}$, because $\forall X \in \var, \theta_3(X) = \theta_1(X) \sqcup_{\sigma(X)} \theta_2(X) \in \den{\sigma(X)}$ by definition.
\end{enumerate}
\end{proof}

We will use the following lemma about non-triviality of substitution denotations as an auxiliary result for proving Lemma \ref{auxDenSubs3}.
\begin{lemma}\label{lemInfDomDenSubst}
For all $\sigma \in Subst_\perp$ we have that $\den{\sigma} \not= \emptyset$ and given $\overline{X} = dom(\sigma)$ then $[\overline{X/\perp}] \in \den{\sigma}$.
\end{lemma}
\begin{proof}
It is enough to prove that if $\overline{X} = dom(\sigma)$ then $[\overline{X/\perp}] \in \den{\sigma}$. First of all $[\overline{X/\perp}] \in CSubst_\perp$ by definition. Now consider some $Y \in \var$.
\begin{enumerate}
    \item[i)] If $Y \in \overline{X}$ then $\sigma(Y) \clto \perp \equiv Y[\overline{X/\perp}]$, by rule \crule{B}.
    \item[ii)] Otherwise $Y \not\in \overline{X} = dom(\sigma)$, hence $\sigma(Y) \equiv Y \clto Y \equiv Y[\overline{X/\perp}]$, by rule \crule{RR}.
\end{enumerate}
\end{proof}

\teoremi{Lemma \ref{auxDenSubs3}}
For all $\sigma \in DSusbt_{\perp}$, $e \in Exp_{\perp}, t \in CTerm_{\perp}$,
$$
\mbox{if } e\sigma \clto t \mbox{ then } \exists \theta \in \den{\sigma} \mbox{ such that } e\theta \clto t
$$

\begin{proof}\label{DEMO_auxDenSubs3}
We proceed by a case distinction over $e$:
\begin{itemize}
 \item If $e \equiv X \in dom(\sigma)$ : Then $e\sigma \equiv \sigma(X) \clto t$, so we can define:
$$
\theta(Y) = \left\{\begin{array}{ll}
t & \mbox{ if } Y \equiv X \\
\perp & \mbox{ if } Y \in dom(\sigma) \setminus \{X\}\\
Y & \mbox{ otherwise}
\end{array} \right.
$$
Then $\theta \in \den{\sigma}$ because obviously $\theta \in CSusbt_\perp$, and given $Z \in \var$.
\begin{enumerate}
    \item[a)] If $Z \equiv X$ then $\sigma(Z) \equiv \sigma(X) \clto t \equiv \theta(Z)$ by hypothesis.
    \item[b)] If $Z \in (dom(\sigma) \setminus \{X\})$ then $\sigma(Z) \clto \perp \equiv \theta(Z)$ by rule \crule{B}.
     \item[c)] Otherwise $Z \not\in dom(\sigma)$ and then $\sigma(Z) \equiv Z \clto Z \equiv \theta(Z)$ by rule \crule{RR}.
\end{enumerate}
But then $e\theta \equiv \theta(X) \equiv t \clto t$ by Lemma \ref{lemmashells}---which also holds for \crwl, because \crwl\ and \crwll\ coincide for c-terms--- , as $t \in CTerm_\perp$.

 \item If $e \equiv X \not\in dom(\sigma)$ : Then given $\overline{Y} = dom(\sigma)$ 
we have $[\overline{Y/\perp}] \in \den{\sigma}$ by Lemma \ref{lemInfDomDenSubst}, so we can take $\theta = \{[\overline{Y/\perp}]\}$ for which $\den{e\sigma} = \den{X\sigma} = \den{X} = \den{X[\overline{Y/\perp}]}= \den{X\theta}$.
 \item If $e \not\in \var$ then we proceed by induction over the structure of $e\sigma \clto t$:
\begin{description}
 \item[Base cases]~
 \begin{description}
  \item[\clrule{B}] Then $t \equiv \perp$, so given $\overline{Y} = dom(\sigma)$ we can take $\theta = \{[\overline{Y/\perp}]\}$ for which $e\theta \clto \perp$ by rule \crule{B}.
  \item[\clrule{RR}] Then $e \in \var$ and we are in the previous case.
  \item[\clrule{DC}] Similar to the case for $e \equiv X \not\in dom(\sigma)$.
 \end{description}
 \item[Inductive steps]~
  \begin{description}
   \item[\clrule{DC}] Then $e \equiv c(e_1, \ldots, e_n)$, as $e \not\in \var$, and we have:
$$
\infer[DC]{e\sigma \equiv c(e_1\sigma, \ldots, e_n\sigma) \clto c(t_1, \ldots, t_n) \equiv t}
          {
       e_1\sigma \clto t_1
    \ \ldots
    \  e_n\sigma \clto t_n
          }
$$
Then by IH or the proof of the other cases we have that $\forall i \in \{1, \ldots, n\}$. $\exists \theta_i \in \den{\sigma}$ such that $e_i\theta_i \clto t_i$. But as $\sigma \in DSusbt_{\perp}$ then $\den{\sigma}$ is directed by Lemma \ref{auxDenSubs2},
therefore there must exist some $\theta \in \den{\sigma}$ such that $\forall i \in \{1, \ldots, n\}. \theta_i \sqsubseteq \theta$, and so by Proposition \ref{PropMonSubstCrwlLet} ---which also holds for \crwl, by Theorem \ref{thEquivCrwlCrwllet}--- we have $\forall i \in \{1, \ldots, n\}. e_i\theta \clto t_i$, so we can build the following proof:
$$
\infer[DC]{e\theta \equiv c(e_1\theta, \ldots, e_n\theta) \clto c(t_1, \ldots, t_n) \equiv t}
          {
            e_1\theta \clto t_1
       \ \ldots
            e_n\theta \clto t_n
           }
$$
   \item[\clrule{OR}] Very similar to the proof of the previous case. We also have $e \equiv f(e_1, \ldots, e_n)$ (as $e \not\in \var$) and given a proof for $e\sigma \equiv f(e_1, \ldots, e_n)\sigma \clto t$, so we can apply the IH or the proof of the other cases to every $e_i\sigma \clto p_i\mu$ to get some $\theta_{i} \in \den{\sigma}$ such that $e_i\theta_i \clto p_i\mu$. Then we can use Lemma \ref{auxDenSubs2} and Proposition \ref{PropMonSubstCrwlLet} to use the obtained $\theta$ to compute the same values for the arguments of $f$, thus using the same substitution $\mu \in CSubst_\perp$ for parameter passing in \crule{OR}.
 \end{description}
\end{description}
\end{itemize}
\end{proof}

\teoremi{Theorem \ref{TAdCrwlRw}}
Let $\prog$ be a \crwl-deterministic program, and  $e,e'\in Exp, t \in CTerm$. Then:
\begin{enumerate}
    \item[a)] $e \rw^* e'$ implies $e \fe e''$ for some $e'' \in LExp$ with $|e''| \sqsupseteq |e'|$.
    \item[b)] $e \rw^* t$\ iff\ $e \fe t$\ iff\ $\cl e \clto t$.
\end{enumerate}
\begin{proof}

\begin{enumerate}
\item[a)] Assume $e \rw^* e'$. By Lemma \ref{LCompCdRw}, $\den{e'} \subseteq \den{e}$ and by Lemma \ref{lemmashells} we have $|e'| \in \den{e'}$, then $|e'| \in \den{e}$. 
Therefore, by Theorem \ref{TheorEquivLetRwDown} there exists
$e'' \in LExp$ such that $e \fe e''$ with $|e''| \sqsupseteq |e'|$.


\item[b)] The parts $e \fe t$ iff $\cl e \clto t$, and $e \fe t$ implies $e \rw^* t$ have been already proved for arbitrary programs in Theorems \ref{TheorEquivLetRwDown} and  \ref{TLSoundRw} respectively. What remains to be proved is that
$e \rw^* t$ implies $e \fe t$ (or the equivalent $\cl e \clto t$).
Assume $e \rw^* t$. Then $\den{t} \subseteq \den{e}$ by Lemma \ref{LCompCdRw}.
Now, by Lemma \ref{lemmashells} $t\in \den{t}$, and therefore $t\in \den{e}$, which exactly
means that $\cl e \clto t$.
\end{enumerate}
\end{proof}

\end{document}